\documentclass[11pt]{article}
\usepackage[T1]{fontenc}
\usepackage[utf8]{inputenc}
\usepackage{lmodern}
\usepackage[a4paper,margin=26mm]{geometry}
\usepackage{amsmath,amssymb,amsthm,microtype}
\usepackage{graphicx,array}
\usepackage[hidelinks]{hyperref}
\numberwithin{equation}{section}
\newtheorem{theorem}{Theorem}[section]
\newtheorem{proposition}[theorem]{Proposition}
\newtheorem{lemma}[theorem]{Lemma}
\newtheorem{corollary}[theorem]{Corollary}
\theoremstyle{definition}

\newtheorem{example}[theorem]{Example}
\newtheorem{conjecture}[theorem]{Conjecture}
\theoremstyle{remark}
\newtheorem{remark}[theorem]{Remark}
\hypersetup{pdftitle={Symmetric amplituhedra: BCFW recursion working draft}}
\providecommand{\Gr}{\operatorname{Gr}}
\providecommand{\OG}{\operatorname{OG}}
\providecommand{\LG}{\operatorname{LG}}
\providecommand{\GL}{\operatorname{GL}}
\providecommand{\Mat}{\operatorname{Mat}}
\providecommand{\pre}{\operatorname{pre}}
\providecommand{\inc}{\operatorname{inc}}
\providecommand{\IP}{\operatorname{IP}}
\providecommand{\refl}{\mathrm{refl}}
\providecommand{\MomAmp}{\mathcal M}
\providecommand{\SymAmp}[2]{\MomAmp^{#1,#2}}
\providecommand{\RefAmp}{\MomAmp^{\refl}}
\providecommand{\FraserPos}{\mathcal F^{\geq}}

\providecommand{\SymGr}[3]{\Gr^{#1,#2}(#3)}
\providecommand{\SymGrNonneg}[3]{\Gr^{#1,#2}_{\geq}(#3)}
\providecommand{\SymGrPos}[3]{\Gr_{>}^{#1,#2}(#3)}
\providecommand{\Hom}{\operatorname{Hom}}
\providecommand{\End}{\operatorname{End}}

\providecommand{\einc}{\operatorname{einc}}

\providecommand{\LinRefDomain}{\mathcal R}
\providecommand{\ExRefDomain}{\mathcal E}
\providecommand{\OrthRefDomain}{\mathcal O}

\providecommand{\SymAmpA}[2]{\MomAmp^{#1,#2;\mathrm A}}
\providecommand{\SymAmpB}[2]{\MomAmp^{#1,#2;\mathrm B}}
\providecommand{\RotAmp}{\MomAmp^{\mathrm{rot}}}

\title{{Symmetric Amplituhedra I: Spinor Helicity}}
\author{Ran J. Tessler\thanks{Department of Mathematics, Weizmann Institute of Science, Rehovot, Israel. \texttt{ran.tessler@weizmann.ac.il}}}

\begin{document}
\maketitle

\begin{abstract}
For every pair of integers $0\leq k\leq N,$ and and every subgroup of $D_N\times\mathbb{Z}_2$, we construct the nonnegative symmetric Grassmannian, which is the subspace of the nonnegative Grassmannian fixed under that group, where the dihedral part acts on columns (with a sign correction to keep nonnegativity) and the $\mathbb{Z}_2$ factor exchanges a space with its orthocomplement (again, with a sign correction). This construction generalizes the earlier constructions of \cite{Karpman2018,Fraser2020,Shevchenko2025}. We study the geometry of this space. Then, when $k$ is further restricted to the range $[2,N-2]$ we also define a symmetric amplituhedron. We analyze its basic properties and its conjectural BCFW decomposition. The ABJM and reflected Lagrangian amplituhedra are the special cases corresponding to the groups $\{1\}\times\mathbb{Z}_2,$ and the group generated by a reflection combined with the exchange, respectively.
\end{abstract}

\tableofcontents
\section{Introduction}
\label{sec:introduction}

The amplituhedron, introduced by Arkani-Hamed and Trnka, is a geometric
space whose volume-like \emph{canonical form} conjecturally encodes tree-level scattering amplitudes in
planar $\mathcal N=4$ super Yang--Mills theory~\cite{ArkaniHamedTrnka2014}.
It is obtained by mapping the nonnegative Grassmannian using positive
external data to a target Grassmannian. Rather than organizing an amplitude as a sum of Feynman
diagrams, this construction suggest a more efficient calculation, in which the amplitude is computed from the sum of canonical forms of recursively defined subspaces which tile the amplituhedron.
It brings the previously observed role of the nonnegative Grassmannian
and the BCFW recursion under a common geometric roof
\cite{BrittoCachazoFeng2005,BrittoCachazoFengWitten2005,AHBCGPT2016}.
The amplituhedron motivated the definition of positive geometries~\cite{ArkaniHamedBaiLam2017}, but additional exceptional properties it possesses, like the BCFW decomposition~\cite{EvenZoharLakrecTessler2025}, make it a rather special member of this family.

Up to date there are only few known amplituhedron-like spaces. Besides the original
amplituhedron, these include its $m\neq 4$ variants, particularly
$m=1$ and $m=2$~\cite{KarpWilliams2019,BaoHe2019,ParisiShermanBennettWilliams2023},
and the momentum amplituhedron, which expresses the geometry in
spinor-helicity variables~\cite{DamgaardEtAl2019}.
The orthogonal or ABJM amplituhedron provides an analogous construction for
three-dimensional ABJM theory, using a nonnegative orthogonal
Grassmannian rather than the ordinary nonnegative Grassmannian
\cite{HuangEtAl2021,HeKuoZhang2021}.
The reflected Lagrangian amplituhedron~\cite{KroviTessler2026}, built on
Karpman's nonnegative Lagrangian Grassmannian~\cite{Karpman2018}, adds a
symplectic example. These geometries share a common architecture:
positive domains and external data, combinatorial independence of the external data, kinematic images, natural
factorization patterns coming from a plabic tangle operad \cite{PlabicTangles}, and recursive tiling constructions or proposals.
It is natural to ask whether these are isolated
constructions or instances of a systematic family.

In this work we introduce a uniform construction, which we name \emph{symmetric amplituhedra}, associated with a symmetry datum
$\rho:G\hookrightarrow D_N\times\mathbb Z_2$.
The dihedral factor permutes the cyclically ordered particles; the $\mathbb{Z}_2$
factor exchanges the two spinor sectors. Interestingly, such group actions lift naturally to the nonnegative Grassmannian, and through it to amplituhedra, and preserve many of the beautiful properties of these spaces.

Write $\rho(\gamma)=(g_\gamma,\epsilon(\gamma))$ for $g_\gamma\in D_N,\epsilon(\gamma)\in\mathbb Z_2$, and let
$T_\gamma$ be the signed orthogonal label map, including the alternating
signs required by exchange, as specified in Section~\ref{sec:symmetries}.
The actions on kinematics, external data, and domain planes are
\begin{equation*}
\begin{array}{c@{\qquad}c@{\qquad}c}
 &\epsilon(\gamma)=0&\epsilon(\gamma)=1\\[3pt]
 \gamma\cdot(\lambda,\widetilde\lambda):
   &(T_\gamma\lambda,T_\gamma\widetilde\lambda)
   &(T_\gamma\widetilde\lambda,T_\gamma\lambda)\\[3pt]
 \gamma\cdot(\Lambda,\widetilde\Lambda):
   &(T_\gamma\Lambda,T_\gamma\widetilde\Lambda)
   &(T_\gamma\widetilde\Lambda,T_\gamma\Lambda)\\[3pt]
 \gamma\cdot C:&T_\gamma C&T_\gamma(C^\perp).
\end{array}
\end{equation*}
Here all entries denote subspaces: the spinors are row planes, the
external matrices denote their column spaces, and $\perp$ is Euclidean
orthogonal complementation. We define the symmetric domain Grassmannian and kinematic
space as the respective loci fixed by $G$, and require the external
pair to be fixed as well. An exchanging element sends a $k$-plane to an
$(N-k)$-plane, so a fixed source requires $k=N-k$. Thus throughout the
paper an action with nontrivial exchange is restricted to even particle number $N=2k$. We denote the domain fixed locus inside the nonnegative Grassmannian by $\Gr_{\geq}^{G,\rho}(k,N)$.
We choose positive equivariant external subspaces satisfying
$\Lambda^\perp\subset\widetilde\Lambda$, and apply the momentum amplituhedron map to the
$G-$fixed nonnegative domain. Its image, denoted $\SymAmp{G}{\rho}_{N,k}(\Lambda,\tilde\Lambda)$,
lies in the corresponding $G-$fixed kinematic space. 

We propose a uniform operadic ansatz for codimension $1$ boundary facets of symmetric amplituhedra: the facets are organized in families, which are the vanishing loci (inside the amplituhedra) of Mandelstam variables $s_I$ for cyclic intervals $I$ coming in two flavors. Up
appropriate end cases, $I$ is either disjoint from its $G-$orbit, or is fixed by a reflection, but disjoint from its orbit under the cyclic shift subgroup of $G$, see Section~\ref{sec:boundaries} for a more accurate description. The corresponding Mandelstam variable may become reducible, because of the symmetry, and this reducibility adds interesting new phenomena. The combinatorial boundary structure is expected to be independent of the external data.

The domain cells mapping to these boundaries admit factorizations in terms of \emph{symmetric plabic tangles}: a domain plane is assembled
from a $(G,\rho)-$symmetric core and symmetry-related blobs, either attached by single
edges or disconnected from the core; see Figure~\ref{fig:schematic_bdries}.
These pictures represent one-particle fiber products and direct sums,
respectively, with soft limits and degenerate cores as end cases.
Only finitely many blob types occur: their local stabilizers lie in
$\mathbb Z_2\times\mathbb Z_2$, generated by a reflection and an exchange. These can be thought of as the universal building blocks of this family, some of which are known spaces: the case of a trivial local stabilizer gives a momentum amplituhedron blob, an exchange stabilizer yields a ABJM blob, and a reflection composed with an exchange gives a reflected Lagrangian blob.

Our main conjecture, Conjecture~\ref{conj:bcfw:tiling}, asserts that these amplituhedra are tiled by images of certain symmetric positroid cells, calculated via a symmetric BCFW procedure. Interestingly, in the symmetric setting, symmetric BCFW shifts can produce high-degree bridge removal equations. Despite the high degree, all our experiments support our conjecture that the symmetric BCFW positroid cells still map injectively onto their image, and form tilings.  

Even though the primary contribution of this paper is the geometric construction, and the proposed operadic description of the codimension $1$ boundaries of the symmetric amplituhedra, we also prove several structural results:
For the domain nonnegative symmetric Grassmannians we establish closed-ball topology, compute dimensions and representation types, and provide equivariant cell decompositions to the symmetric positroid strata strata
(Theorems~\ref{thm:dom:ball} and~\ref{thm:dom:cells}).
This extends the orthogonal cell theory~\cite{HuangWen2014,KimLee2014}
and the reflected, cyclic and rotational constructions of
Karpman, Fraser and Shevchenko~\cite{Karpman2018,Fraser2020,Shevchenko2025}.
We define the resulting amplituhedra in both projected-plane and kinematic
pictures, and determine their dimensions and spinor representations.
The momentum, ABJM and reflected Lagrangian maps become specializations
of this construction, alongside with many new examples. Regarding the Mandelstam boundaries, we prove the vanishing of the relevant variables on the images of the corresponding domain cells, and establish image-dimension statements.

While this write-up does not discuss much the  positive geometry perspective, we do expect that both the constructed domain spaces and the symmetric amplituhedra, are positive geometries with canonical forms calculated via the BCFW procedure. This expectation has been tested for several instances, in low dimensions.

\paragraph{The sequels.}
This paper focuses on the spinor helicity side of the symmetric amplituhedra, and their domain spaces. The sequels will treat the symmetric momentum twistor amplituhedron construction, discuss canonical forms, and include proofs of special cases of the BCFW tiling conjecture, Conjecture~\ref{conj:bcfw:tiling}, under the simplifying assumption of immanant-positivity used in \cite{Galashin2024,OrenPerlsteinTessler2025} in their proofs of the analogous BCFW tiling conjectures.
\paragraph{Content of the paper.}
This paper is structured as follows. Section~\ref{sec:background} describes the previously known construction of amplituhedra and their factorizations. It also reviews two additional known constructions of symmetric nonnegative Grassmannians. Section~\ref{sec:symmetries} describes the different symmetry actions we will focus on, fixes conventions and terminology. Section~\ref{sec:domains} constructs the symmetric nonnegative Grassmannians, and studies their basic properties. Section~\ref{sec:external-and-map} defines the symmetric amplituhedra, and analyzes their basic features. The conjectural facet classification of the symmetric amplituhedra, which is the main input for the BCFW tiling is studied in \ref{sec:boundaries}. The BCFW procedure itself is addressed in Section~\ref{sec:bcfw}. The appendices include tabular information (per-action dimension of the symmetric Grassmannian, per-action and shift type degree of the BCFW shift, summary of some experiments) and more technical details and proofs (the form of the Mandelstam variables in symmetric cases; how to select the correct factor of the Mandelstam variable whenever it is reducible). 
\paragraph{AI disclosure.}
The construction, theorem statements, conjectures and detailed proof strategies were all conceived by the author. The author have used OpenAI ChatGPT6 Astra extensively for writing and coding. It was used to
convert the author's detailed proof sketches into fully written proofs, that were verified by the author, and iteratively imrpoved by the author and ChatGPT6. It was also used to write Python programs for testing the author's mathematical predictions. These verifications and experiments have uncovered a small number of edge cases that required the author to perform minor adjustments to hypotheses or arguments. The author is solely responsible for all results, and any remaining errors.

\paragraph{Acknowledgments.}
R.T. wants to that Michael Oren-Perlstein and Syl Tom Krovi for stimulating discussions related to the paper. R.T. was supported by the ISF (grant No. 1729/23) and by the Minerva foundation with funding from the Federal German Ministry for Education and Research. 
\section{Background: Amplituhedra and symmetric Grassmannians}
\label{sec:background}

We recall the positive domains, their elementary constructions, and the
kinematic form of the previously known amplituhedron maps, as well as their boundary factorizations and BCFW tilings. We then describe two models of symmetric nonnegative Grassmannians by Fraser and Shevchenko~\cite{Fraser2020,Shevchenko2025}. Our work generalizes all these constructions.

\subsection{The nonnegative Grassmannian and its elementary operations}
\label{subsec:bg:grassmannian}

A point of $\Gr(k,N)$ is represented by a full-rank $k\times N$ row matrix
$C$, modulo left multiplication by $\GL_k$. Its ordered maximal minors
$\Delta_I(C)$, $I\in\binom{[N]}k$, are homogeneous coordinates. We use
\begin{equation}
\label{eq:bg:positive-grassmannian}
 \Gr_{\geq}(k,N)=\{C:\Delta_I(C)\geq0\text{ for all }I\},
 \qquad
 \Gr_{>}(k,N)=\{C:\Delta_I(C)>0\text{ for all }I\},
\end{equation}
where a representative with the indicated signs is required. The dimension
is $k(N-k)$. Prescribing which minors vanish gives the \emph{positroid
cells}; these admit equivalent descriptions by bounded affine permutations,
Grassmann necklaces, and reduced plabic graphs. Postnikov's boundary measurements of
positive graph weights provide parameterizations for each cell~\cite{Postnikov2006,AHBCGPT2016}.

We call the following standard operations \emph{atomic operations}:
\begin{equation}
\label{eq:bg:ordinary-atoms}
 \begin{aligned}
 \pre_i &: \Gr_{\geq}(k,N)\longrightarrow\Gr_{\geq}(k,N+1)
       &&\text{(insert a zero column)},\\
 \inc_i &: \Gr_{\geq}(k,N)\longrightarrow\Gr_{\geq}(k+1,N+1)
       &&\text{(insert a coloop)},\\
 x_i(t)&=I+t e_{i,i+1},\qquad
 y_i(t)=(x_i(t))^T,
       &&t>0,\quad i<N.
 \end{aligned}
\end{equation}
Here $e_{ab}$ is a matrix unit; the bridges act by $C\mapsto Cx_i(t)$ or $C\mapsto Cy_i(t)$.
A coloop is an inserted coordinate line; the signs at an arbitrary insertion
position are fixed by ordered-minor positivity. Graphically, $\pre$ and
$\inc$ add black and white lollipops, respectively. An admissible reduced
bridge adds one positive parameter. Every positroid cell can be obtained
from loop--coloop seeds by such bridges, allowing cyclic relabelling and
bridges across intervening lollipops~\cite{AHBCGPT2016}.

For reference, we also introduce the signed cyclic shift
and fix the reversal and duality conventions by
\begin{subequations}
\label{eq:sym:positive-lifts}
\begin{align}
 S_k e_i&=e_{i+1}\ (i<N),&
 S_k e_N&=(-1)^{k-1}e_1,& Je_i&=e_{N+1-i},
 \label{eq:sym:dihedral-lifts}\\
 A_N&=\operatorname{diag}(1,-1,\ldots,(-1)^{N-1}),&&
 \mathcal D(C)&=A_N(C^\perp),
 \label{eq:sym:positive-duality}\\
 S_k^N&=(-1)^{k-1}I,& J^2&=I,&JS_kJ&=S_k^{-1}.
 \label{eq:sym:linear-relations}
\end{align}
\end{subequations}
The complement is Euclidean; reversal changes all ordered maximal minors
by one common sign. These maps preserve nonnegativity, with
$\mathcal D:\Gr_{\geq}(k,N)\to\Gr_{\geq}(N-k,N)$; its minors are the
complementary minors of $C$, up to a common scalar~\cite[Lemma~3.3(ii)]{KarpWilliams2019}.
The cyclic bridge at $N\mid1$ uses $(-1)^{k-1}t e_{N,1}$.
Operators on the ambient vector space are written on column vectors;
thus $T(C)$ has row representative $CT^T$.

\subsection{Amplituhedra and the kinematic picture}
\label{subsec:bg:amplituhedra}

For $k+m\leq N$ and a positive $N\times(k+m)$ matrix $Z$, the tree
amplituhedron and its
$B$-picture are
\begin{equation}
\label{eq:bg:tree-and-B}
 \mathcal A_{N,k,m}(Z)=\{CZ:C\in\Gr_{\geq}(k,N)\},
 \qquad
 \mathcal B_{N,k,m}(Z)=\{C^\perp\cap\operatorname{colspan}Z:
                         C\in\Gr_{\geq}(k,N)\}.
\end{equation}
The first lies in $\Gr(k,k+m)$ and the second in
$\Gr(m,\operatorname{colspan}Z)\hookrightarrow\Gr(m,N)$. They are equivalent realizations:
$Y$ corresponds to $Z(\ker Y)$, viewed as an $m$-plane in $\mathbb R^N$
~\cite{ArkaniHamedTrnka2014,KarpWilliams2019}.
In compatible frames, its $m$-brackets are proportional to
$\det(Y,Z_{i_1},\ldots,Z_{i_m})$.
The $m=4$ construction encodes planar $\mathcal N=4$ super Yang--Mills
scattering; the $B$-picture expresses its geometry through projected
kinematic data~\cite{ArkaniHamedTrnka2014,ArkaniHamedThomasTrnka2018}.

\paragraph{Momentum amplituhedra.}
Fix $2\leq k\leq N-2$. Following~\cite{DamgaardEtAl2019}, take
external matrices $\Lambda,\widetilde\Lambda$ with
$\Lambda^\perp$ and $\widetilde\Lambda$ positive, respectively. Put
\begin{subequations}
\label{eq:bg:momentum-map}
\begin{align}
 &\begin{gathered}
 \Lambda\in\Mat_{N,N-k+2},\qquad
 \widetilde\Lambda\in\Mat_{N,k+2},\\
 \mathcal W=\operatorname{colspan}\widetilde\Lambda\in\Gr_{>}(k+2,N),\qquad
 \mathcal U=(\operatorname{colspan}\Lambda)^\perp
           \in\Gr_{>}(k-2,N),
 \end{gathered}
 \label{eq:bg:external-data}\\
 &\Phi_{\Lambda,\widetilde\Lambda}(C)
       =(Y,\widetilde Y)=(C^\perp\Lambda,C\widetilde\Lambda),\qquad
 \MomAmp_{N,k}(\Lambda,\widetilde\Lambda)
       =\Phi_{\Lambda,\widetilde\Lambda}(\Gr_{\geq}(k,N)),
 \label{eq:bg:momentum-Y}\\
 &\lambda=C\cap\mathcal U^\perp,\qquad
 \widetilde\lambda=C^\perp\cap\mathcal W,\qquad
 \lambda\widetilde\lambda^T=0.
 \label{eq:bg:spinor-incidence}
\end{align}
\end{subequations}
The two intersections are two-planes, and identify the image with a region
in the fixed-external kinematic space
\begin{equation}
\label{eq:bg:momentum-kinematics}
 \mathcal K_{\mathcal U,\mathcal W}
 =\{(L,T)\in\Gr(2,\mathcal U^\perp)\times\Gr(2,\mathcal W):L\perp T\}.
\end{equation}
Its regular dimension, and the image dimension, is $2N-4$
~\cite{DamgaardEtAl2019,Galashin2024}. Row-basis changes are already
quotiented out; particlewise rescalings are \emph{not} additionally
quotiented out. It will be convenient, in what follows, to add the nested condition $\mathcal U\subset\mathcal W$. It is
an extra external-data specialization, not part of 
\eqref{eq:bg:external-data}.

Choose two-row spinor representatives and define
\begin{subequations}
\label{eq:bg:kinematic-invariants}
\begin{align}
 p_i&=\lambda_i\widetilde\lambda_i^T,&
 Q_I&=\sum_{i\in I}p_i,&Q_{[N]}&=0,
 \label{eq:bg:momentum-sums}\\
 \langle ij\rangle&=\det(\lambda_i,\lambda_j),&
 [ij]&=\det(\widetilde\lambda_i,\widetilde\lambda_j),&
 s_I&=\det Q_I=\sum_{\substack{i<j\\i,j\in I}}\langle ij\rangle[ij].
 \label{eq:bg:mandelstams}
\end{align}
\end{subequations}
For $p_i\ne0$, the factorization $p_i=\lambda_i\widetilde\lambda_i^T$
is unique up to the little-group rescaling
$(\lambda_i,\widetilde\lambda_i)\mapsto
(t_i\lambda_i,t_i^{-1}\widetilde\lambda_i)$, $t_i\in\mathbb R^\times$.
This leaves the momenta unchanged, but does not define an additional
quotient of the fixed-external kinematic space used here. Momentum conservation gives $s_I=s_{I^c}$. Homogeneous expressions in these brackets will be
called \emph{functionaries}~\cite{EvenZoharLakrecTessler2025}; their zero
loci are independent of the chosen frames. Planar channels are cyclic
intervals $I$. A two-particle channel has the two chiral factors
$\langle i,i+1\rangle$ and $[i,i+1]$~\cite{DamgaardEtAl2019}. The importance of the Mandelstam variables $s_I$ is that their vanishing indicates that the internal momentum $Q_I$ becomes massless. It is conjectured~\cite{DamgaardEtAl2019} that the codimension $1$ boundaries of the momentum amplituhedron are attained precisely when $s_I=0$ for a cyclic interval $I.$ This conjecture is proven under stronger positivity assumptions on the external data by Galashin in \cite{Galashin2024}.

\subsection{Factorization, bridges and BCFW tilings}
\label{subsec:bg:factorization}
Even without assuming the stronger positivity, one can describe the biggest cells of the domain Grassmannian on which $s_I$ vanishes, for a cyclic interval $I.$ For a generic point with $s_I=0$ and $Q_I\ne0$, introduce an
internal particle with $p_*=-Q_I$. The kinematics splits to subsystems supported on label sets
$I\sqcup\{*\}$ and $\{*\}\sqcup I^c$, with opposite internal momenta.
The corresponding source planes are joined by \emph{one-particle amalgamation}
~\cite{AHBCGPT2016,Galashin2024}:
\begin{subequations}
\label{eq:bg:amalgamation}
\begin{align}
 C_L\star_* C_R
   &=\{(x_I,x_{I^c}):\exists u,\ (x_I,u)\in C_L,
                                  \ (u,x_{I^c})\in C_R\},
 \label{eq:bg:fiber-product}\\
 N_L+N_R&=N+2,\qquad k_L+k_R=k+1,\qquad
 \dim\Pi_{\mathrm{fac}}=\dim\Pi_L+\dim\Pi_R-1.
 \label{eq:bg:ordinary-factorization-counts}
\end{align}
\end{subequations}
The cyclic gluing signs are understood. The subtraction by one is the
common internal rescaling; for two ordinary BCFW factors this gives
$\dim\Pi_{\mathrm{fac}}=2N-5$.
Graphically, the two plabic graphs are joined along their $*-$labelled internal legs.
This is different from $Q_I=0$, where a source may split as a direct sum
with no nonzero internal particle. Three-particle rank-one and corank-one
factors supply the adjacent chiral end cases, rather than nondegenerate
two-spinor systems~\cite{AHBCGPT2016}.

Moving a bridge from the domain to the external data gives a useful common
identity. If $C'=C_0M(t)$, the unbridged presentation of the same target is
\begin{equation}
\label{eq:bg:bridge-transfer}
 \begin{gathered}
 (C';\Lambda,\widetilde\Lambda)\rightsquigarrow
 (C_0;M(t)^{-T}\Lambda,M(t)\widetilde\Lambda),\\
 \lambda(t)=\lambda M(t)^{-1},\qquad
 \widetilde\lambda(t)=\widetilde\lambda M(t)^T.
 \end{gathered}
\end{equation}
For $M=x_i$ this transfers a rank-one matrix $tq$ from
$p_{i+1}$ to $p_i$, where $q=\lambda_i\widetilde\lambda_{i+1}^T$.
Consequently $s_I(t)$ is affine-linear. Bridge removal solves the
factorization equation in these \emph{shifted} data
~\cite{EvenZoharLakrecTessler2025,KroviTessler2026}.

A \emph{tile} is the injective image of an open domain cell of image
dimension. A collection of tiles is a \emph{tiling} when the open images are
pairwise disjoint and their closures cover the amplituhedron. BCFW cells
are constructed by iterating lower-point factorizations and adjoining the
normal bridge. The $m=4$ tree BCFW tilings are proved in
\cite{EvenZoharLakrecTessler2025,EvenZoharEtAl2025}; momentum and ABJM
tilings are proved under strengthened positivity hypotheses in
\cite{Galashin2024,OrenPerlsteinTessler2025}.
These hypotheses require positivity of the relevant nonzero coefficients
in Temperley--Lieb immanant expansions, schematically
\begin{equation}
\label{eq:bg:immanant-positivity}
 F(C;\mathrm{ext})=\sum_\tau c^F_\tau(\mathrm{ext})\,\mathcal I_\tau(C),
 \qquad \mathcal I_\tau(C)\geq0,\qquad c^F_\tau(\mathrm{ext})>0,
\end{equation}
with the conventional signs of $F$ fixed. They are stronger than positivity
of the maximal external minors. 

Amplituhedra are conjectured to be positive geometries in the sense of~\cite{ArkaniHamedBaiLam2017}. This  proposes that the logarithmic forms of the tiles in a tiling
push forward and add to the canonical form:
\begin{equation}
\label{eq:bg:canonical-sum}
 \Omega_{\mathcal A}=\sum_{\Pi\in\mathfrak T}\Phi_*\omega_\Pi,
 \qquad
 \omega_\Pi=\bigwedge_a d\log t_a
 \quad\text{in ordinary positive bridge coordinates}.
\end{equation}

\subsection{The orthogonal Grassmannian and the ABJM amplituhedron}
\label{subsec:bg:orthogonal}

For $N=2n$, the nonnegative orthogonal Grassmannian is
\begin{equation}
\label{eq:bg:orthogonal-source}
 \begin{gathered}
 \OG_{\geq}(n,2n)
 =\{C\in\Gr_{\geq}(n,2n):CA_{2n}C^T=0\}
 =\{C\in\Gr_{\geq}(n,2n):C=\mathcal D(C)\},\\
 \dim\OG_{\geq}(n,2n)=\frac{n(n-1)}2.
 \end{gathered}
\end{equation}
Its cells, also called \emph{orthitroid cells}, are indexed by perfect
matchings and parametrized by reduced medial (OG) graphs
~\cite{HuangWen2014,KimLee2014}.
Their atomic construction uses insertion of an adjacent positive isotropic
line $(1:1)$, increasing $(n,2n)$ to $(n+1,2n+2)$, and admissible
orthogonal bridges. A bridge is the identity outside the indicated
cyclically adjacent block and has
\begin{equation}
\label{eq:bg:orthogonal-bridge}
 R_i(t)|_{(i,i+1)}=
 \begin{pmatrix}\cosh t&\sigma_i\sinh t\\\sigma_i\sinh t&\cosh t\end{pmatrix},
 \qquad R_i(t)A_{2n}R_i(t)^T=A_{2n},\qquad t>0.
\end{equation}
Here $i+1$ is read modulo $2n$, with $\sigma_i=1$ for $i<2n$ and
$\sigma_{2n}=(-1)^{n-1}$. The seam sign from
\eqref{eq:sym:dihedral-lifts} ensures ordered-minor positivity; either
sign preserves the alternating form. Pair insertions and reduced bridge
sequences generate the orthitroid charts
~\cite{KimLee2014,OrenPerlsteinTessler2025}.

For positive $\widetilde\Lambda\in\Mat_{2n,n+2}$, $n\geq2$, the ABJM map and
its kinematics are~\cite{HuangEtAl2021,HeKuoZhang2021}
\begin{equation}
\label{eq:bg:ABJM-map}
 \begin{gathered}
 \mathcal O_n(\widetilde\Lambda)=\{C\widetilde\Lambda:C\in\OG_{\geq}(n,2n)\},\qquad
 \xi=C^\perp\cap\mathcal W,\quad \xi A_{2n}\xi^T=0,\\
 \lambda=\xi A_{2n},\qquad\widetilde\lambda=\xi,\qquad
 p_i=(-1)^{i-1}\xi_i\xi_i^T,\qquad
 \dim\mathcal O_n=2n-3.
 \end{gathered}
\end{equation}
This specializes \eqref{eq:bg:momentum-map} by taking
$\Lambda=A_{2n}\widetilde\Lambda$. The fixed-external kinematic space is
the isotropic two-plane locus
in $\Gr(2,\mathcal W)$; the symmetric matrix equation imposes three
conditions on its regular part. Here the two chiral sectors are identified,
with no permutation of the particles.

Nonterminal factorization facets have odd $|I|$: writing
$|I|+1=2n_L$, $|I^c|+1=2n_R$, both augmented systems are orthogonal
and $n_L+n_R=n+1$. Their amalgamation has no continuous internal
rescaling after orthogonal normalization, so BCFW boundary dimensions add:
\begin{equation}
\label{eq:bg:ABJM-factorization}
 (2n_L-3)+(2n_R-3)=2n-4.
\end{equation}
An orthogonal bridge restores the full dimension. Recursion ends at the
four-particle geometry; its endpoint boundaries are treated separately.
Even nonterminal channels give deeper strata, not additional generic
factorization facets~\cite{HuangWen2014,HuangEtAl2021,OrenPerlsteinTessler2025}.

\subsection{The reflected Lagrangian model}
\label{subsec:bg:reflected}

\paragraph{Domain and atoms.}
Write $\bar i=2n+1-i$ and use Karpman's anti-diagonal symplectic form
~\cite{Karpman2018,KarpmanSu2015}:
\begin{equation}
\label{eq:bg:reflected-source}
 \begin{gathered}
 (E_{\refl})_{ij}=(-1)^j\delta_{i,\bar j},\qquad
 E_{\refl}^T=-E_{\refl},\quad E_{\refl}^2=-I,\\
 \LG_{\geq}(n,2n)^{\refl}
   =\{C\in\Gr_{\geq}(n,2n):CE_{\refl}C^T=0\},\\
 C^\perp=CE_{\refl},\quad
 \dim\LG_{\geq}(n,2n)^{\refl}=\frac{n(n+1)}2.
 \end{gathered}
\end{equation}
Its graph charts are invariant under reflection across the gaps
$2n\mid1$ and $n\mid n+1$, together with exchange of black and white
vertices; reflected edges have equal weights. Reduced symmetric bridge
graphs yield the positive cells~\cite{Karpman2018,KarpmanSu2015}.
The elementary operations are
\begin{equation}
\label{eq:bg:reflected-atoms}
 \IP^{\refl}_i=\inc_i\pre_{\bar i},\qquad
 M_i^{\refl,\rightarrow}(t)
  =I+t(e_{i,i+1}+e_{\overline{i+1},\bar i}),\qquad 1\leq i<n,
\end{equation}
and the transposed bridges. The insertion adds one reflected loop--coloop
pair and one row. The two axis bridges are
$I+t e_{n,n+1}$ and $I+(-1)^{n-1}t e_{2n,1}$, with their transposes.
Starting from the empty plane, insertions and admissible reduced bridges
generate the cell charts~\cite{Karpman2018}. All preserve $E_{\refl}$;
an axis bridge here carries a free positive parameter.

\paragraph{Map and kinematics.}
The recently constructed reflected Lagrangian amplituhedron~\cite{KroviTessler2026} is
\begin{equation}
\label{eq:bg:reflected-map}
 \RefAmp_n(\widetilde\Lambda)=\{C\widetilde\Lambda:C\in\LG_{\geq}(n,2n)^{\refl}\},
 \qquad \widetilde\Lambda\in\Mat^{>}_{2n,n+2},\quad n\geq2.
\end{equation}
In the $B$-picture it is a region of the symplectically isotropic two-planes in
$\mathcal W$, with regular dimension $2n-1$. Specifically,
\begin{subequations}
\label{eq:bg:reflected-kinematics}
\begin{align}
 \xi&=C^\perp\cap\mathcal W,&
 \xi E_{\refl}\xi^T&=0,&
 \widetilde\lambda&=\xi,\quad\lambda=\xi E_{\refl},
 \label{eq:bg:reflected-spinors}\\
 \lambda_i&=(-1)^i\xi_{\bar i},&
 p_{\bar i}&=-p_i^T,&
 s_I&=s_{\bar I}=s_{I^c}=s_{\overline{I^c}}.
 \label{eq:bg:reflected-folding}
\end{align}
\end{subequations}
This is the momentum construction with
$\Lambda=E_{\refl}^T\widetilde\Lambda$, so only one external matrix is
independent. The single symplectic equation is momentum conservation.

\paragraph{Boundary factorizations.}
For the reflection-invariant intervals $J_r=[\bar r,r]_{\mathrm{cyc}}$,
\begin{equation}
\label{eq:bg:reflected-squares}
 \mu_r=\sum_{i=1}^r(-1)^{i+1}[i\bar i],\qquad
 \mu_0=\mu_n=0,\qquad
 Q_{J_r}=\mu_r\begin{pmatrix}0&1\\-1&0\end{pmatrix},\qquad
 s_{J_r}=\mu_r^2.
\end{equation}
The reduced boundary function is $\mu_r$, not its square. In the boundary
description of~\cite{KroviTessler2026}, its maximal domain locus is the
direct sum of two independently reflected blocks:
\begin{equation}
\label{eq:bg:reflected-direct-sum}
 \begin{gathered}
 C=C_1\oplus C_2,\qquad
 C_1\in\LG_{\geq}(r,J_r)^{\refl},\quad
 C_2\in\LG_{\geq}(n-r,J_r^c)^{\refl},\\
 d_{\mathrm{dom}}=\frac{r(r+1)}2+
                  \frac{(n-r)(n-r+1)}2.
 \end{gathered}
\end{equation}
Both subsystems have zero total momentum. Under the boundary hypotheses
of~\cite{KroviTessler2026}, their induced image has dimension $2n-2$.

For a half-contained interval $I=[a,b]\subset[1,n]$, put $m=|I|$ and
$h=n-m+1$. The proposed maximal domain strata over $s_I=0$ are
indexed by $1\leq\kappa\leq m$ and have the form
\begin{equation}
\label{eq:bg:reflected-blob}
 \begin{gathered}
 C=A\star_* H\star_{\bar *}A^{\refl,\vee},\qquad
 A\in\Gr_{\geq}(\kappa,m+1),\quad
 H\in\LG_{\geq}(h,2h)^{\refl},\\
 d_{\mathrm{dom}}=\frac{h(h+1)}2+
                         \kappa(m+1-\kappa)-1.
 \end{gathered}
\end{equation}
Here $A^{\refl,\vee}$ is the reflected color-dual of $A$, of rank
$m+1-\kappa$; it contributes no independent parameters. The core carries
the remaining labels and the pair $*,\bar *$. Insertions
\eqref{eq:bg:reflected-atoms} grow the blobs, and internal bridges preserve
the equation $s_I=0$. Nondegenerate facet sectors have
$2\leq\kappa\leq m-1$; at $m=2$, ranks $1,2$ give the two chiral
branches~\cite{KroviTessler2026}.
Rank-one reflected factors, including $h=1$, mean factors with the
\emph{induced} boundary data, not the full-rank map
\eqref{eq:bg:reflected-map} with $n=1$.

Modulo reflection and complement, the proposed primary supports are
$\mu_r=0$ and $s_{[a,b]}=0$, with the latter resolved into chiral factors
when $b=a+1$. The identity
\begin{equation}
\label{eq:bg:reflected-crossing}
 s_{[\bar b,a]_{\mathrm{cyc}}}
     =s_{[a+1,b]}+\mu_a\mu_b\qquad(1\leq a<b\leq n)
\end{equation}
explains the reduction of asymmetric axis-crossing channels. When $b=a+1$
the first term is zero; under the strong positivity conditions, the other cases
are confined to intersections of existing supports~\cite{KroviTessler2026}.

\paragraph{Reflected BCFW decomposition.}
For a chosen reflected bridge, apply
\eqref{eq:bg:bridge-transfer} to the reduced boundary function $F_B$.
Every primary removal equation is affine-linear:
\begin{equation}
\label{eq:bg:reflected-removal}
 F_B^M(t)=F_B+t\dot F_{B,M},\qquad
 t_B=-\frac{F_B}{\dot F_{B,M}}\quad(\dot F_{B,M}\ne0).
\end{equation}
A boundary is \emph{hit} precisely when its reduced function depends on
$t$. A paired bridge at $i\mid i+1$ hits $\mu_i$ and the half-contained
channels separating $i$ and $i+1$; an axis bridge hits no $\mu_r$.
Recursively refine the lower factors in
\eqref{eq:bg:reflected-direct-sum}--\eqref{eq:bg:reflected-blob} into
BCFW cells and adjoin the normal bridge. This gives the proposed
$(2n-1)$-dimensional tiling cells and the corresponding instance of
\eqref{eq:bg:canonical-sum}~\cite{KroviTessler2026}.
The completeness of this boundary description and the reflected tiling
are stated there with their strong-positivity proofs deferred to
\cite{OrenPerlsteinKrovi}.

\subsection{Fraser's cyclic nonnegative Grassmannian}
\label{subsec:bg:Fraser}

In this subsection and the following one we recall only the domain geometry, since the corresponding amplituhedra were not constructed prior to this work. Let $N=p\ell$,
$p\geq2$, and $r=S_k^\ell$. Fraser studies
\begin{equation}
\label{eq:bg:Fraser-source}
 \FraserPos_{p;k,N}=(\Gr_{\geq}(k,N))^r.
\end{equation}
It is a closed ball; the intersections with $\ell$-periodic positroid cells
give its cell decomposition~\cite[Theorem~5.1]{Fraser2020}.

\paragraph{Representations and dimension.}
The following centered indexing fixes the character names used later.
Over $\mathbb C$ define
\begin{subequations}
\label{eq:bg:cyclic-modes}
\begin{align}
 z_t&=\exp\!\left(\frac{\pi\mathrm i(2t-k+1)}N\right),\quad
 v_t=(1,z_t^{-1},\ldots,z_t^{-(N-1)})^T,\quad S_kv_t=z_tv_t,
 \label{eq:bg:Fourier-basis}\\
 \zeta_j&=\exp\!\left(\frac{\pi\mathrm i(2j-k+1)}p\right),\qquad
 E_j=\ker(r-\zeta_j I)\simeq V_j\otimes\mathbb C^\ell.
 \label{eq:bg:cyclic-characters}
\end{align}
\end{subequations}
Here $t\in\mathbb Z/N\mathbb Z$, $j\in\mathbb Z/p\mathbb Z$, and
$V_j$ is the one-dimensional representation on which $r$ acts as
$\zeta_j$. This is a representation of the \emph{signed} cyclic lift:
$r^p=(-1)^{k-1}I$.
Thus its scalar sign must not be discarded when $k$ is even.
Fraser proves that all of \eqref{eq:bg:Fraser-source} lies in the unique
complex component containing the cyclically symmetric positive point
$K_{k,N}$ of \eqref{eq:bg:Karp-point}
~\cite[Definition/Lemma~4.8]{Fraser2020}. The point has complexification
$\operatorname{span}(v_0,\ldots,v_{k-1})$, and
$rv_t=z_t^\ell v_t=\zeta_{t\bmod p}v_t$. Thus component selection
determines the labelled multiplicities, not just their unordered collection.
Writing $k=\alpha p+\beta$, $0\leq\beta<p$, we obtain
\begin{equation}
\label{eq:bg:balanced-multiplicities}
 \begin{gathered}
 C_{\mathbb C}=\bigoplus_j V_j\otimes C_j,\qquad c_j:=\dim C_j,\\
 c_j=\#\{0\leq t<k:t\equiv j\pmod p\}
 =\begin{cases}\alpha+1,&0\leq j<\beta,\\\alpha,&\beta\leq j<p.
 \end{cases}
 \end{gathered}
\end{equation}
Consequently the distinguished complex component and its dimension are
~\cite[Proposition~4.1 and Definition/Lemma~4.8]{Fraser2020}
\begin{equation}
\label{eq:bg:Fraser-component}
 X^{\mathrm{cyc}}_{p;k,N,\mathbb C}\simeq
 \prod_j\Gr(c_j,\ell),\qquad
 d^{\mathrm{cyc}}_{p;k,N}
 =\sum_jc_j(\ell-c_j)
 =\frac{k(N-k)-\beta(p-\beta)}p.
\end{equation}
This is also the real dimension of \eqref{eq:bg:Fraser-source}.
The nonnegative locus is \emph{not} asserted to be a product of ordinary
positive Grassmannians: positivity is imposed in the original particle
basis. Complex conjugation pairs the modes by
\begin{equation}
\label{eq:bg:cyclic-mode-pairing}
 j^\vee=k-1-j\pmod p,\qquad
 \overline{\zeta_j}=\zeta_{j^\vee},\qquad c_j=c_{j^\vee}.
\end{equation}
For later use, the relation $srs=r^{-1}$ gives the same pairing for any dihedral
reflection. Non-self-paired characters form the two slots of a real
rotation representation; self-paired modes satisfy $2j=k-1\pmod p$.
There is one such mode for odd $p$, and two or none for even $p$ according
as $k$ is odd or even. These labels refer to the fixed ambient signed
action, not to a separately recentered indexing at each rank.

\paragraph{Atomic construction.}
Let $K_{q,M}\in\Gr_{>}(q,M)$ denote the cyclically symmetric positive
point. By Karp's theorem, it is the unique nonnegative point fixed by
the one-step signed cyclic shift~\cite[Theorem~1.1]{Karp2019}.
Its ordered minors are
\begin{equation}
\label{eq:bg:Karp-point}
 \Delta_{i_1\cdots i_q}(K_{q,M})
   \ \propto\ \prod_{a<b}\sin\!\frac{\pi(i_b-i_a)}M.
\end{equation}
The common proportionality factor can be chosen positive. This follows
from Scott's trigonometric determinant identity~\cite{Scott1879};
see \cite[Lemma~2.3 and Proposition~2.5]{Karp2019} for a proof covering
both parities, and \cite[Lemma~3.1(ii)]{GalashinKarpLam2022}.
The complexification of $K_{q,M}$ is spanned by the centered modes
$v_0,\ldots,v_{q-1}$ of
\eqref{eq:bg:Fourier-basis}, with $(k,N)=(q,M)$.
For remainder $\beta>0$, place $K_{\beta,p}$ on one $p$-element label
orbit; for $\beta=0$ start from the empty plane. Insert $\alpha$ full
coloop orbits and make all other orbits loops. These are the zero-cell
seeds of~\cite[Example~5.4]{Fraser2020}. In our atomic notation, with
positions understood in the \emph{final cyclic order}, the orbit operations and folded bridges are
\begin{equation}
\label{eq:bg:Fraser-atoms}
 \begin{gathered}
 O_i=\{i,i+\ell,\ldots,i+(p-1)\ell\},\qquad
 \pre^{r}_{O_i}=\prod_{a=0}^{p-1}\pre_{i+a\ell},\quad
 \inc^{r}_{O_i}=\prod_{a=0}^{p-1}\inc_{i+a\ell},\\
 B^r_{u,v}(t)=\prod_{a=0}^{p-1}B_{u+a\ell,v+a\ell}(t),
 \qquad u<v<u+\ell,\quad t>0.
 \end{gathered}
\end{equation}
The admissibility and signs of the ordinary bridges include intervening
lollipops and cyclic seams. The orbit members commute and share one
parameter. Applied equivariantly, Fraser's bridge reduction constructs every fixed cell from a
zero-cell seed; 
a maximal reduced sequence of orbit operations reaching the top cell has \(d^{\mathrm{cyc}}_{p;k,N}\) parameters
~\cite[Lemma~5.9 and Section~5.3]{Fraser2020}.
Its positroid index satisfies $f(i+\ell)=f(i)+\ell$; the top index is
$f(i)=i+k$. The atomic core can remain unresolved: resolving it by an
ordinary plabic graph need not preserve visible cyclic symmetry, although
the represented point is symmetric~\cite[Examples~6.6 and~6.8]{Fraser2020}.

\subsection{Shevchenko's rotated Lagrangian Grassmannian}
\label{subsec:bg:Shevchenko}

Let $N=2n$ and write $\vartheta(i)=i+n$ modulo $2n$. Using the
alternating diagonal matrix $A_n$ of \eqref{eq:sym:positive-duality},
Shevchenko's nonnegative Lagrangian Grassmannian~\cite{Shevchenko2025}
is defined by
\begin{equation}
\label{eq:bg:rotated-domain}
 \begin{gathered}
 \Omega_{\mathrm{rot}}=\begin{pmatrix}0&A_n\\-A_n&0\end{pmatrix}
                      =-S_n^nA_{2n},\\
 \LG_{\geq}(n,2n)^{\mathrm{rot}}
 =\{C\in\Gr_{\geq}(n,2n):C\Omega_{\mathrm{rot}}C^T=0\}.
 \end{gathered}
\end{equation}
Equivalently, it is the fixed locus of $C\mapsto S_n^n\mathcal D(C)$:
its Pl\"ucker coordinates satisfy
$\Delta_I(C)=\Delta_{\vartheta(I^c)}(C)$ for every $n$-subset $I$.
Thus the symmetry combines a half-turn with exchange, rather than
imposing either separately. The domain is a closed ball of dimension
$n(n+1)/2$. Its nonempty intersections with positroid cells are cells,
represented by plabic graphs invariant under a half-turn together with
exchange of black and white vertices; paired edges have equal
weights~\cite[Theorem~1.2]{Shevchenko2025}.

The atomic operations insert a loop and a coloop at opposite positions
on the evolving cyclic label set, increasing $(n,2n)$ to $(n+1,2n+2)$,
and attach paired bridges
\begin{equation}
\label{eq:bg:rotated-atoms}
 M_i^{\mathrm{rot},\rightarrow}(t)
   =I+t(e_{i,i+1}+e_{i+n+1,i+n}),\qquad
 M_i^{\mathrm{rot},\leftarrow}(t)
   =(M_i^{\mathrm{rot},\rightarrow}(t))^T,
 \quad 1\leq i<n,\quad t>0.
\end{equation}
At cyclic seams use the inherited signed bridges. Each pair consists of
an ordinary bridge and its half-turned color-dual, with one common
parameter. Together with the loop--coloop insertions, admissible paired
bridges give positive coordinates on all the cells
\cite[Sections~5.2--5.5]{Shevchenko2025}.
\section{Particle symmetries and their actions}
\label{sec:symmetries}
In this section we study the possible positivity preserving symmetry groups relevant to our study of nonnegative Grassmannians and amplituhedra. 
\subsection{Symmetry data and positive lifts}
\label{subsec:symmetry-data}
Let $N\geq3$ and let $[N]=\{1,\ldots,N\}$ be cyclically ordered, with indices understood
modulo $N$, and let $D_N$ be the dihedral group of order $2N$.
Write $\tau(i)=i+1$ and $j(i)=N+1-i$ for its standard generators.
A \emph{symmetry datum} is a finite group $G$ together with an injective
homomorphism
\begin{equation}
\label{eq:sym:datum}
 \rho:G\hookrightarrow D_N\times E,
 \qquad E=\mathbb Z/2\mathbb Z,
 \qquad \rho(\gamma)=(g_\gamma,\epsilon(\gamma)).
\end{equation}
We call $L:=\operatorname{pr}_{D_N}\rho(G)$ the \emph{label group},
$\epsilon$ the \emph{exchange character}, and $H:=\ker\epsilon$ the
\emph{linear subgroup}. The factor $E$ fixes particle labels and exchanges
the two spinor sectors, and likewise the two external subspaces.
Thus, on the spinor slots $[N]\times E$, with $0$ denoting $\lambda$ and
$1$ denoting $\widetilde\lambda$, the action is
\begin{equation}
\label{eq:sym:slots}
 \gamma\cdot(i,a)=(g_\gamma(i),a+\epsilon(\gamma)).
\end{equation}

The action on source planes uses the signed maps in
\eqref{eq:sym:positive-lifts}. On subspaces, $\mathcal D^2=1$,
$\mathcal D S_k=S_{N-k}\mathcal D$, and $\mathcal D J=J\mathcal D$.
Consequently they realize the direct product in~\eqref{eq:sym:datum}.
In particular, a fixed source of rank $k$ can have exchange symmetry only if
\begin{equation}
\label{eq:sym:middle-rank}
 \epsilon\ne0\quad\Longrightarrow\quad N=2k.
\end{equation}
Throughout this paper, since our construction involves $(G,\rho)-$fixed spaces, exchange-bearing models have even $N=2k$; no
additional evenness assumption is imposed on $k$.

For a fixed-rank problem satisfying~\eqref{eq:sym:middle-rank}, let
$\widehat g_k$ denote the signed label map induced by $S_k,J$, defined
up to a scalar sign. The source-plane, spinor and external actions are
\begin{subequations}
\label{eq:sym:actions}
\begin{align}
 T_\gamma&:=\widehat{g_\gamma}_k A_N^{\epsilon(\gamma)},
 \label{eq:sym:signed-map}\\
 \gamma\cdot C&=
 \begin{cases}
 T_\gamma C,&\epsilon(\gamma)=0,\\
 T_\gamma(C^\perp),&\epsilon(\gamma)=1,
 \end{cases}
 \label{eq:sym:source-action}\\
 \gamma\cdot(X_0,X_1)
   &=\bigl(T_\gamma X_{\epsilon(\gamma)},
           T_\gamma X_{1-\epsilon(\gamma)}\bigr),
 \qquad
 (X_0,X_1)=(\lambda,\widetilde\lambda)
 \text{ or }(\Lambda,\widetilde\Lambda).
 \label{eq:sym:pair-action}
\end{align}
\end{subequations}
All equalities concern \emph{subspaces}, not chosen matrix representatives.
In particular, the spinor and external subspaces inherit the ambient
signed maps; their seam signs are not chosen anew from their dimensions.
The alternating signs in an exchange can equivalently be absorbed into
its signed label map. Scalar signs disappear on Grassmannians but must
be retained for representation calculations: we use the finite signed
lift $\widehat G=\langle -I,T_\gamma:\gamma\in G\rangle$ when working
on vector spaces.

\subsection{Subgroups and principal families}
\label{subsec:symmetry-families}

We identify $G$ with $\rho(G)$ when no confusion can arise.
Every label group is cyclic or dihedral. For $N=p\ell$, write
\begin{equation}
\label{eq:sym:label-groups}
 r=\tau^\ell,\qquad s_c(i)=c-i\pmod N,\qquad
 C_p=\langle r\rangle,\qquad D_{p,c}=\langle r,s_c\rangle.
\end{equation}
We allow $p=1$, so that $C_1$ is trivial and $D_{1,c}$ is a single
reflection. For either label group $L$, every subgroup projecting onto
$L$ has exactly one of the forms
\begin{equation}
\label{eq:sym:subgroup-classification}
 G=L\times E
 \qquad\text{or}\qquad
 G=\Gamma_\chi(L):=\{(g,\chi(g)):g\in L\},
 \quad \chi\in\operatorname{Hom}(L,E).
\end{equation}
Indeed, the kernel of the label projection is either $E$ or trivial.
In the first case both lifts of every label symmetry occur; in the
second, the unique lifts define the character $\chi$.
The choice $\chi=0$ gives the purely linear actions.

For $L=C_p$, a character is determined by $\alpha=\chi(r)$, with
$p\alpha=0$ in $E$. Besides the linear cyclic family, this gives
\begin{equation}
\label{eq:sym:cyclic-exchange}
 G_{\mathrm{sAB}}=C_p\times E,
 \qquad
 G_{\mathrm{sSh}}=\langle(r,1)\rangle\quad(p\text{ even}),
\end{equation}
called the \emph{symmetric ABJM} and \emph{symmetric Shevchenko}
families. For odd $p$, the generator $(r,1)$ instead generates the
product $C_p\times E$, since $(r,1)^p=(1,1)$.

For $L=D_{p,c}$, put $s=s_c$. A character is determined by
\begin{equation}
\label{eq:sym:dihedral-character}
 \alpha=\chi(r),\qquad\beta=\chi(s),\qquad p\alpha=0,
 \qquad \chi(r^a s)=a\alpha+\beta\quad\text{in }E.
\end{equation}
The graph subgroups $\Gamma_\chi(L)$ are therefore
\begin{equation}
\label{eq:sym:dihedral-families}
\begin{array}{c|c|l}
 (\alpha,\beta)&\text{generators}&\text{family}\\ \hline
 (0,0)&(r,0),(s,0)&\text{linear dihedral}\\
 (0,1)&(r,0),(s,1)&\text{reflection-twisted}\\
 (1,0)&(r,1),(s,0)&\text{mixed phase }0\quad(p\text{ even})\\
 (1,1)&(r,1),(s,1)&\text{mixed phase }1\quad(p\text{ even}).
\end{array}
\end{equation}
The remaining family is the product $D_{p,c}\times E$, called
\emph{dihedral ABJM}. In the mixed families one reflection class acts
linearly and the other with exchange. The phase is relative to the
specified reflection $s$: replacing $s$ by $rs$ changes $\beta$ to
$\beta+1$, without changing the action itself.

\paragraph{Basic specializations.}
The trivial group gives ordinary momentum geometry;
$C_p\times\{0\}$ gives the cyclic domain studied by
Fraser~\cite[Theorem~5.1]{Fraser2020}; and $\langle(s,0)\rangle$ gives a linear
reflection. At $N=2n$, $k=n$, write $\vartheta=\tau^n$ for the
half-turn of Section~\ref{subsec:bg:Shevchenko}. The three basic
exchange actions are
\begin{equation}
\label{eq:sym:basic-exchange}
\begin{array}{c|c|l}
 G&\text{fixed-plane equation}&\text{domain geometry}\\ \hline
 \langle(1,1)\rangle&C=\mathcal D(C)&\text{orthogonal (ABJM)}\\
 \langle(j,1)\rangle&C=J\mathcal D(C)&\text{reflected Lagrangian}\\
 \langle(\vartheta,1)\rangle&C=S_n^n\mathcal D(C)&\text{rotated Lagrangian}.
\end{array}
\end{equation}
The first is the domain of the ABJM momentum amplituhedron
\cite{HuangEtAl2021,HeKuoZhang2021}; the other two are the positive
Lagrangian domains of Karpman--Su and Shevchenko,
respectively~\cite{KarpmanSu2015,Shevchenko2025}.
The distinction is already visible in the bilinear forms: $A_N$ is
symmetric, whereas $JA_N$ and $S_n^nA_N$ are skew-symmetric.

\subsection{Fixed vertices, fixed edges and canonical arcs}
\label{subsec:reflection-geometry}

Suppose $L=D_{p,c}$. Since $r^a s_c=s_{c+a\ell}$, the labelled subgroup
depends only on $c\pmod\ell$. Conjugating by $\tau^b$ sends $c$ to
$c+2b$. Thus the label action has one type up to cyclic relabelling
when $\ell$ is odd, and two when $\ell$ is even.

A \emph{fixed vertex} is a particle fixed by some reflection in $L$.
A \emph{fixed edge} is a boundary edge preserved setwise by some
reflection in $L$; that reflection exchanges its endpoints. These
endpoints are called \emph{nearly fixed vertices}. Writing
$b_i=\{i,i+1\}$ for a boundary edge, the marked sets are
\begin{subequations}
\label{eq:sym:marked-sets}
\begin{align}
 V_\rho&=\{i\in[N]:2i\equiv c\pmod\ell\},
 &E_\rho&=\{b_i:2i+1\equiv c\pmod\ell\},
 \label{eq:sym:fixed-congruences}
\end{align}
\end{subequations}
Here ``fixed'' is with respect to some reflection, not necessarily to all of $G$. 
On the quotient by $\langle r\rangle$, the possibilities are
\begin{equation}
\label{eq:sym:quotient-types}
\begin{array}{c|c|c}
 \text{quotient type}&\text{fixed vertices}&\text{fixed edges}\\ \hline
 \ell\text{ odd}&1&1\\
 \ell\text{ even},\ c\text{ even (vertex type)}&2&0\\
 \ell\text{ even},\ c\text{ odd (edge type)}&0&2.
\end{array}
\end{equation}
The quotient retains its $\ell$ distinct boundary-edge positions, also
when $\ell=1,2$. Each fixed quotient vertex or edge lifts to an orbit
of $p$ vertices or edges upstairs. Vertices outside $V_\rho$ have
label orbits of size $2p$.

The closed boundary arcs between consecutive elements of $M_\rho$,
with no fixed or nearly fixed vertex in their interiors, are the \emph{canonical arcs}.
Together with the designation of fixed vertices and fixed edges, they
form the \emph{marked-arc decomposition} of the label action.

For the full symmetry datum, the marks must also retain their exchange
information. A lift $(s,0)$ of a reflection is \emph{linear}, whereas
$(s,1)$ is \emph{exchange-twisted}; a product group contains both lifts
of every reflection. In a mixed family, the two reflection classes
alternate according to~\eqref{eq:sym:dihedral-character}. When $\ell$
is odd they have different geometric axis types, one through vertices
and one through edges; when $\ell$ is even they have the same axis type.
Thus neither the abstract group nor the unlabelled quotient polygon
alone specifies the symmetry relevant to a boundary: its stabilizer in
$G$, including the exchange character, must be retained.
\begin{example}\label{ex:fixed_vertex_edge}
    Consider subgroups of $D_4\times\mathbb{Z}_2$ which is generated by a single reflection. The subgroup generated by the reflection which fixes $1,3$ and flips $2,4$ has two fixed vertices, $1,3.$ The subgroup generated by the reflection which takes $i$ to $5-i$ has no fixed vertex, but two fixed edges $\{1,4\}$ and $\{2,3\}.$ All the four labels are nearly-fixed vertices. In the former case the canonical arcs are $\{1,2,3\},\{3,4,1\},$ and in the latter they are $\{1,2\},\{3,4\}.$
\end{example}

\section{Symmetric nonnegative Grassmannians}
\label{sec:domains}
This section studies the symmetric nonnegative Grassmannians, which are the fixed loci of the nonnegative Grassmannian under the symmetry actions of the previous section. We describe the symmetry action on the positroids, the topology and a representation theoretic perspective, the atomic operations and how they can be used to construct the cells. 
\subsection{Fixed loci, topology and representation type}
\label{subsec:domain:geometry}

Fix a symmetry datum $(G,\rho)$ as in \eqref{eq:sym:datum}, with the standard
positive lifts of \eqref{eq:sym:positive-lifts}. We define its
\emph{symmetric Grassmannian} and its nonnegative and positive parts by
\begin{equation}
\label{eq:dom:definition}
 \begin{gathered}
 X=\SymGr{G}{\rho}{k,N}
   :=\{C\in\Gr(k,N):\gamma\cdot C=C\text{ for every }\gamma\in G\},\\
 X^{\geq}=\SymGrNonneg{G}{\rho}{k,N}:=X\cap\Gr_{\geq}(k,N),\qquad
 X^>=\SymGrPos{G}{\rho}{k,N}:=X\cap\Gr_{>}(k,N).
 \end{gathered}
\end{equation}
Here the action is \eqref{eq:sym:source-action}, not the linear action of
all signed matrices. An exchanging element sends a $k$-plane to an
$(N-k)$-plane; hence $\epsilon\ne0$ forces $N=2k$, as in
\eqref{eq:sym:middle-rank}. We assume this rank compatibility throughout, whenever the exchange factor is non trivial.
We treat $0<k<N$; the rank-zero and full-rank linear cases are points.

\paragraph{Positroid symmetry.}
For $I\subset[N]$, write
\begin{equation}
\label{eq:dom:subset-action}
 \gamma\star I=
 \begin{cases}g_\gamma(I),&\epsilon(\gamma)=0,\\
               g_\gamma(I^c),&\epsilon(\gamma)=1.
 \end{cases}
 \qquad
 \mathcal B(\gamma\cdot C)=\gamma\star\mathcal B(C),
\end{equation}
where $\mathcal B(C)=\{I:\Delta_I(C)\ne0\}$ is the set of bases of its
matroid. Thus linear elements relabel the positroid, whereas exchanging
elements also take its matroid dual. More precisely, after normalizing
$\sum_I\Delta_I(C)=1$, the fixed-point equations on the nonnegative
Grassmannian are exactly
\begin{equation}
\label{eq:dom:Plucker-fixed}
 \Delta_I(C)=\Delta_{\gamma\star I}(C)
 \qquad(\gamma\in G,\ I\in\tbinom{[N]}k).
\end{equation}
Indeed, the signed lifts permute the projective Pl\"ucker coordinates
without relative signs. In the decorated trip-permutation convention
of~\cite{Postnikov2006}, this becomes
\begin{equation}
\label{eq:dom:permutation-action}
 \pi_{\gamma\cdot C}
   =g_\gamma\pi_C^{\,(-1)^{o(g_\gamma)+\epsilon(\gamma)}}g_\gamma^{-1},
 \qquad o(g)=\begin{cases}0,&g\text{ a rotation},\\1,&g\text{ a reflection}.
 \end{cases}
\end{equation}
Fixed-point decorations are interchanged precisely when $\epsilon=1$:
exchange swaps loops and coloops, while linear symmetries preserve them.
The nonempty intersections $X^{\geq}\cap\Pi_{\mathcal B}$ give the induced
positroid stratification. Invariance of $\mathcal B$ is necessary for such
an intersection. 
Theorem~\ref{thm:dom:cells} below proves that every nonempty intersection
is a cell. 

\paragraph{The distinguished component and its dimension.}
Let $K=K_{k,N}$ be the point in \eqref{eq:bg:Karp-point}, and put
$D=K^\perp$. Let $\widehat G$ be the finite signed lift from
Section~\ref{subsec:symmetry-data}, with lifted exchange character still
denoted $\epsilon$, and set $\widehat H=\ker\epsilon$.
Characters of $K$ and $D$ below are characters of $\widehat H$;
Note that an exchanging element is not an endomorphism of $K$.

For a finite-dimensional real $\widehat H$-representation $R$, we write
$\chi_R(h)=\operatorname{tr}_R(h)$, $h\in\widehat H$.
Thus $\chi_R$ is a real-valued representation character, 
not to be confused with the
$E$-valued homomorphism $\chi$ of
\eqref{eq:sym:subgroup-classification}. In particular $\chi_K(h)=\operatorname{tr}(h|_K)$ and $\chi_D(h)=\operatorname{tr}(h|_D)$ for $h\in\widehat H$. 
An exchanging element $a$ interchanges $K,D$, but $a^2\in\widehat H$,
so $\chi_K(a^2)$ is defined.
\begin{theorem}[The positive component]
\label{thm:dom:ball}
The point $K$ belongs to $X^>$. Exactly one irreducible component of the
complex fixed locus meets $X^{\geq}$: the component $X_0$ containing $K$.
Its complex dimension equals the real dimension of $X^>$, namely
\begin{equation}
\label{eq:dom:character-dimension}
 d_\rho=\frac1{|\widehat G|}
 \left(\sum_{h\in\widehat H}\chi_K(h^{-1})\chi_D(h)
       -\sum_{a\in\widehat G\setminus\widehat H}\chi_K(a^2)\right).
\end{equation}
The second sum is absent for a purely linear action. There is a
homeomorphism of pairs
\begin{equation}
\label{eq:dom:ball-pair}
 (X^{\geq},X^>)\simeq(\overline{\mathbb B}^{\,d_\rho},
                           \mathbb B^{d_\rho}),
\end{equation}
where $\mathbb B^d,\overline{\mathbb B}^d\subseteq\mathbb{R}^d$ are a $d-$dimensional open ball and its closure, respectively. Moreover, for every $C\in X^{\geq}$, including its boundary,
\begin{equation}
\label{eq:dom:constant-type}
 C\simeq K,\qquad C^\perp\simeq D
 \quad\text{as }\widehat H\text{-representations}.
\end{equation}
The orthogonal decompositions $C\oplus C^\perp$ and $K\oplus D$ are also
isomorphic as graded $\widehat G$-representations, with exchanging elements
interchanging the summands.
\end{theorem}
For cyclic symmetry, component selection and dimension are already Fraser's~\cite[Definition/Lemma~4.8]{Fraser2020}, as is the ball theorem~\cite[Theorem~5.1 and Section~5.2]{Fraser2020}. The proof extends Fraser's argument to allow the reflection and exchange symmetries. In particular, the characters in \eqref{eq:dom:character-dimension} can also be computed
explicitly from Fraser's cyclic multiplicities and the signed
reflection and exchange actions, see the paragraph
``Reading the representations'' below. The proof itself is technical, and may be skipped upon first reading.

\begin{proof}
\emph{The common positive point.}
Karp's uniqueness theorem~\cite[Theorem~1.1]{Karp2019} and
\eqref{eq:sym:linear-relations} imply $JK=K$: reversal takes $K$ to
another positive $S_k$-fixed point. At middle rank, the same argument
applies to $\mathcal D$, since it commutes with the signed shift on
Grassmannian points. Thus $K$ is fixed by every admissible symmetry.
In the Fourier basis of \eqref{eq:bg:Fourier-basis},
\begin{equation}
\label{eq:dom:Karp-decomposition}
 K_{\mathbb C}=\bigoplus_{t=0}^{k-1}\mathbb C v_t,\qquad
 D_{\mathbb C}=\bigoplus_{t=k}^{N-1}\mathbb C v_t.
\end{equation}
Strict positivity follows from \eqref{eq:bg:Karp-point}.

\emph{A global chart and linear fixed equations.}
Consider the graph chart centered at $K$,
\begin{equation}
\label{eq:dom:graph-chart}
 \mathcal U_K=\{C:C\cap D=0\},\qquad
 C_\xi=\{u+\xi u:u\in K\},\quad \xi\in\Hom_{\mathbb R}(K,D).
\end{equation}
It contains the entire nonnegative Grassmannian, as in
\cite[Proposition~3.4]{GalashinKarpLam2022}. Indeed, for positive and
nonnegative row representatives $B$ of $K$ and $C$, respectively,
Cauchy--Binet gives
$\det(CB^T)=\sum_I\Delta_I(C)\Delta_I(B)>0$.
Consequently orthogonal projection $C\to K$ is invertible.

In the orthogonal splitting $K\oplus D$, write a linear element as
$\operatorname{diag}(h_K,h_D)$ and an exchanging element as
$\left(\begin{smallmatrix}0&a_D\\a_K&0\end{smallmatrix}\right)$.
For $u\in K$ and $v\in D$, the identity
$\langle u+\xi u,-\xi^Tv+v\rangle=0$ and dimension comparison give
$C_\xi^\perp=\{-\xi^Tv+v:v\in D\}$.
Applying $h$ and $a$ to these respective graph vectors gives
$h_Ku+h_D\xi u$ and $a_Dv-a_K\xi^Tv$; using $h_Ku$ and $a_Dv$
as the new $K$-coordinates yields the \emph{linear} actions
\begin{equation}
\label{eq:dom:graph-action}
 h\cdot\xi=h_D\xi h_K^{-1},\qquad
 a\cdot\xi=-a_K\xi^Ta_D^{-1}.
\end{equation}
Since each plane in $\mathcal U_K$ has a unique graph coordinate,
$\gamma\cdot C_\xi=C_\xi$ is equivalent to $\gamma\cdot\xi=\xi$.
Consequently,
\begin{equation}
\label{eq:dom:fixed-tangent}
 \mathcal T_\rho:=\Hom_{\mathbb R}(K,D)^{\widehat G},
 \qquad X\cap\mathcal U_K=\{C_\xi:\xi\in\mathcal T_\rho\}.
\end{equation}
Thus $\mathcal T_\rho$ describes the fixed locus throughout this chart,
not just its first-order deformations. Its complexification is an
irreducible affine space containing $\xi=0$, so only $X_0$ meets this
chart; since $X^{\geq}\subset\mathcal U_K$, no other component meets
$X^{\geq}$. All minors of $K$ are strictly positive, so the same holds
for sufficiently small $\xi\in\mathcal T_\rho$: $K$ is a relative
interior point of $X^{\geq}$ in $X_0(\mathbb R)$.

\emph{The dimension $d_\rho$.} 
Recall $\xi\in\Hom(K,D)\simeq D\otimes K^*$. Thus, the first operator in
\eqref{eq:dom:graph-action} has trace $\chi_K(h^{-1})\chi_D(h)$. 
For the second, the matrix-unit identity
$\operatorname{tr}(X\mapsto AX^TB)=\operatorname{tr}(AB^T)$ and
$(a_D^{-1})^T=a_D$ give trace
$-\operatorname{tr}(a_Da_K)=-\chi_K(a^2)$.
To calculate the fixed locus dimension we need to calculate the projection on the trivial representation. Writing $\mathsf R_\gamma$ for the induced operator,
$P_{\mathrm{inv}}=|\widehat G|^{-1}\sum_{\gamma\in\widehat G}\mathsf R_\gamma$
is an idempotent with image $\mathcal T_\rho$; hence its trace equals
$\dim\mathcal T_\rho$. Averaging the displayed traces therefore gives
\eqref{eq:dom:character-dimension}, the dimension of the affine chart.
In particular the actual signed square $a^2$ is essential, even when it
is a scalar on the ambient space.

\emph{The contracting flow.}
Set
\begin{equation}
\label{eq:dom:heat-flow}
 \mathsf H=S_k+S_k^T,\qquad F_t=e^{t\mathsf H},\qquad
 \widehat\gamma\mathsf H\widehat\gamma^{-1}
       =(-1)^{\epsilon(\gamma)}\mathsf H.
\end{equation}
Rotations and $J$ commute with $\mathsf H$, whereas $A_N\mathsf H A_N
=-\mathsf H$ when $N$ is even. These are exact matrix identities.
Since $(F_tC)^\perp=F_{-t}C^\perp$ and $A_NF_{-t}=F_tA_N$,
\[
 \mathcal D(F_tC)=A_NF_{-t}C^\perp
                 =F_tA_NC^\perp=F_t\mathcal D(C).
\]
The signed dihedral label maps also commute with $F_t$, so the
\emph{action on source planes} commutes with $C\mapsto F_tC$,
including for exchange.
Galashin--Karp--Lam's positivity-improving property gives
\begin{equation}
\label{eq:dom:flow-positivity}
 F_t(X^{\geq})\subset X^>\qquad(t>0)
\end{equation}
by~\cite[Corollary~3.8]{GalashinKarpLam2022}.

In the graph chart, the flow takes the form
\begin{equation}
\label{eq:dom:flow-contraction}
 F_tC_\xi=C_{f_t(\xi)},\qquad
 f_t(\xi)=e^{t\mathsf H|_D}\xi e^{-t\mathsf H|_K}.
\end{equation}
This is the contractive flow of
\cite[Section~3.3]{GalashinKarpLam2022} on $\Hom_{\mathbb R}(K,D)$.
By the equivariance established above, it preserves $\mathcal T_\rho$;
its restriction is therefore contractive with respect to the induced norm.

The graph preimage of $X^{\geq}$ is compact in $\mathcal T_\rho$,
and that of $X^>$ is open. By \eqref{eq:dom:flow-positivity}, every
positive-time map sends the former into the latter; letting $t\downarrow0$
also shows that the former is the closure of the latter.
Thus \cite[Lemma~2.3]{GalashinKarpLam2022} applies on $\mathcal T_\rho$,
proving \eqref{eq:dom:ball-pair}. In particular, $X^>$ is connected
and has dimension $d_\rho$.


\emph{Representation type, including the boundary.}
For $\xi\in\mathcal T_\rho$, the graph maps $u\mapsto u+\xi u$ and
$v\mapsto -\xi^Tv+v$ are $\widehat H$-equivariant isomorphisms. This proves
\eqref{eq:dom:constant-type} without a genericity assumption. For the
full graded assertion, let $P_U$ denote orthogonal projection onto $U$.
The construction
\begin{equation}
\label{eq:dom:graded-trivialization}
 M_C=P_CP_K+P_{C^\perp}P_D,\qquad
 R_C=M_C(M_C^TM_C)^{-1/2}
\end{equation}
is well-defined: $M_C$ is invertible, and its polar part $R_C$ is
orthogonal. Both send the two summands to $C$ and $C^\perp$. Every element of $\widehat G$ either preserves both
pairs of projections or exchanges both, so it commutes with $M_C$ and
$R_C$. Thus $R_C$ is the required equivariant orthogonal isomorphism.
\end{proof}

\paragraph{Reading the representations.}
We make \eqref{eq:dom:constant-type} explicit for later use. For real
irreducible $\widehat H$-modules $R_\nu$, write
\begin{equation}
\label{eq:dom:real-multiplicities}
 \begin{gathered}
 \mathbb R^N\simeq\bigoplus_\nu R_\nu^{\oplus m_\nu},\qquad
 C\simeq\bigoplus_\nu R_\nu^{\oplus c_\nu},\qquad
 C^\perp\simeq\bigoplus_\nu R_\nu^{\oplus(m_\nu-c_\nu)},\\
 \mathbb D_\nu=\End_{\widehat H}(R_\nu),\qquad
 e_\nu=\dim_{\mathbb R}\mathbb D_\nu,\qquad
 c_\nu=\frac1{e_\nu|\widehat H|}
       \sum_{h\in\widehat H}\chi_K(h)\chi_{R_\nu}(h).
 \end{gathered}
\end{equation}
For a purely linear action, the real fixed locus of this representation
type is a product of Grassmannians over $\mathbb D_\nu$.
A chart in its $\nu$th factor has $c_\nu(m_\nu-c_\nu)$ free entries in
$\mathbb D_\nu$, each contributing $e_\nu$ real coordinates. Hence
\begin{equation}
\label{eq:dom:linear-multiplicity-dimension}
 d_\rho=\sum_\nu e_\nu c_\nu(m_\nu-c_\nu)\qquad(\epsilon=0).
\end{equation}
Crucially, this is an ambient description, not a product decomposition of the
nonnegative part. A nonreal cyclic character pair is a real rotation
plane whose commuting endomorphisms are the complex scalars; adding a
reflection requires these scalars to commute with complex conjugation,
leaving only real scalars. Thus $e_\nu=2$ for such cyclic pairs and
$e_\nu=1$ for real one-dimensional characters and dihedral irreducibles.

Here is a direct prescription in the Fourier notation already fixed.
Let $q$ be the order of the rotational \emph{label subgroup of $H$}, and
put $\ell_H=N/q$, $r_H=S_k^{\ell_H}$. For mixed and symmetric-Shevchenko
families $q=p/2$, not $p$; in the other families $q=p$.
Fraser's component selection~\cite[Definition/Lemma~4.8]{Fraser2020}
applies to this cyclic subgroup. In
\eqref{eq:bg:cyclic-characters}--\eqref{eq:bg:balanced-multiplicities}
replace $p$ by $q$ to obtain
\begin{equation}
\label{eq:dom:kernel-cyclic-types}
 \begin{gathered}
 E_j^H=\bigoplus_{t\equiv j\, (q)}\mathbb Cv_t,\qquad
 \dim E_j^H=\ell_H,\\
 c_j^H=\#\{0\leq t<k:t\equiv j\pmod q\},\qquad
 j^\vee=k-1-j\pmod q.
 \end{gathered}
\end{equation}
Scalar $-I$ acts by $-1$ throughout. 
The remaining ambient generators
are explicitly
\begin{equation}
\label{eq:dom:Fourier-generators}
 \iota(t)=k-1-t\pmod N,\qquad
 Jv_t=(-1)^{k-1}z_t v_{\iota(t)},\qquad
 A_Nv_t=v_{t+n}\quad(N=2n).
\end{equation}
These identities, together with $S_kv_t=z_tv_t$, determine all the
characters in \eqref{eq:dom:character-dimension} and
\eqref{eq:dom:real-multiplicities}.

If $H$ contains a reflection $s_c$, choose its lift
$\widehat s_c=S_k^{c-1}J$. A two-element orbit $\{j,j^\vee\}$ in
\eqref{eq:dom:kernel-cyclic-types} forms a \emph{dihedral multiplet}:
the two cyclic characters are the two eigenlines of one real
two-dimensional irreducible, occurring $c_j^H$ times in $C$ and
$\ell_H$ times in the ambient space. For a self-paired mode,
$2j=k-1\pmod q$, split by the reflection sign. With
$K_j=K_{\mathbb C}\cap E_j^H$, its multiplicities are
\begin{equation}
\label{eq:dom:self-paired-types}
 m_{j,\pm}=\frac{\ell_H\pm\operatorname{tr}(\widehat s_c|_{E_j^H})}2,
 \qquad
 c_{j,\pm}=\frac{c_j^H\pm\operatorname{tr}(\widehat s_c|_{K_j})}2.
\end{equation}Indeed, $\widehat s_c^2=I$, so its two eigenspace dimensions sum to the
total dimension and their difference is its trace.
There is no unspecified choice of reflection signs here: for either
Fourier-index set $F$ defining $E_j^H$ or $K_j$, the trace is
\begin{equation}
\label{eq:dom:reflection-trace}
 \operatorname{tr}\!\left(\widehat s_c\bigm|\bigoplus_{t\in F}\mathbb Cv_t\right)
 =(-1)^{k-1}
   \sum_{\substack{t\in F\\2t\equiv k-1\, (N)}}z_t^{\,2-c}.
\end{equation}
Thus \eqref{eq:dom:kernel-cyclic-types}--\eqref{eq:dom:reflection-trace}
fully specify the decomposition of every $C\in X^{\geq}$ and of its
complement, with the common character labels needed for external data.

When exchange is present, choose $a\in\widehat G\setminus\widehat H$.
Although $\widehat H$ preserves $C$, the linear map $a$ carries $C$ to
$C^\perp$. Since $h(ax)=a(a^{-1}ha)x$, transporting an irreducible type
by $a$ twists its action by conjugation. Thus
\begin{equation}
\label{eq:dom:exchange-partners}
 \nu^a(h)=\nu(a^{-1}ha),\qquad
 a(C)=C^\perp,\qquad
 m_{\nu^a}=m_\nu,\qquad c_{\nu^a}=m_\nu-c_\nu.
\end{equation}
The operation $\nu\mapsto\nu^a$ is an involution on isomorphism classes,
independent of the choice in the exchange coset: $a^2\in\widehat H$
and inner conjugation does not change an irreducible type.
For distinct partner types, choosing the component of $C$ in one
isotypic summand determines its component in the partner summand:
take the orthogonal complement there of its $a$-image. On a self-paired
summand, the same relation constrains the chosen component itself.

For example, if $a^2=\sigma I$ with $\sigma\in\{1,-1\}$, define the
nondegenerate bilinear form $b_a$ on $\mathbb R^N$ by twisting the
Euclidean inner product with $a$:
\begin{equation}
\label{eq:dom:exchange-form}
 b_a(x,y)=\langle ax,y\rangle,\qquad
 b_a(y,x)=\sigma b_a(x,y),\qquad b_a|_{C\times C}=0.
\end{equation}
Orthogonality gives $a^T=a^{-1}=\sigma a$, explaining the symmetry
identity; the vanishing is equivalent to $a(C)=C^\perp$ at middle rank.
Thus $\sigma=1$ gives the orthogonal condition and $\sigma=-1$ the
Lagrangian condition. In a larger group the same calculation is made
on its paired multiplicity spaces, keeping the actual signed maps.
The full exchange action is therefore recorded on the graded pair
$C\oplus C^\perp$, not on $C$ alone.

The resulting closed dimension formulas for the principal families,
including all axis and rank parities, are collected in
Appendix~\ref{app:domain-dimensions}.
\subsection{Equivariant atomic operations and local stabilizers}
\label{subsec:domain:atoms}
Below we describe the atomic operations for our setting.
The operations are organized for a \emph{fixed group $G$ and exchange character},
while the cyclic label set, its action and the rank may change. In particular,
deleting a fixed vertex can create a fixed edge. Subsequent operations use
this induced action, not the original axis type; an ineffective action is
replaced by its effective image. All insertions and contractions below have
the ordered-minor signs of the new rank.

\paragraph{Lollipop insertions.}
Let $O$ be an orbit of new labels. A lollipop assignment
$\eta:O\to\mathbb Z/2\mathbb Z$ specifies loops by $0$ and coloops by $1$.
The compatibility and rank rules are
\begin{equation}
\label{eq:atoms:lollipop-rule}
 \eta(g_\gamma i)=\eta(i)+\epsilon(\gamma),\qquad
 O_a=\eta^{-1}(a),\qquad (N',k')=(N+|O|,k+|O_1|).
\end{equation}
Choosing the assignment at one representative gives $\pre^\rho_O$ or
$\inc^\rho_O$. Their signed matrix embeddings are characterized by
\begin{equation}
\label{eq:atoms:insertion-minors}
 \Delta_{I\cup O_1}(\mathsf{Ins}_{O,\eta}C)=\Delta_I(C),\qquad
 \Delta_J(\mathsf{Ins}_{O,\eta}C)=0\quad(O_1\not\subset J\text{ or }J\cap O_0\ne\varnothing).
\end{equation}
An assignment exists exactly when a label stabilizer contains no exchanging
element. Linear actions preserve its type, whereas exchange interchanges
loops and coloops. Existing lollipop orbits can always be deleted with the
induced signs; their rank feasibility is automatic. This is the equivariant
form of~\cite[Lemmas~7.8--7.9]{Lam2015}, including the paired insertions of
Karpman and Shevchenko~\cite{Karpman2018,Shevchenko2025}.

\paragraph{Small blocks at boundary stabilizers.}
At a boundary vertex or edge midpoint the effective local stabilizer is one of
\begin{equation}
\label{eq:atoms:local-stabilizers}
 1,\qquad\langle R\rangle,\qquad\langle\mathcal D\rangle,
 \qquad\langle R\mathcal D\rangle,\qquad\langle R,\mathcal D\rangle.
\end{equation}
Here $R$ is the local reflection. 
A \emph{normalized rigid graft} is a prescribed positive block,
with specified attaching legs and local signed action, whose gluing adds no
parameter. Its entire $G$-orbit is inserted simultaneously, using dual
connector normalizations on exchange-paired grafts. For a \emph{single copy},
before taking the orbit, grafting $B\in\Gr(b,M)$ at one existing leg or
inserting it as a detached summand gives, respectively,
\begin{equation}
\label{eq:atoms:graft-counts}
 (\Delta N,\Delta k)=(M-2,b-1),\qquad (\Delta N,\Delta k)=(M,b).
\end{equation}
For disjoint orbit copies, add these increments, using rank $M-b$ for an
exchange-dual copy. Thus the rank increment is multiplied by the orbit
size only when all copies have the same rank. An exchange-stabilized graft
has $2b=M$. We describe the required blocks by their local stabilizers;
all displayed local lifts are understood up to a common scalar sign and
the inherited cyclic signs.

\emph{Fixed edge: $\langle R\rangle$, $\langle\mathcal D\rangle$, or
$\langle R,\mathcal D\rangle$.}
On a positive two-label line $(1,a)$, reflection swaps the two labels,
and both $R$ and $\mathcal D$ send $a$ to $a^{-1}$. Either condition
therefore forces $a=1$, giving $K_{1,2}=\operatorname{rowspan}(1,1)$.
If only $R\mathcal D$ is imposed, every $a>0$ is fixed: this is the free
Karpman axis-bridge case~\cite[Section~4]{Karpman2018}, not a forced rigid
pair. The detached $K_{1,2}$ insertion is the edge-increment:
\begin{equation}
\label{eq:atoms:edge-increment}
 \einc(C)=\operatorname{rowspan}
 \begin{pmatrix}1&0_{1\times N}&(-1)^k\\0_{k\times1}&C&0_{k\times1}\end{pmatrix},
 \qquad(\Delta N,\Delta k)=(2,1).
\end{equation}
The seam is placed between the two new labels. For a purely linear
dihedral action, the type-preserving insertions on $p$ disjoint fixed
edges have
\begin{equation}
\label{eq:atoms:dihedral-increments}
 \begin{array}{c|ccc}&\pre^\rho&\einc^\rho&\inc^\rho\\\hline
 \Delta N&2p&2p&2p\\\Delta k&0&p&2p\end{array}.
\end{equation}
These are not the only allowed changes of action.
\emph{Linearly fixed vertex: $\langle R\rangle$.}
There are two dual grafts, both written in local cyclic order $(a,b,*)$,
where $*$ is the attaching leg.

\emph{Parallel splitting.}
A fixed vertex with a nonzero column $u$ (that is, not a loop) is
replaced by two neighboring columns $u,u$. This is gluing with
\begin{equation*}
 K_{1,3}=\operatorname{rowspan}(1,1,1).
\end{equation*}
The reflection exchanges $a,b$ and fixes the attaching coordinate $*$.
Before normalization, the reflection-fixed positive block is the line
$(1,1,c)$, $c>0$; fixing the relative scale of the attaching coordinate
sets $c=1$. The inverse operation fuses the parallel pair, and the
per-copy change is $(\Delta N,\Delta k)=(1,0)$.

\emph{Series splitting.}
The positive-dual operation replaces a fixed vertex that is not a
coloop by a series pair. Its gluing block is
\begin{equation*}
 K_{2,3}=\operatorname{rowspan}
          \begin{pmatrix}1&0&-1\\0&1&1\end{pmatrix},
\end{equation*}
the positive dual of $K_{1,3}$. 
Explicitly, place the attaching column last, using the signed cyclic
convention, and write $C=(D\mid u)$, with $\operatorname{rank}D=k$.
Series splitting is represented by
\begin{equation*}
 (D\mid u)\longmapsto
 \begin{pmatrix}
 D&u&0\\
 0_{1\times(N-1)}&1&1
 \end{pmatrix},
\end{equation*}
where the last two columns replace $u$.
Its inverse quotients out the supported line
$\mathbb R(0,\ldots,0,1,1)$ and replaces the two final coordinates
$x_a,x_b$ by $x_a-x_b$. 
We use the connector normalization dual to that of parallel splitting.
The per-copy change
is $(\Delta N,\Delta k)=(1,1)$.

Neither normalized graft introduces a parameter. An exchanging element
can transport a parallel graft to a series graft at a different vertex,
but cannot stabilize either three-leg block individually, since it
interchanges their ranks $1$ and $2$.

\emph{Vertex fixed by reflection and exchange separately:
$\langle R,\mathcal D\rangle$.}
Here the rigid block is $K_{2,4}$. In cyclic order $(a,b,c,*)$, one
representative and its signed reflection are
\begin{equation*}
 \begin{gathered}
 K_{2,4}=\operatorname{rowspan}
 \begin{pmatrix}1&0&-1&-\sqrt2\\1&\sqrt2&1&0\end{pmatrix},\\
 \widehat R e_a=e_c,\quad \widehat R e_c=e_a,\quad
 \widehat R e_b=e_b,\quad \widehat R e_*=-e_*.
 \end{gathered}
\end{equation*}
This plane is fixed by $\widehat R$ and is isotropic for
$A_4=\operatorname{diag}(1,-1,1,-1)$, hence is also fixed by
$\mathcal D$. Its graft replaces the old fixed vertex by $a,b,c$:
reflection fixes $b$ and exchanges $a,c$. The per-copy change is $(\Delta N,\Delta k)=(2,1)$. 
For an explicit gluing formula, put the old attaching column last and
write $C'=(D\mid u)$ of rank $q$. Set
\begin{equation}
\label{eq:atoms:vertex-graft}
 \mathsf V_*(C')=
 \operatorname{rowspan}\begin{pmatrix}
 0&(-1)^q&(-1)^q\sqrt2&(-1)^q\\
 D&u/\sqrt2&0&-u/\sqrt2
 \end{pmatrix}.
\end{equation}
The three new columns replace $u$. A maximal minor choosing one new column
is a lower minor not using $u$, multiplied by $1,\sqrt2,1$; one choosing two
is a lower minor using $u$, multiplied by $1,\sqrt2,1$ for the first two, first and last, or last two new columns, respectively. Minors choosing zero or three vanish. Hence positivity and the
positroid correspondence are exact. The supported triple line is
$\mathbb R(1,\sqrt2,1)$. The inverse sends the new coordinates to the
old attaching coordinate by $(x_a,x_b,x_c)\mapsto(x_a-x_c)/\sqrt2$,
quotienting out that supported line and leaving the other coordinates
unchanged. Reversal fixes the supported line and acts by $-1$ on the
quotient coordinate, as on $*$ in the displayed block. For the local
alternating form the supported line is isotropic and orthogonal to
$(1,0,-1)/\sqrt2$, whose squared norm is that of the old coordinate.
Thus both reflection and orthogonality pass to the quotient. This is the
signed one-particle gluing with $K_{2,4}$.

\paragraph{Continuous operations.}
A continuous operation acts by $C\mapsto CM(t)$, where
$M(t)=\exp(tX)\in\GL_N$ and $X$ is its infinitesimal generator.
Transport by a group element acts on these matrices by
\begin{equation}
\label{eq:atoms:bridge-action}
 \mathcal A_\gamma(M)=
 \begin{cases}T_\gamma MT_\gamma^{-1},&\epsilon=0,\\
 T_\gamma M^{-T}T_\gamma^{-1},&\epsilon=1,
 \end{cases}
 \qquad
 a_\gamma(X)=
 \begin{cases}T_\gamma XT_\gamma^{-1},&\epsilon=0,\\
 -T_\gamma X^TT_\gamma^{-1},&\epsilon=1.
 \end{cases}
\end{equation}
Equivariance means $\mathcal A_\gamma(M(t))=M(t)$ for every $\gamma$,
equivalently $a_\gamma(X)=X$. We retain the ordinary bridge notation
$x_i,y_i$ from
\eqref{eq:bg:ordinary-atoms}, now including the seam:
\begin{equation*}
 x_i(t)=I+t\sigma_i e_{i,i+1},\qquad
 y_i(t)=I+t\sigma_i e_{i+1,i},\qquad
 \sigma_i=\begin{cases}1,&i<N,\\(-1)^{k-1},&i=N.\end{cases}
\end{equation*}
Here $i+1$ is cyclic. We will sometimes call a signed adjacent generator a \emph{bridge} or an \emph{arrow}, drawing
the generators of $x_i$ and $y_i$ as $i\to i+1$ and
$i+1\to i$, respectively, with the seam signs understood.
An arrow crosses the cut $I\mid I^c$ if exactly one of its
endpoints belongs to $I$. If $\mathcal O$ is a full orbit of distinct signed
directed adjacent generators $\sigma_i e_{i,i+1}$ or
$\sigma_i e_{i+1,i}$, define
\begin{equation}
\label{eq:atoms:orbit-exponential}
 X_{\mathcal O}=\sum_{X\in\mathcal O}X,\qquad
 M_{\mathcal O}(t)=\exp(tX_{\mathcal O}),\qquad t\geq0.
\end{equation}
Each generator occurs once. When the generators commute, we call their product with the
common parameter $t$ an \emph{ordinary bridge packet}:
\begin{equation}
\label{eq:atoms:bridge-packet}
 \mathsf B^\rho_{\mathcal O}(t)=\prod_{X\in\mathcal O}(I+tX).
\end{equation}
For distinct arrows, the commutation criterion follows from
$[e_{ab},e_{cd}]=\delta_{bc}e_{ad}-\delta_{da}e_{cb}$:
disjoint arrows, arrows with a common tail, and arrows with a common head
commute; a directed two-step path and opposite arrows on one edge do not.
For example, $x_i$ and $y_{i-1}$ commute
when $N>2$, whereas $x_i$ and $y_i$ do not.
Unequal orbit weights or arbitrarily ordered noncommuting products do
not satisfy the definition of an orbit operation.

The two-column hyperbolic operations are
\begin{subequations}
\label{eq:atoms:hyperbolic-data}
\begin{align}
 X_i&=\sigma_i(e_{i,i+1}+e_{i+1,i}),\label{eq:atoms:hyperbolic-generator}\\
 \mathsf H_i(t)&=e^{tX_i}
 =\begin{pmatrix}\cosh t&\sigma_i\sinh t\\
 \sigma_i\sinh t&\cosh t\end{pmatrix}_{i,i+1}.
 \label{eq:atoms:orthogonal-layer}
\end{align}
\end{subequations}
For disjoint edge orbits their product is
\begin{equation}
\label{eq:atoms:hyperbolic-orbit}
 \mathsf H^\rho_{\mathcal O}(t)
 =\prod_{b_i\in\mathcal O}\mathsf H_i(t)
 =\exp\!\left(t\sum_{b_i\in\mathcal O}X_i\right).
\end{equation}
These generalize the orthogonal bridges~\cite[Section~4.2]{KimLee2014}.
At an edge fixed only by $R\mathcal D$, $x_i$ and
$y_i$ are individually
allowed. With $R$ or $\mathcal D$ separately imposed they are paired into
the hyperbolic operation.

At a vertex $b$ fixed by the local reflection, let $(a,b,c)$ be
consecutive in cyclic order, with the reflection exchanging $a,c$.
The local directed orbits give the following descriptions
(and their transposes):
\begin{equation}
\label{eq:atoms:vertex-generators}
 \begin{array}{c|c|l}
 \text{stabilizer}&\text{generator}&\text{operation}\\\hline
 \langle R\rangle&e_{ab}+e_{cb}&\text{commuting inward pair}\\
 \langle R\mathcal D\rangle&e_{ab}+e_{bc}&\text{unipotent palindrome}\\
 \langle R,\mathcal D\rangle&e_{ab}+e_{ba}+e_{bc}+e_{cb}&\text{three-vertex boost}.
 \end{array}
\end{equation}
In a non-seam chart, with $b=a+1$ and $c=a+2$, the last two are
\begin{subequations}
\label{eq:atoms:vertex-operations}
\begin{align}
 U(t)&=I+t(e_{ab}+e_{bc})+\tfrac12t^2e_{ac}
      =x_a(t/2)x_b(t)x_a(t/2),
 \label{eq:atoms:palindrome}\\
 V(\tau)&=I+\frac{\sinh\tau}{\sqrt2}X
                  +\frac{\cosh\tau-1}{2}X^2,\qquad
 X=e_{ab}+e_{ba}+e_{bc}+e_{cb},\quad X^3=2X.
 \label{eq:atoms:vertex-boost}
\end{align}
\end{subequations}
Both operations are local instances of
\eqref{eq:atoms:orbit-exponential}: with
$\mathcal O_{\mathrm{pal}}=\{e_{ab},e_{bc}\}$ and
$\mathcal O_{\mathrm{boost}}=\{e_{ab},e_{ba},e_{bc},e_{cb}\}$,
the respective local stabilizer orbits, we have
\[
 U(t)=M_{\mathcal O_{\mathrm{pal}}}(t),\qquad
 V(\tau)=M_{\mathcal O_{\mathrm{boost}}}(\tau/\sqrt2).
\]
For the global equivariant operation, take the full $G$-orbit of
the initial signed generator, including the inherited seam signs
and counting each generator once. 
Orbit translates are imposed together. For overlapping patches use the
full generator, not an arbitrarily ordered product of local blocks.

We call the orbit operations $M_{\mathcal O}(t)$ with nilpotent
generator $X_{\mathcal O}$ \emph{additive operations}, with parameter
$t>0$. The two-column and normalized three-vertex hyperbolic
operations, including their disjoint orbit products, are called
\emph{boosts}; for these we use the multiplicative parameter
$u=e^\tau>1$.
\begin{lemma}[Positive equivariant orbit operations]
\label{lem:atoms:hyperbolic}
The orbit exponentials~\eqref{eq:atoms:orbit-exponential} and hyperbolic
edge layers~\eqref{eq:atoms:hyperbolic-orbit} preserve $X^{\geq}$ and
$X^>$. The rigid grafts preserve nonnegativity and their stated induced
symmetries.
\end{lemma}
\begin{proof}
A full orbit sum, or a union of such orbits, is fixed by
\eqref{eq:atoms:bridge-action}, so its exponential is equivariant.
Let $\mathcal L_{\mathcal O}$ be the induced infinitesimal action on the
$k$th exterior power: for row vectors $v_1,\ldots,v_k$,
\begin{equation*}
 \mathcal L_{\mathcal O}(v_1\wedge\cdots\wedge v_k)
 =\sum_{j=1}^k v_1\wedge\cdots\wedge
       (v_jX_{\mathcal O})\wedge\cdots\wedge v_k.
\end{equation*}
Thus, for the column $p(C)=(\Delta_I(C))_I$ of ordered Pl\"ucker
coordinates, $p(CM_{\mathcal O}(t))=e^{t\mathcal L_{\mathcal O}}p(C)$.
In this basis $\mathcal L_{\mathcal O}$ has nonnegative entries: its
only nonzero entries are the adjacent basis exchanges, and the seam
sign cancels the ordering sign. Consequently
$e^{t\mathcal L_{\mathcal O}}$ is entrywise nonnegative and has positive
diagonal. This proves both positivity assertions. 
This argument applies uniformly to the ordinary bridge packets,
hyperbolic edge layers, vertex palindromes and three-vertex boosts,
including overlapping orbit configurations. 

For the grafts, the minor formulas above prove positivity; their defining
line or isotropic quotient proves equivariance. Transport these normalized
maps along the orbit. Exchanging elements transport a parallel graft to
its series dual, and the middle-rank graft to itself.
\end{proof}
As in the bare Grassmannian case, positive operation need not add a coordinate on every stratum. 
\subsection{Symmetric positroid cells}
\label{subsec:domain:cells}

Let $\Pi_f^>$ be the ordinary nonnegative positroid cell indexed by the
bounded affine permutation $f$, and set
\begin{equation}
\label{eq:cells:fixed-stratum}
 \Sigma_f^{G,\rho}=\Pi_f^>\cap\SymGrNonneg{G}{\rho}{k,N}.
\end{equation}
We suppress $G,\rho$ when the current action is clear. The following theorem supplies
positive coordinates and a symmetric atomic decomposition. 

\begin{theorem}[Equivariant cell reduction]
\label{thm:dom:cells}
For every symmetry datum as above, each nonempty
$\Sigma_f^{G,\rho}$ is homeomorphic to $(\mathbb R_{>0})^{d_f}$ for some
$d_f\geq0$. It has an inductive parametrization by the operations of
Section~\ref{subsec:domain:atoms}, starting from coordinate planes or cyclically symmetric positive
points. The reductions may change the action on the cyclic label set.
Each continuous reducing step contributes one parameter; normalized rigid
grafts and lollipop insertions contribute none.
\end{theorem}
\paragraph{Idea of the proof.}
The proof is an equivariant version of ordinary bridge reduction \cite[Lemmas~7.8--7.9, Proposition~7.10 and Theorem~7.12]{Lam2015}. The different types of atomic operations and grafts arise naturally from this reduction process.
Starting from a nonempty symmetric positroid cell, we first remove any
loop--coloop orbit. Otherwise Lam's bridge theorem supplies an ordinary
removable adjacent bridge; we replace it by its full symmetry orbit.
The possible directed orbit shapes are very restricted: apart from the
terminal cyclic cases, they are built from single edges, two-edge paths,
opposite-arrow pairs or three-vertex bidirectional patches. The key combinatorial point is that there is at most one strictly smaller symmetric positroid to which applying the symmetric operation yields the original symmetric positroid. Hence, whenever
backward motion reaches the boundary in finite time, every point of the
cell exits into the same smaller symmetric positroid cell. The exit time
is then a positive coordinate and gives a product decomposition
\[
   \Sigma_f\simeq\Sigma_{f_\downarrow}\times\mathbb R_{>0}.
\]
Thus one symmetric orbit operation removes one continuous cell
coordinate, even though it may consist of several ordinary bridges. 

The only extra issue is that for the hyperbolic two- and three-label
operations the backward flow can remain nonnegative for all time.
Positivity then forces the affected labels into one of a short list of
rigid local configurations: loops or coloops, a normalized parallel or
series pair, an isolated $K_{1,2}$ block, or the rigid three-label
self-dual flag. These patches can be contracted by the rigid graft
operations, reducing the number of labels rather than the cell
dimension. The same alternative holds uniformly on the whole positroid
stratum. Inducting on the pair
\[
   (\text{number of labels},\ \text{ordinary positroid dimension})
\]
therefore repeatedly either removes one positive parameter or contracts
a rigid patch, until only a zero-dimensional coordinate plane or cyclic
Karp point remains. Reversing the reductions gives the claimed atomic
parametrization and identifies every nonempty symmetric positroid
stratum with $(\mathbb R_{>0})^{d_f}$.
\subsubsection{Proof of Theorem~\ref{thm:dom:cells}.}
\paragraph{The possible orbit shapes.}
Draw the generator of $x_i$ as $i\to i+1$ and that
of $y_i$ as its reverse. Forget coefficients but keep both
arrows when both occur. Path length means the number of edges, not the
number of vertices; a \emph{directed} path must follow the arrows.
The following elementary classification also isolates the terminal cases. For an arrow orbit $\mathcal O$, its \emph{rightward} and \emph{leftward index sets} are $A=\{i\in[N]:(i\to i+1)\in\mathcal O\}$ and $B=\{i\in[N]:(i+1\to i)\in\mathcal O\}$, respectively, with indices read modulo $N$.

\begin{lemma}[Directed orbit classification]
\label{lem:cells:orbit-shapes}
For $N>2$, a $G$ orbit of a single arrow has one of the following forms:
\begin{enumerate}
\item an acyclic directed graph, with every directed path having at most
      two edges;
\item a disjoint union of opposite-arrow pairs $a\rightleftarrows b$
      (directed cycles of length two);
\item a disjoint union of bidirectional paths on three vertices, whose
local stabilizer contains $R$ and $\mathcal D$ separately;
\item the whole coherently directed cycle, or the whole bidirectional cycle.
\end{enumerate}
In the last case the nonnegative fixed locus is a single cyclically symmetric positive point.
In the other cases neither all rightward bridges nor all leftward bridges
occur: each orientation's index set is a proper subset of the cyclic
set $\{1,\ldots,N\}$.
\end{lemma}
\begin{proof}
For a graph subgroup, encode the direction by $d\in\{1,-1\}$. The two
label generators act on edge--direction pairs by
\begin{equation}
\label{eq:cells:arrow-arithmetic}
 r:(i,d)\mapsto(i+\ell,(-1)^\alpha d),\qquad
 s_c:(i,d)\mapsto(c-1-i,(-1)^{1+\beta}d),
\end{equation}
where $\alpha=\epsilon(r)$, $\beta=\epsilon(s_c)$. Pure exchange reverses
$d$; we omit the reflection in a cyclic group. For $\ell\geq3$ the underlying
edge set uses at most two residues modulo $\ell$, so its connected
components are single edges or two-edge paths. If both directions occur on
one edge they occur on its entire orbit. A bidirectional two-edge path
requires a vertex reflection together with pure exchange: a linear
reflection alone gives inward/outward arrows, and reflection with exchange
gives a directed two-step path. Thus the two-edge directed paths in the
acyclic case are the vertex palindromes of
\eqref{eq:atoms:palindrome}. They occur at exchange-twisted fixed vertices,
including those in reflection-twisted and mixed groups, rather than at
the fixed edges of the original Karpman model.

For $\ell=1,2$ the only additional underlying graph is the full cycle.
Formula~\eqref{eq:cells:arrow-arithmetic} gives the alternating source--sink
pattern, the repeated $\rightarrow\,\rightarrow\,\leftarrow\,\leftarrow$
pattern, a coherent orientation, or both
directions. The first two are acyclic. In either of the last two cases the
group contains a one-step linear rotation, or, at $N=2k$, a linear two-step
rotation and a vertex reflection with exchange. The former fixed locus is
the cyclically symmetric positive point~\cite[Theorem~1.1]{Karp2019}. The latter is the
reflection-twisted vertex model with quotient size two; its dimension in
\eqref{eq:dom:family-dimensions} is zero, so Theorem~\ref{thm:dom:ball}
again gives that single point. Additional symmetries do not change it.
Finally, a full rightward or leftward index set contains the coherent cycle,
so this occurs only in the terminal case.
\end{proof}

\paragraph{A uniform predecessor, including collisions.}
Let $A$ and $B$ be the rightward and leftward index sets of a nonterminal
orbit, and let $s_i$ exchange the adjacent positions $i,i+1$, with periodic
continuation across the seam. The subgroups
$W_A=\langle s_i:i\in A\rangle$ and $W_B=\langle s_j:j\in B\rangle$
permute the consecutive blocks joined by these adjacent exchanges.
Because both index sets are proper, these are finite groups (finite
parabolic subgroups of the affine symmetric group). A single adjacent exchange gives an $S_2$ factor. On a consecutive
triple, a vertex palindrome gives an $S_3$ factor on one side
(right or left, according to its orientation), while the
three-vertex boost gives an $S_3$ factor on each side.
A positive elementary bridge acts on the positroid index by the
corresponding adjacent transposition when this enlarges the positroid,
and leaves the index unchanged otherwise; these conditional
right and left actions are the elementary $0$-Hecke operations~\cite[Lemma~2.12]{Shevchenko2025}.

\begin{lemma}[Uniqueness of an invariant predecessor]
\label{lem:cells:two-predecessors}
Let $\mathcal O$ be the $G$-orbit of one arrow joining cyclically
consecutive labels, excluding the two full-cycle cases of
Lemma~\ref{lem:cells:orbit-shapes}.
Let $\Pi_f^>$ be a $G$-stable ordinary positroid cell, indexed by
the bounded affine permutation $f$, and suppose that
\[
 C M_{\mathcal O}(t)\in\Pi_f^>
 \qquad\text{for every }C\in\Pi_f^>\text{ and }t>0.
\]
Then there is at most one other $G$-stable ordinary positroid cell
$\Pi_g^>$ such that
\[
 C_0 M_{\mathcal O}(t)\in\Pi_f^>
 \qquad\text{for every }C_0\in\Pi_g^>\text{ and }t>0.
\]
If such a cell exists, denote its index by $f_\downarrow^{\mathcal O}$.
Its set of nonzero Pl\"ucker coordinates is a proper subset of
that of $\Pi_f^>$.
\end{lemma}
When $\mathcal O$ is clear from the context we may omit the superscript from $f_\downarrow^{\mathcal O}$ and just write $f_\downarrow.$
\begin{proof}
Let $\Pi_g^>$ be a $G$-stable positroid cell carried into $\Pi_f^>$
by $M_{\mathcal O}(t)$ for $t>0$. The identity
\[
 p(C_0M_{\mathcal O}(t))
   =e^{t\mathcal L_{\mathcal O}}p(C_0)
\]
shows exactly which minors become positive: they are those obtained
from an initially positive minor by finitely many adjacent label
replacements along arrows of $\mathcal O$. There is no cancellation,
since $\mathcal L_{\mathcal O}$ has nonnegative entries. Thus the same
final support can be obtained by applying the elementary bridges
from $\mathcal O$ successively until none enlarges the support.
Each change of index is right multiplication by some $s_i$, $i\in A$,
or left multiplication by some $s_j$, $j\in B$. Consequently,
\[
 g\in W_B f W_A.
\]

For any bounded affine permutation $q$ whose positroid cell $\Pi_q^>$
is preserved setwise by $G$, symmetry identifies the comparisons
associated with the arrows of $\mathcal O$: it permutes the corresponding
left and right adjacent exchanges and preserves whether they increase
or decrease the affine inversion number. Since $\mathcal O$ is one
$G$-orbit, all these comparisons have the same sign. Equivalently,
either
\[
 q(i)<q(i+1)\quad(i\in A),\qquad
 q^{-1}(j)<q^{-1}(j+1)\quad(j\in B),
\]
or all the displayed inequalities are reversed.

For each of these two choices, there is at most one such $q$ in
$W_B f W_A$. Indeed, $W_A$ reorders positions within its consecutive
blocks, while $W_B$ reorders values within its consecutive blocks.
The number of positions in each input block sent to each output block
is fixed throughout $W_B f W_A$. In the increasing case, these numbers
determine how the positions of each input block are assigned to the
ordered output blocks; the increasing order of $q^{-1}$ then determines
which value occupies each position. Reversing both orders gives the
unique decreasing choice. All blocks and comparisons are understood
periodically, including across the cyclic seam.

By hypothesis, applying $M_{\mathcal O}(t)$ to $\Pi_f^>$ leaves its
support unchanged. Hence none of the constituent positive elementary
bridges enlarges that support, so the elementary bridge rule identifies
$f$ with the increasing choice. Any distinct invariant predecessor must
therefore be the decreasing choice. If it occurs, denote it by
$f_\downarrow^{\mathcal O}$. Its support is a proper subset of the
support of $\Pi_f^>$, since positive bridge operations retain every
previously nonzero minor and distinct positroid cells have distinct
supports. This proves uniqueness; existence, when needed, will follow
from an actual finite first-exit point.
\end{proof}

\paragraph{First exit.}
Assume first that $\Pi_f^>$ admits an elementary bridge
$b(s)=I+sZ$, equal to $x_i(s)$ or $y_i(s)$, to which Lam's
bridge-removal theorem applies
\cite[Proposition~7.10]{Lam2015}.
Concretely, there is an ordinary positroid cell $\Pi_h^>$, depending
only on $f$ and $b$, such that for every $C\in\Pi_f^>$ there is
$a_b(C)\in(0,\infty)$ with
\[
 \begin{cases}
 Cb(-s)\in\Pi_f^>,&0\leq s<a_b(C),\\
 Cb(-a_b(C))\in\Pi_h^>,&
       \dim\Pi_h^>=\dim\Pi_f^>-1,\\
 Cb(-s)\notin\Gr_{\geq}(k,N),&s>a_b(C).
 \end{cases}
\]
Thus the bridge acts nontrivially throughout $\Pi_f^>$ and permits
a positive amount of removal without changing the support.

Let $\mathcal O$ be the $G$-orbit of $Z$, and write
$M(t)=M_{\mathcal O}(t)$.
Since $G$ preserves $\Pi_f^>$ setwise, every translated elementary
bridge has the same removal properties, although its ordinary
endpoint cell may differ. Their simultaneous stopping time need not
equal any individual stopping time, and may be infinite.
For $C\in\Sigma_f$, define its \emph{first-exit time} by
\begin{equation}
\label{eq:cells:first-exit}
 t_*(C)=\sup\{t\geq0:CM(-t)\in\Gr_{\geq}(k,N)\}.
\end{equation}
Let $\mathcal S_f$ be the set of nonzero Pl\"ucker coordinates on
$\Pi_f^>$. For every elementary bridge in $\mathcal O$, small removal
preserves this support. Its induced Pl\"ucker operator has nonnegative
entries, so it has no entry from a coordinate in $\mathcal S_f$
to one outside $\mathcal S_f$: otherwise an initially zero minor
would become negative immediately under removal.
The coordinate subspace supported on $\mathcal S_f$ is therefore
invariant under $\mathcal L_{\mathcal O}$ and its exponentials.
Consequently, positive motion preserves $\Pi_f^>$, and sufficiently
small backward motion does so too: zero minors remain zero, and
positive minors remain positive. In particular, $t_*(C)>0$.

The admissible times in which $CM(-t)$ is nonnegative form a closed interval: closedness follows from
that of the nonnegative Grassmannian, and if $t$ is admissible then
every $0\leq s\leq t$ is admissible, since
\[
 CM(-s)=CM(-t)M(t-s)
\]
and $M(t-s)$ preserves nonnegativity.

Suppose now that $t_*(C)<\infty$ for every $C\in\Sigma_f$.
The endpoint
\[
 C_{\mathrm{end}}=CM(-t_*(C))
\]
is nonnegative by closedness. It cannot belong to $\Sigma_f$:
small removal would then remain possible, contradicting maximality
of $t_*(C)$. Since $C=C_{\mathrm{end}}M(t_*(C))$,
Lemma~\ref{lem:cells:two-predecessors} places every such endpoint
in the same $\Sigma_{f_\downarrow}$, where
$f_\downarrow=f_\downarrow^{\mathcal O}$.
We obtain
\begin{equation}
\label{eq:cells:product-chart}
 \Sigma_{f_\downarrow}\times\mathbb R_{>0}\ \simeq\ \Sigma_f,
 \qquad(C_0,t)\longmapsto C_0M(t).
\end{equation}
Indeed, for $C_0\in\Sigma_{f_\downarrow}$ and $t>0$, the defining
property of the predecessor gives $C=C_0M(t)\in\Sigma_f$.
Removing $t$ recovers $C_0$, so $t_*(C)\geq t$.
If removal could continue to $t+\varepsilon$, $\varepsilon>0$, then
\[
 D:=CM(-t-\varepsilon)=C_0M(-\varepsilon)
\]
would be nonnegative, with
\[
 C_0=DM(\varepsilon),\qquad C=DM(t+\varepsilon).
\]
But the support of $DM(s)$ is independent of $s>0$, so these planes
would have the same support, contrary to $f_\downarrow\ne f$.
Hence $t_*(C)=t$, and the inverse of
\eqref{eq:cells:product-chart} is
\[
 C\longmapsto\bigl(CM(-t_*(C)),\,t_*(C)\bigr).
\]
Before the stopping time the support is still $f$, by the same
independence of forward support from the positive parameter.
Thus a time below $t_*(C)$ remains admissible near $C$, since all
its supported minors remain positive. At a time above $t_*(C)$,
a negative transported minor persists near $C$, so that time remains
inadmissible. These two observations prove continuity of $t_*$
and hence of the displayed inverse. Consequently,
\[
 \dim\Sigma_f=\dim\Sigma_{f_\downarrow^{\mathcal O}}+1.
\]
This dimension drop concerns the symmetric strata; the ordinary
positroid dimensions need not differ by one for the full orbit.

For an acyclic orbit of Lemma~\ref{lem:cells:orbit-shapes},
the first-exit time is finite: Order the labels so that every arrow
goes forward. This makes $X_{\mathcal O}$ triangular with zero
diagonal. Each allowed replacement in a Pl\"ucker index strictly
increases the sum of the positions of its selected labels in this
ordering, so $\mathcal L_{\mathcal O}$ is nilpotent as well.
In particular, it cannot generate nontrivial scaling of a fixed
Pl\"ucker line: its only eigenvalue is zero.

Fix a nonnegative Pl\"ucker vector $w$ for $C$ and transport it as
\[
 w(t)=e^{-t\mathcal L_{\mathcal O}}w,
\]
without time-dependent normalization.
Suppose removal remained nonnegative for all $t\geq0$.
Continuity keeps this nonzero lift in the nonnegative cone, so
$w(t)\geq0$. Since $\mathcal L_{\mathcal O}$ is entrywise nonnegative,
$e^{t\mathcal L_{\mathcal O}}\geq I$ entrywise, and therefore
\[
 w=e^{t\mathcal L_{\mathcal O}}w(t)\geq w(t)\geq0.
\]
Nilpotence makes each coordinate of $w(t)$ a polynomial in $t$;
the displayed bounds force each polynomial to be constant.
Hence $\mathcal L_{\mathcal O}w=0$.

This contradicts the nontriviality supplied by Lam's removal theorem.
For the initial elementary bridge, with induced operator
$\mathcal L_Z$, removal acts by
\[
 p(Cb(-s))=w-s\mathcal L_Zw.
\]
At $s=a_b(C)$ an initially positive minor vanishes, so
$\mathcal L_Zw\ne0$. Since every orbit summand is entrywise
nonnegative,
\[
 \mathcal L_{\mathcal O}w
   =\sum_{Z'\in\mathcal O}\mathcal L_{Z'}w
   \geq\mathcal L_Zw\ne0,
\]
a contradiction. Thus $t_*(C)<\infty$.
This includes the palindrome and the overlapping
$\rightarrow\,\rightarrow\,\leftarrow\,\leftarrow$ orbits,
without requiring commuting generators.
For boosts, nonzero eigenvalues can instead rescale a Pl\"ucker
vector without moving its plane; their possible infinite backward
motion is treated by the rigid-contraction argument below.
\paragraph{Infinite exit forces a rigid contraction.}
It remains to treat disjoint two- and three-vertex boosts. In this case we cannot rely on nilpotence, and need to preform a finer spectral study of the dynamics. The following
linear-algebra observation prevents an unexplained terminal stratum.
\begin{lemma}[Positive backward orbits on a small path]
\label{lem:cells:backward-path}
Let $J=(j_1,\ldots,j_m)$, $m\in\{2,3\}$, be consecutive labels
in a cyclic chart where $J$ does not cross the seam. Set
\[
 X_J=\sum_{\nu=1}^{m-1}
       (e_{j_\nu,j_{\nu+1}}+e_{j_{\nu+1},j_\nu}),
\]
with all other entries zero. Suppose
$Ce^{-tX_J}\in\Gr_{\geq}(k,N)$ for every $t\geq0$. Put
\begin{equation}
\label{eq:cells:patch-ranks}
 a=\dim(C\cap\mathbb R^J),\qquad b=\dim\pi_J C.
\end{equation}
Then $b-a\leq1$. The supported and projected planes are spanned by
the eigenvectors of $X_J|_{\mathbb R^J}$ with the $a$ and $b$
largest eigenvalues, respectively.
\end{lemma}
\begin{proof}
For $0\leq d\leq m$, fix a set $T\subseteq J^c$ of $k-d$ labels
and collect the minors with exactly these labels outside $J$:
\[
 w_{T,d}=\bigl(\Delta_{T\cup I}(C)\bigr)_{I\in\binom Jd}.
\]
Under backward motion this vector evolves as
$e^{-tX_J^{[d]}}w_{T,d}$, where $X_J^{[d]}$ is the induced
infinitesimal action on $\bigwedge^d\mathbb R^J$.
For $0<d<m$, this matrix is symmetric with nonnegative off-diagonal
entries and induces a connected graph: adjacent moves connect all $d$-subsets
of the path. Its top eigenline is therefore strictly positive.
Every nonzero $w_{T,d}$ must lie on that line. Otherwise its component
with smallest eigenvalue would dominate under backward motion,
giving, after positive rescaling, a nonzero nonnegative eigenvector
orthogonal to the strictly positive top vector, which is impossible.
The cases $d=0,m$ are one-dimensional.

If $a<b$, choose a basis $u_1,\ldots,u_a$ of
$C\cap\mathbb R^J$ and complete it by $v_1,\ldots,v_{k-a}$
to a basis of $C$. The outside projections
$z_i=\pi_{J^c}v_i$ form a basis of $\pi_{J^c}C$.
The wedges obtained by omitting one $z_i$ form a basis of
$\bigwedge^{k-a-1}\pi_{J^c}C$, so their coordinate matrix,
with rows indexed by $T$, has full column rank.
Expanding the wedge of the chosen basis of $C$ therefore shows
that the vectors $w_{T,a+1}$ span precisely the image of
\[
 \pi_JC/(C\cap\mathbb R^J)
 \longrightarrow \bigwedge^{a+1}\mathbb R^J,\qquad
 [v]\longmapsto u_1\wedge\cdots\wedge u_a\wedge v.
\]
This map is well defined and injective, since wedging with
$u_1\wedge\cdots\wedge u_a$ annihilates exactly their span.
Its image consequently has dimension $b-a$.
All $w_{T,a+1}$ lie on the same top eigenline, so $b-a\leq1$.

Finally, the nonzero vectors $w_{T,a}$ and $w_{T,b}$ represent
the top exterior powers of $C\cap\mathbb R^J$ and $\pi_JC$,
respectively. Their top-eigenline property identifies these planes
with the corresponding spans of largest-eigenvalue eigenvectors.
This also covers $a=b$, when only one inside degree occurs.
\end{proof}
We now apply the lemma to show that infinite backward motion forces
a rigid local configuration, so the current stratum is obtained from
one on fewer labels by a normalized rigid graft or a lollipop insertion.

Write the disjoint boost operation as $M(t)=\prod_\nu M_\nu(t)$,
where $M_\nu$ acts on the patch $J_\nu$. If $t_*(C)=\infty$, then
\[
 CM_\nu(-t)
   =\bigl(CM(-t)\bigr)\prod_{\mu\ne\nu}M_\mu(t)
   \in\Gr_{\geq}(k,N)\qquad(t\geq0),
\]
since the patches are disjoint and each positive boost preserves
nonnegativity. Thus Lemma~\ref{lem:cells:backward-path} applies
to each patch separately.

Consider a two-label patch $J=(i,j)$, placed last in a signed cyclic
chart. Its local generator
$\left(\begin{smallmatrix}0&1\\1&0\end{smallmatrix}\right)$ has
eigenvalues $1,-1$, with top eigenline
$L_+=\mathbb R(e_i+e_j)$.
With $a,b$ as in \eqref{eq:cells:patch-ranks}, the lemma gives
$0\leq a\leq b\leq2$ and $b-a\leq1$, leaving precisely
\[
\begin{array}{c|l}
(a,b)&\text{configuration on }J\\ \hline
(0,0)&C_i=C_j=0:\ \text{two loops},\\
(2,2)&\mathbb R^J\subset C:\ \text{two coloops},\\
(1,1)&C=(C\cap\mathbb R^{J^c})\oplus L_+:
          \ \text{an isolated }K_{1,2},\\
(0,1)&\pi_JC=L_+:\ \text{two equal nonzero columns (parallel pair)},\\
(1,2)&C\cap\mathbb R^J=L_+:
          \ \text{a series pair}.
\end{array}
\]
The first three cases are reduced by deleting the two labels,
that is, projecting onto $J^c$; the rank drops by $0,2,1$,
respectively. For a parallel pair, retain just one of the equal
columns. For a series pair, replace the coordinates $x_i,x_j$
by $x_i-x_j$, leaving the other coordinates unchanged.
This map has kernel $L_+$ on $C$, so it reduces the rank by one.
These last two operations, called parallel and series fusion,
are the inverses of the corresponding splittings in
Section~\ref{subsec:domain:atoms}.

At a linearly fixed edge, reflection also forces these equal weights:
it exchanges the two proportional columns in the parallel case,
or the two coefficients of the supported line in the series case,
sending their positive ratio to its reciprocal.
If pure exchange is imposed, it sends the rank pair $(a,b)$ to
$(2-b,2-a)$; hence $a=2-b$, and only the isolated-pair case
$(a,b)=(1,1)$ remains. All reductions are performed on the full
orbit, with the inherited cyclic signs.

For a three-label boost the stabilizer contains linear reflection and
pure exchange. Self-duality gives $b=3-a$, so $b-a\leq1$ forces
$(a,b)=(1,2)$. Reflection and positivity make the two outer entries of
the supported line equal, and isotropy fixes its middle entry:
\begin{equation}
\label{eq:cells:rigid-triple}
 C\cap\mathbb R^J=\mathbb R(1,\sqrt2,1),\qquad
 \pi_JC=\operatorname{span}\{(1,\sqrt2,1),(1,0,-1)\}.
\end{equation}
The inverse of~\eqref{eq:atoms:vertex-graft} contracts this triple to one
vertex. It decreases $(N,k)$ by $(2,1)$ per patch, preserving middle rank.

Importantly, this alternative is uniform on the positroid stratum. The
ranks~\eqref{eq:cells:patch-ranks} and the supports of the pair/triple
factors depend only on $f$. Once an infinite-exit point exists, these ranks,
the local symmetry and positivity force the same normalized small flags
at every point of $\Sigma_f$. Pair fusion or triple contraction therefore
gives a homeomorphism of the entire $\Sigma_f$ with one induced fixed
positroid stratum on fewer labels. The inverse minor formulas for the
rigid grafts show both the fixed support and the continuous inverse.
Symmetry-related patches are contracted together, with exchange
transporting parallel and series choices into one another. To identify the induced symmetry on the smaller label set, choose
unit representatives of the surviving one-dimensional quotients,
with compatible dual choices for the complementary plane.
The signed ambient symmetries then carry each surviving coordinate
vector to another up to sign; exchanging elements additionally take
orthogonal complements. The inverse graft formulas give the
ordered-minor signs for the reduced rank, identifying these maps
with the standard lifts on the contracted cyclic labels, up to
a common scalar sign for each group element. If this alternative
does not occur, all exits are finite and
\eqref{eq:cells:product-chart} applies.

\begin{proof}[Proof of Theorem~\ref{thm:dom:cells}]
Induct lexicographically on the number of labels and the dimension of the
ordinary positroid cell. Delete any present lollipop orbit first. Otherwise
Lam supplies a removable adjacent bridge; use its complete directed orbit.
The terminal-cycle case of Lemma~\ref{lem:cells:orbit-shapes} is a point.
Every other case either has the product reduction
\eqref{eq:cells:product-chart}, decreasing ordinary cell dimension, or a
rigid contraction decreasing the number of labels. The latter uses the
new induced action; it need not preserve the old vertex/edge type.
For $k=0$ or $k=N$, the Grassmannian itself is a point.
On at most two labels, the only remaining case is $(N,k)=(2,1)$:
the ordinary positroid strata are the two coordinate lines and
$\{\operatorname{rowspan}(1,t):t>0\}$.
Each symmetry acts on the positive ratio by $t\mapsto t$ or
$t\mapsto t^{-1}$, so its nonempty fixed strata are points or
this entire positive interval. The latter is parametrized by
$(1,0)x_1(t)=(1,t)$; when symmetry forces $t=1$, the stratum
is the rigid point $K_{1,2}$ instead. The induction terminates and yields
\begin{equation}
\label{eq:cells:iterated-chart}
 \Sigma_f\simeq\Sigma_{\rm seed}\times(\mathbb R_{>0})^{d_f},
 \qquad\Sigma_{\rm seed}\text{ a point}.
\end{equation}
All reductions start from a nonempty stratum and construct an actual
smaller point. Thus no realizability theorem for abstract invariant
positroids was used.
\end{proof}

\begin{remark}[Comparison with previous constructions]
    The cyclic specialization recovers Fraser's cell parametrizations
\cite[Theorem~5.1 and Section~5.3]{Fraser2020}; the half-turn exchange
specialization agrees with Shevchenko's~\cite[Theorem~5.1]{Shevchenko2025}.
The construction does not require an ordinarily reduced symmetric plabic
graph: injectivity follows from recovery of each orbit parameter.
Shevchenko's reducedness qualification~\cite[Proposition~5.19]{Shevchenko2025}
is therefore compatible with the theorem. 
\end{remark}
We call the nonempty sets
\eqref{eq:cells:fixed-stratum} \emph{symmetric positroid cells}; arbitrary
unions of them and their nonnegative factorization loci are not thereby
single cells.
\subsection{Examples}
\label{subsec:domain:blobs}
The ordinary Grassmannian,
orthogonal Grassmannian and reflected Lagrangian domain were recalled in
Sections~\ref{subsec:bg:grassmannian}, \ref{subsec:bg:orthogonal}
and~\ref{subsec:bg:reflected}; here we describe the remaining reflection
variants, focusing on their dimensions, representation types and elementary
positive operations. Each family is defined on its own cyclically ordered
label set. We write $\widehat s$ for the chosen signed reflection and
$\mathcal D$ for the positive duality~\eqref{eq:sym:positive-duality}.

For all these families Theorem~\ref{thm:dom:ball} supplies a closed ball,
with positive interior and constant representation type, including at the
nonnegative boundary. As usual, the product descriptions below concern the distinguished complex
component containing $K_{k,N}$, not Cartesian products of standard nonnegative parts.

\begin{example}[Linear reflection]
\label{ex:domain:linear-reflection}
On $M$ cyclically ordered labels, with rank $r$, set
\begin{equation}
\label{eq:blob:linear-source}
 \LinRefDomain^{\geq}_{r,M}(s)
 :=\{C\in\Gr_{\geq}(r,M):\widehat s C=C\}.
\end{equation}
This is \eqref{eq:dom:definition} for $\langle(s,0)\rangle$.
Let $V_\pm=\ker(\widehat s\mp I)$, $m_\pm=\dim V_\pm$, and
$\epsilon_r=r\bmod2$. For the standard lifts,
\eqref{eq:dom:self-paired-types}--\eqref{eq:dom:reflection-trace} give
\begin{equation}
\label{eq:blob:reflection-types}
\renewcommand{\arraystretch}{1.2}
\begin{array}{c|c|c|c}
 \text{type}&\widehat s&(m_+,m_-)&\dim\LinRefDomain^{\geq}_{r,M}(s)\\\hline
 M=2m,\ \text{gap axis}&J_M&(m,m)&
       \lfloor r(M-r)/2\rfloor\\
 M=2m,\ \text{vertex axis}&S_rJ_M&(m+\epsilon_r,m-\epsilon_r)&
       \lceil r(M-r)/2\rceil\\
 M=2m+1&J_M&(m+1,m)&r(M-r)/2
\end{array}
\end{equation}
Here $J_Me_i=e_{M+1-i}$ and $S_r$ is the signed shift on these labels. In every row,
\begin{equation}
\label{eq:blob:reflection-representations}
 \begin{gathered}
 C=C_+\oplus C_-,\quad C_\pm=C\cap V_\pm,\qquad
 c_+:=\dim C_+=\lceil r/2\rceil,\quad
 c_-:=\dim C_-=\lfloor r/2\rfloor,\\
 C\simeq\mathbf1_+^{\oplus c_+}\oplus\mathbf1_-^{\oplus c_-},\qquad
 C^\perp\simeq\mathbf1_+^{\oplus(m_+-c_+)}
                  \oplus\mathbf1_-^{\oplus(m_--c_-)},\\
 (\LinRefDomain_0)_{\mathbb C}\simeq
 \Gr(c_+,m_+)\times\Gr(c_-,m_-).
 \end{gathered}
\end{equation}
The characters refer to the \emph{signed} lift; central $-I$ acts by $-1$
on each summand. An overall change of lift interchanges both sign labels.
The signs are determined by the chosen lift, not merely by the number
of fixed labels. For example, the two vertex labels need not have the same
sign when $r$ is even.

Normalized minors satisfy $\Delta_I=\Delta_{s(I)}$, and the strata are the
nonempty intersections with the corresponding reflection-stable positroids.
The allowed insertions pair loops with loops and coloops with coloops;
individual fixed labels may also be loops or coloops. Use the rigid edge
increment at fixed gaps, ordinary admissible reflected packets, and the
hyperbolic edge layers of Lemma~\ref{lem:atoms:hyperbolic}. In particular,
the gap-axis case has two singleton hyperbolic edge orbits, and the odd case
has one. Theorem~\ref{thm:dom:cells} supplies the cell parametrizations, allowing
action-changing parallel/series grafts. Rank-one and corank-one cases are
included.
\end{example}

\begin{example}[Vertex-reflection exchange]
\label{ex:domain:exchange-vertex}
Let $M=2n$ and let $s$ have a vertex axis, with
$\widehat s=S_nJ_{2n}$. Define
\begin{equation}
\label{eq:blob:exchange-vertex}
 \begin{gathered}
 \ExRefDomain^{\geq}_{n,2n}(s)
   :=\{C\in\Gr_{\geq}(n,2n):C=\widehat s\mathcal D(C)\}
    =\{C\in\Gr_{\geq}(n,2n):CB_sC^T=0\},\\
 B_s=\widehat s A_{2n},\qquad B_s^T=B_s,\quad B_s^2=I,\qquad
 \dim\ExRefDomain^{\geq}_{n,2n}(s)=\frac{n(n-1)}2.
 \end{gathered}
\end{equation}
The symmetric form $B_s$ has signature $(n,n)$: its trace is zero and it
is an involution. Thus the ambient component is a maximal orthogonal
Grassmannian. In general this is \emph{not} the usual ABJM positivity embedding:
the defining minor equalities are $\Delta_I=\Delta_{s(I^c)}$, rather than
$\Delta_I=\Delta_{I^c}$, and $\widehat s C=C$ is not imposed separately.
The signed linear kernel acts only by scalars on $C$; $B_s$ exchanges
$C$ and $C^\perp$, and has two ambient eigenspaces of dimension $n$.
This specifies the graded representation, as in
\eqref{eq:dom:exchange-partners}--\eqref{eq:dom:exchange-form}.

Nonfixed reflection pairs allow loop--coloop insertions. Neither fixed
vertex permits a lollipop, by~\eqref{eq:atoms:lollipop-rule}. Use the exchange-vertex palindrome~\eqref{eq:atoms:palindrome}, together
with the equivariant insertions and orbit operations of
Section~\ref{subsec:domain:atoms}.
For comparison, replacing the vertex axis by a gap axis makes
$B_s^T=-B_s$ and $B_s^2=-I$: this is the Karpman domain already described
in Section~\ref{subsec:bg:reflected}, of dimension $n(n+1)/2$.
\end{example}

\begin{example}[Reflection-fixed orthogonal domains]
\label{ex:domain:reflection-orthogonal}
Impose pure exchange and linear reflection \emph{independently}:
\begin{equation}
\label{eq:blob:orthogonal-reflection}
 \OrthRefDomain^{\geq}_{n,2n}(s)
 :=\{C\in\OG_{\geq}(n,2n):\widehat s C=C\}
  =\{C\in\Gr_{\geq}(n,2n):CA_{2n}C^T=0,\ \widehat s C=C\}.
\end{equation}
These are reflection-fixed ABJM domains. Their normalized minors satisfy
$\Delta_I=\Delta_{I^c}=\Delta_{s(I)}$; hence their strata lie in
reflection-invariant orthogonal positroids, whose ordinary ABJM
combinatorics is recalled in~\cite{KimLee2014}. Put
$a=\lceil n/2\rceil$, $b=\lfloor n/2\rfloor$. The representations
of $C$ and $C^\perp$ are given by
\eqref{eq:blob:reflection-representations} at $r=n$, with dimensions and
ambient components
\begin{equation}
\label{eq:blob:orthogonal-reflection-types}
\renewcommand{\arraystretch}{1.25}
\begin{array}{c|c|c|c}
 \text{axis}&A_{2n}(V_\pm)&(\OrthRefDomain_0)_{\mathbb C}&
      \dim\OrthRefDomain^{\geq}_{n,2n}(s)\\\hline
 \text{gap}&V_\mp&\Gr(a,n)&\lfloor n^2/4\rfloor\\
 \text{vertex}&V_\pm&
       \OG(a,2a)^{\circ}\times\OG(b,2b)^{\circ}&
       \lfloor(n-1)^2/4\rfloor
\end{array}
\end{equation}
Here $\OG^\circ$ means the component containing the cyclically symmetric positive plane, and
$\OG(0,0)$ is a point. In the gap case $A_{2n}$ pairs $V_+$ and $V_-$,
so $C_-$ is the annihilator of $C_+$ for this pairing. In the vertex case
its restrictions to $V_\pm$ are split nondegenerate forms, and $C_\pm$
are maximal isotropic independently. This proves the two component
formulas and also records how exchange acts on the reflection types.
Pure exchange forbids all lollipop insertions. Instead use symmetric rank-one pair insertions, the rigid
$K_{2,4}$ vertex graft~\eqref{eq:atoms:vertex-graft}, and the edge and vertex
boosts. These extend the ordinary orthogonal
operations~\cite[Section~4.2]{KimLee2014}.
\end{example}

\section{External data and the symmetric amplituhedron}
\label{sec:external-and-map}
This section is devoted to constructing the symmetric amplituhedra and their basic properties.

Fix $2\leq k\leq N-2$ and a symmetry datum $(G,\rho)$. We retain the
signed action $T_\gamma$ of \eqref{eq:sym:signed-map}, its finite lift
$\widehat G$, and its linear kernel $\widehat H=\ker\epsilon$.
Write $X^{\geq}=\SymGrNonneg{G}{\rho}{k,N}$ and
$X^>=\SymGrPos{G}{\rho}{k,N}$.

\subsection{Positive equivariant external data}
\label{subsec:external:datum}

An \emph{external datum} consists of full-column-rank matrices
$\Lambda\in\Mat_{N,N-k+2}$ and
$\widetilde\Lambda\in\Mat_{N,k+2}$ whose column spaces satisfy
\begin{subequations}
\label{eq:ext:requirements}
\begin{align}
 &\mathcal W:=\operatorname{colspan}\widetilde\Lambda\in\Gr_{>}(k+2,N),\quad
 \mathcal U:=(\operatorname{colspan}\Lambda)^\perp
       \in\Gr_{>}(k-2,N),\qquad \mathcal U\subset\mathcal W
 \label{eq:ext:positive-flag}\\
 &h\mathcal U=\mathcal U,\quad h\mathcal W=\mathcal W
          \quad(h\in\widehat H),\qquad
 a\mathcal U=\mathcal W^\perp,\quad a\mathcal W=\mathcal U^\perp
          \quad(a\in\widehat G\setminus\widehat H).
 \label{eq:ext:equivariance}
\end{align}
\end{subequations}
Here and below orthogonality is Euclidean. Thus the nesting condition is
precisely $\Lambda^\perp\subset\widetilde\Lambda$ when matrices denote
their column spaces. Symmetry is imposed on these spaces, not on a
particular choice of frames. An exchanging element interchanges the
two external spaces, as in \eqref{eq:sym:pair-action}; it is present
only at $N=2k$. We use the capital-letter convention of
\eqref{eq:bg:external-data}. We recall again that nesting is an additional specialization of
the ordinary momentum datum of~\cite{DamgaardEtAl2019}. None of the
results in this section requires immanant positivity.

\paragraph{External representation types.}
Positivity fixes the $\widehat H$-types of $\mathcal U$ and $\mathcal W$:
they are those of the cyclically symmetric positive planes $K_{k-2,N}$ and
$K_{k+2,N}$, respectively,
with the \emph{same ambient signed action} as at rank $k$.
This follows from Theorem~\ref{thm:dom:ball} applied to the linear kernel;
in the cyclic case it is Fraser's distinguished-component statement
\cite[Definition/Lemma~4.8]{Fraser2020}. In the common Fourier basis
\eqref{eq:bg:Fourier-basis}, the reference windows are
\begin{equation}
\label{eq:ext:mode-windows}
 \begin{array}{c|c|c}
 \text{space}&\text{Fourier indices }t&\text{multiplicity in residue }j\pmod q\\\hline
 K_{k-2,N}&1,\ldots,k-2&u_j=\#\{1\leq t\leq k-2:t\equiv j\pmod q\}\\
 K_{k,N}&0,\ldots,k-1&c_j=\#\{0\leq t<k:t\equiv j\pmod q\}\\
 K_{k+2,N}&-1,\ldots,k&w_j=\#\{-1\leq t\leq k:t\equiv j\pmod q\}.
 \end{array}
\end{equation}
Here $q$ is the order of the rotational label subgroup of $H$, as in
\eqref{eq:dom:kernel-cyclic-types}; $q=1$ is allowed. The matrices
$\Lambda,\widetilde\Lambda$ therefore carry cyclic multiplicities
$N/q-u_j,w_j$, respectively.
Linear reflections join and split these modes by
\eqref{eq:dom:self-paired-types}--\eqref{eq:dom:reflection-trace}, using
the displayed windows; exchanging elements relate the two external
spaces by~\eqref{eq:ext:equivariance}. 

The reference pair
$\mathcal U=K_{k-2,N}\subset K_{k+2,N}=\mathcal W$
is an admissible datum for every symmetry under consideration. Indeed,
these reference planes are positive by \eqref{eq:bg:Karp-point}, the windows
are nested and reflection-stable, and at $N=2k$ the map $A_N$ carries
the first window onto the complement of the third. Thus the external
conditions are nonempty, including at $k=2,N-2$, where a zero or full
space is interpreted in the usual way.

\subsection{The two pictures and their functionaries}
\label{subsec:external:pictures}

The \emph{symmetric amplituhedron maps} in the auxiliary ($A$) and
kinematic ($B$) pictures are the restrictions of the momentum maps:
\begin{subequations}
\label{eq:amp:maps}
\begin{align}
 \Phi^{\mathrm A}(C)&=(Y,\widetilde Y)
       =(C^\perp\Lambda,C\widetilde\Lambda)
       \in\Gr(N-k,N-k+2)\times\Gr(k,k+2),
 \label{eq:amp:A-map}\\
 \Phi^{\mathrm B}(C)&=(\lambda,\widetilde\lambda)
       =(C\cap\mathcal U^\perp,C^\perp\cap\mathcal W),
 \label{eq:amp:B-map}\\
 \SymAmpA{G}{\rho}_{N,k}(\Lambda,\widetilde\Lambda)
       &:=\Phi^{\mathrm A}(X^{\geq}),\qquad
 \SymAmpB{G}{\rho}_{N,k}(\mathcal U,\mathcal W)
       :=\Phi^{\mathrm B}(X^{\geq}).
 \label{eq:amp:images}
\end{align}
\end{subequations}
The intersections in \eqref{eq:amp:B-map} have dimension two for every
nonnegative $C$, and both maps are well defined there
\cite[Proposition~3.1(1)]{Galashin2024}.

Karp--Williams prove that the ordinary $A$- and $B$-amplituhedra are
homeomorphic~\cite[Proposition~3.12]{KarpWilliams2019}. Their ambient
correspondence in Lemma~3.10 is in fact a diffeomorphism. Applying it
to the two factors gives the momentum correspondence
\begin{equation}
\label{eq:amp:AB-diffeomorphism}
 \Psi_{\Lambda,\widetilde\Lambda}(Y,\widetilde Y)
       =\bigl(\Lambda(\ker Y),
                    \widetilde\Lambda(\ker\widetilde Y)\bigr),\qquad
 \Phi^{\mathrm B}=\Psi_{\Lambda,\widetilde\Lambda}\circ\Phi^{\mathrm A}.
\end{equation}
It is equivariant by~\eqref{eq:ext:equivariance}, so it restricts to a
diffeomorphism between our symmetric images as well. 
We henceforth identify the
two pictures, write $\SymAmp{G}{\rho}_{N,k}$, and work mainly in the
$B$-picture.

For later smooth statements, let
\begin{equation}
\label{eq:amp:regular-kinematics}
 \mathcal K^{\mathrm{reg}}_{\mathcal U,\mathcal W}
 :=\{(L,T)\in\Gr(2,\mathcal U^\perp)\times\Gr(2,\mathcal W):
       L\perp T,\ L\cap\mathcal W^\perp=0,\ T\cap\mathcal U=0\}.
\end{equation}
Every image point lies here. Since $T\cap\mathcal U=0$, varying $L$
supplies all four normal directions to $LT^T=0$. Thus this locus is
smooth of dimension $2N-4$. The group acts by~\eqref{eq:sym:pair-action}, and
our image lies in its fixed locus. This is a \emph{fixed-external}
kinematic space: neither external parameters nor further particlewise
rescaling quotients enter its dimension.

\paragraph{Brackets and Mandelstams.}
We use the angle brackets $\langle ij\rangle=\det(\lambda_i,\lambda_j)$,
square brackets $[ij]=\det(\widetilde\lambda_i,\widetilde\lambda_j)$,
and momentum sums $Q_I$ and Mandelstams $s_I$ defined in
\eqref{eq:bg:momentum-sums}--\eqref{eq:bg:mandelstams}.
In particular $s_I=\det Q_I$ and $Q_{[N]}=0$; planar $I$ are cyclic
intervals. These are homogeneous functionaries, with frame-independent
zero loci. In the $A$-picture their representatives are
\begin{equation}
\label{eq:amp:A-brackets}
 \langle Y;ij\rangle_{\Lambda}
    :=\det\!\begin{pmatrix}Y\\\Lambda_i\\\Lambda_j\end{pmatrix},
 \qquad
 [\widetilde Y;ij]_{\widetilde\Lambda}
    :=\det\!\begin{pmatrix}\widetilde Y\\\widetilde\Lambda_i\\\widetilde\Lambda_j\end{pmatrix}.
\end{equation}
Here $\Lambda_i,\widetilde\Lambda_i$ are external rows. Each family is
proportional to the corresponding spinor brackets by one common
nonzero scalar, by~\cite[Lemma~3.10, equation~(3.11)]{KarpWilliams2019}.
Consequently the same sum of products in~\eqref{eq:bg:mandelstams}
computes the lifted Mandelstam, up to a common normalization.

\begin{lemma}[Adjacent signs]
\label{lem:amp:adjacent-signs}
The two spinor planes can be oriented so that every twisted consecutive
angle and square bracket is nonnegative on $\SymAmp{G}{\rho}_{N,k}$,
and strictly positive on $\Phi^{\mathrm B}(X^>)$.
\end{lemma}
\begin{proof}
For $i<N$, Cauchy--Binet gives, up to positive normalizations independent
of $i$,
\begin{equation}
\label{eq:amp:adjacent-CB}
 \langle i,i+1\rangle\ \propto\!
 \sum_{|J|=k-2}\Delta_J(\mathcal U)\Delta_{J\cup\{i,i+1\}}(C),
 \quad
 [i,i+1]\ \propto\!
 \sum_{|J|=k}\Delta_J(C)\Delta_{J\cup\{i,i+1\}}(\mathcal W),
\end{equation}
where $J$ avoids $\{i,i+1\}$. The terms have the asserted signs; at the
seam use the sign $(-1)^{k-1}$ from~\eqref{eq:sym:dihedral-lifts}. This is also
\cite[Section~2.2 and Appendix~B]{DamgaardEtAl2019}, or the adjacent-sign
part of~\cite[Proposition~3.1]{Galashin2024}.
\end{proof}

\subsection{Submersivity, dimension and kinematic representations}
\label{subsec:external:geometry}

\begin{proposition}[The positive image]
\label{prop:amp:submersion}
The map $\Phi^{\mathrm B}:X^>\to
(\mathcal K^{\mathrm{reg}}_{\mathcal U,\mathcal W})^G$ is a smooth
submersion. Its image $\mathcal M^\circ:=\Phi^{\mathrm B}(X^>)$ is a
connected smooth manifold, open in the fixed component it meets, and
\begin{equation}
\label{eq:amp:manifold-closure}
 \SymAmp{G}{\rho}_{N,k}
       =\overline{\mathcal M^\circ}
\end{equation}
is its compact closure.
\end{proposition}
\begin{proof}
For the ordinary amplituhedron the submersion statement is
\cite[Lemma~8.1]{EvenZoharLakrecTessler2025}; its linear-algebra argument
extends to the momentum map. More explicitly, on the open part of the domain where
the two intersections have dimension two, a fiber over $(L,T)$ is an
open subset of $\Gr(k-2,T^\perp/L)$. The incidence construction is an
open Grassmann bundle over~\eqref{eq:amp:regular-kinematics}; equivalently,
at $D=C^\perp$,
\begin{equation}
\label{eq:amp:ordinary-kernel}
 \ker d\Phi^{\mathrm B}_C
       \simeq\Hom(C/L,D/T),\qquad
 \operatorname{rank}d\Phi^{\mathrm B}_C
       =k(N-k)-(k-2)(N-k-2)=2N-4.
\end{equation}
The restriction to fixed points is submersive: average any tangent lift
of an invariant target vector over the finite group. The image is
therefore open and smooth; it is connected by Theorem~\ref{thm:dom:ball}.
That theorem also gives $\overline{X^>}=X^{\geq}$, so continuity and
compactness imply~\eqref{eq:amp:manifold-closure}.
\end{proof}

\paragraph{The two spinor representations.}
All characters in the following formulas are real characters of
$\widehat H$. Put
$c=\chi_C$, $d=\chi_{C^\perp}$, $u=\chi_{\mathcal U}$ and
$w=\chi_{\mathcal W}$. The constant types above and the equivariant
exact sequences
\begin{equation}
\label{eq:amp:spinor-exact-sequences}
 0\longrightarrow L\longrightarrow C\longrightarrow\mathcal U^*
       \longrightarrow0,\qquad
 0\longrightarrow T\longrightarrow\mathcal W\longrightarrow C^*
       \longrightarrow0
\end{equation}
give $a:=\chi_L=c-u$ and $b:=\chi_T=w-c$. These are actual
representations at every nonnegative source point, not generic virtual
differences. In the common cyclic indexing their multiplicities are
\begin{equation}
\label{eq:amp:spinor-modes}
 a_j=\#\{t\in\{0,k-1\}:t\equiv j\pmod q\},\qquad
 b_j=\#\{t\in\{-1,k\}:t\equiv j\pmod q\}.
\end{equation}
Coincident residues are counted twice. Reflection pairing and signs are
read from the same generators as in~\eqref{eq:ext:mode-windows}.
An exchanging element sends $L$ to $T$ and $T$ to $L$; the full signed
group thus acts on their graded pair, not on either plane separately.

\paragraph{Dimension and momentum conservation.}
Let $P_\rho$ denote the dimension of the relevant fixed component of
$\Gr(2,\mathcal U^\perp)\times\Gr(2,\mathcal W)$, and let $q_\rho$
count its independent momentum equations. Since $g^2\in\widehat H$ for
each exchanging $g$, character averaging gives
\begin{subequations}
\label{eq:amp:dimension}
\begin{align}
 P_\rho&=\frac1{|\widehat G|}
       \sum_{h\in\widehat H}\bigl(a(h)d(h)+b(h)c(h)\bigr),
 \label{eq:amp:product-dimension}\\
 q_\rho&=\frac1{|\widehat G|}
       \left(\sum_{h\in\widehat H}a(h)b(h)
             +\sum_{g\in\widehat G\setminus\widehat H}a(g^2)\right),
 \label{eq:amp:momentum-constraints}\\
 \dim\SymAmp{G}{\rho}_{N,k}&=P_\rho-q_\rho.
 \label{eq:amp:target-dimension}
\end{align}
\end{subequations}
Indeed, at $(L,T)$ the tangent space to the ambient product is
\begin{equation*}
 \Hom(L,\mathcal U^\perp/L)\oplus\Hom(T,\mathcal W/T).
\end{equation*}
The two quotient representations have characters $d$ and $c$,
respectively, by~\eqref{eq:amp:spinor-exact-sequences}.
Since duals of real orthogonal representations have the same characters,
a linear element $h\in\widehat H$ has trace
$a(h)d(h)+b(h)c(h)$. An exchanging element interchanges the two tangent
summands, so its trace on their direct sum is zero.

Momentum conservation is the vanishing of the Euclidean pairing
restricted to $L\times T$, viewed as an element of $L^*\otimes T^*$.
A linear element therefore has trace $a(h)b(h)$ on this space.
For an exchanging $g$, write $g_L:L\to T$ and $g_T:T\to L$ for its
restrictions, in orthonormal bases. On a matrix $Z$ representing
a bilinear form, the action and its trace are
\begin{equation*}
 Z\longmapsto g_T Z^T g_L^T,\qquad
 \operatorname{tr}\bigl(g|_{L^*\otimes T^*}\bigr)
   =\operatorname{tr}(g_Tg_L)=a(g^2).
\end{equation*}
This uses the same transpose-trace identity as the proof of
Theorem~\ref{thm:dom:ball}, but without the minus sign arising there
from orthogonal complementation.

The differential of the pairing is surjective: $T\cap\mathcal U=0$
makes $\mathcal U^\perp\to T^*$ surjective, so varying $L$ alone
produces every infinitesimal pairing. Averaging a lift of an invariant
normal vector preserves surjectivity on invariants, as in
Proposition~\ref{prop:amp:submersion}. Character averaging therefore
gives $P_\rho$ and $q_\rho$, and their difference is the image dimension.
Here $g^2$ is the actual signed operator: if $g^2|_L=\pm I$,
then $a(g^2)=\pm2$. 

For an equivalent dimension calculation one may subtract the invariant dimension of
\eqref{eq:amp:ordinary-kernel} from the domain dimension
\eqref{eq:dom:character-dimension}. In particular, irreducible sectors
absent from both spinor planes contribute only to the fiber, not to
the image dimension.
\begin{example}
Ordinary momentum has four
independent equations and dimension $2N-4$; pure exchange gives three
and dimension $2n-3$; either basic Lagrangian exchange gives one and
dimension $2n-1$, where $N=2n$.

For a purely linear group, writing $m_\nu,c_\nu,a_\nu,b_\nu$ for the
real irreducible multiplicities and $e_\nu$ as in
\eqref{eq:dom:real-multiplicities}, the same answer is
\begin{equation}
\label{eq:amp:multiplicity-dimension}
 \dim\SymAmp{G}{\rho}_{N,k}
 =\sum_\nu e_\nu
       \bigl(a_\nu(m_\nu-c_\nu)+b_\nu c_\nu-a_\nu b_\nu\bigr).
\end{equation}
\end{example}
\section{Mandelstam boundaries and factorization subspaces}
\label{sec:boundaries}
This section describes our ansatz for codimension $1$ boundary facets of the symmetric amplituhedra. They correspond to vanishing loci of certain Mandelstam variables, which we specify in the next subsection. We then study the irreducible factors of these Mandelstam variables, the domain cells which map to their zero loci, their factorizations and image dimensions. The factorization follows the architecture of the plabic tangles operad of \cite{PlabicTangles}, only that here the core carries the group $G$ symmetry, and the blobs are exchanged by the group action, and may carry a local symmetry. See Figure~\ref{fig:schematic_bdries} for the illustration.

Keep the external datum of Section~\ref{sec:external-and-map}, and abbreviate
$\Phi=\Phi^{\mathrm B}$, $X^{\geq}=\SymGrNonneg{G}{\rho}{k,N}$ and
$\mathcal M=\SymAmp{G}{\rho}_{N,k}$. We first organize the proposed
codimension-one supports, then construct cells mapping to them. 
We retain \emph{factorization subspace} for a
locus that may include several such cells or their degenerations; it is
not necessarily a linear space.
We prove vanishing and compute dimensions of the constructed images using
ordinary positive external data. Their position on the topological boundary,
and completeness of the proposed list, are separate conjectural assertions
to be addressed in sequels, and will be proven under the simplifying assumption of immanant positivity under which the corresponding results for the usual momentum amplituhedron and ABJM amplituhedron were proven in \cite{Galashin2024,OrenPerlsteinTessler2025}. 
\subsection{Candidate channels, symmetries and reduced functions}
\label{subsec:boundary:channels}
\paragraph{Boundary ansatz: two geometric families.}
For a cyclic interval $I$ with $2\leq |I|\leq N-2$, its
\emph{interval channel} is the unordered cut $[I]=\{I,I^c\}$.
Thus choosing the opposite side does not produce a new channel.
We also identify channels related by the label-group action.

Our \emph{boundary ansatz} is that the channel families described
below, together with their end cases and compatible factorizations
of codimension-one image, exhaust the distinct Mandelstam boundary
supports. 
The list concerns distinct supports, rather than every
Mandelstam variable that may vanish on them.

For a dihedral label group, the ansatz has two geometric families,
defined using the canonical arcs of
Section~\ref{subsec:reflection-geometry}.

\smallskip\noindent
\emph{(a) One-arc intervals.}
Choose a canonical arc and remove its endpoints that are fixed
vertices; endpoints that are only nearly fixed are retained.
A \emph{one-arc interval} is a cyclic interval contained in the
remaining set of labels. It therefore stays on one side of the
neighboring reflection axes. The two-label interval consisting of
a fixed edge is assigned to the next family.

\smallskip\noindent
\emph{(b) Reflection-symmetric channels.}
Here some label reflection $h$ preserves the unordered cut:
\begin{equation*}
 hI=I \qquad\text{or}\qquad hI=I^c.
\end{equation*}
These two possibilities have different geometric meanings.
If $hI=I$, the interval itself is preserved: it is centered at a
fixed vertex or at the midpoint of a fixed edge. We call it a
\emph{reflection-centered interval}. For the principal local
constructions, we begin with such intervals extending into the
canonical arcs on either side of their center, with the length
restrictions specified below. A vertex-centered interval has odd
length, whereas an edge-centered interval has even length.

If $hI=I^c$, the reflection exchanges the two sides of the cut, so
$N$ is even and $|I|=N/2$. When the entire label group preserves
$[I]$, this gives a two-sided half-circle factorization: the two
factors are related by symmetry, not chosen independently.
The condition $hI=I^c$ alone does not imply this for the full group.
For example, if $p>2$, a half-circle overlaps its translate by the
basic rotation. Such overlapping intervals are not used as the
outer blobs of a two-sided factorization; any occurrence in the
ansatz must be accounted for by the actual factorization support.

\paragraph{Local length ranges.}
Write $N=p\ell$, where as usual $p$ is the order of the rotational label
subgroup and $\ell$ is the number of labels in its quotient.
For $\ell\geq2$, the principal local constructions use the maximal
lengths in~\eqref{eq:bdy:geometric-ranges} below.
The second column refers to family~(a); the third refers to the
reflection-centered subcase $hI=I$ of family~(b), not to the
side-swapping half-circles.
The quotient types are those of
\eqref{eq:sym:quotient-types}: for odd $\ell$ both vertex and edge
centers occur, while for even $\ell$ the action is of vertex or
edge type.
\begin{equation}
\label{eq:bdy:geometric-ranges}
\begin{array}{c|c|l}
\text{quotient type}
 &\text{largest one-arc length}
 &\text{largest centered length}\\\hline
\ell=2a+1
 &a
 &\ell\text{ at a vertex};\ \ell-1\text{ at an edge}\\
\ell=2a,\ \text{vertex}
 &a-1
 &\ell-1\text{ at a vertex}\\
\ell=2a,\ \text{edge}
 &a
 &\ell\text{ at an edge}.
\end{array}
\end{equation}
For example, in the even vertex-type case, consecutive fixed
vertices are separated by $a$ boundary edges. Removing these two
endpoints leaves $a-1$ labels available for a one-arc interval.
All lengths must also satisfy $2\leq |I|\leq N-2$, and
complementation may give a shorter representative.

For a nontrivial purely rotational label group, there are no
reflection axes defining canonical arcs. The analogous initial
family consists instead of intervals with $2\leq |I|\leq\ell$.
For the trivial label group there is no additional length cutoff.

\paragraph{Symmetric core--blob factorizations.}
The principal factorizations in our boundary ansatz have the two
forms shown in Figure~\ref{fig:schematic_bdries}. Write $\mathcal O(I)$
for the set of distinct label translates of the chosen interval $I$.
For the one-arc, reflection-centered and purely rotational ranges
above, these intervals are pairwise disjoint, including at the
maximal permitted lengths. They carry the outer blobs, while the
core carries the remaining external labels.

In the connected configuration of Figure~\ref{fig:schematic_bdries}(a),
each blob is attached directly to the core by one internal edge,
representing one-particle amalgamation. There are no edges between
distinct blobs: after contracting each factor to a vertex, the
diagram is a star. In the disconnected configuration of
Figure~\ref{fig:schematic_bdries}(b), the factors are combined by
direct sum, without internal connecting particles. These are two
factorization mechanisms, not a one-to-one assignment to the two
geometric channel families.

The factors are not independent. The core carries its induced
$G$-action, each blob and connecting edge respects its stabilizer,
and symmetry-related blobs are carried to one another by the
prescribed signed action. Thus only one blob's data is chosen per
orbit. A symmetry exchanging the two sides of a half-channel
relates their factors; the cut and its internal connection are
counted only once. Soft limits replace blobs by single lollipops;
a degenerate core may have only connector labels, or may disappear
in a direct sum. These end cases do not create shared external
particle labels between the actual factors.

These geometric conditions and length ranges organize the principal
constructions; they do not assert that every listed channel
contributes a factorization subspace with codimension-one image:
this depends on a finer, case-dependent analysis of the rank, the
signed symmetry action and the admissible factorization.
Nor do they list every interval whose Mandelstam vanishes on one
of the same supports. In some end cases, additional Mandelstam variables vanish on the
same codimension-one locus, without defining new boundary components, and the intervals indexing these vanishings may overlap one another,
even though the external-label sets of the actual factorization
blocks remain disjoint.

\begin{figure}
    \centering
    \includegraphics[width=1
\linewidth]{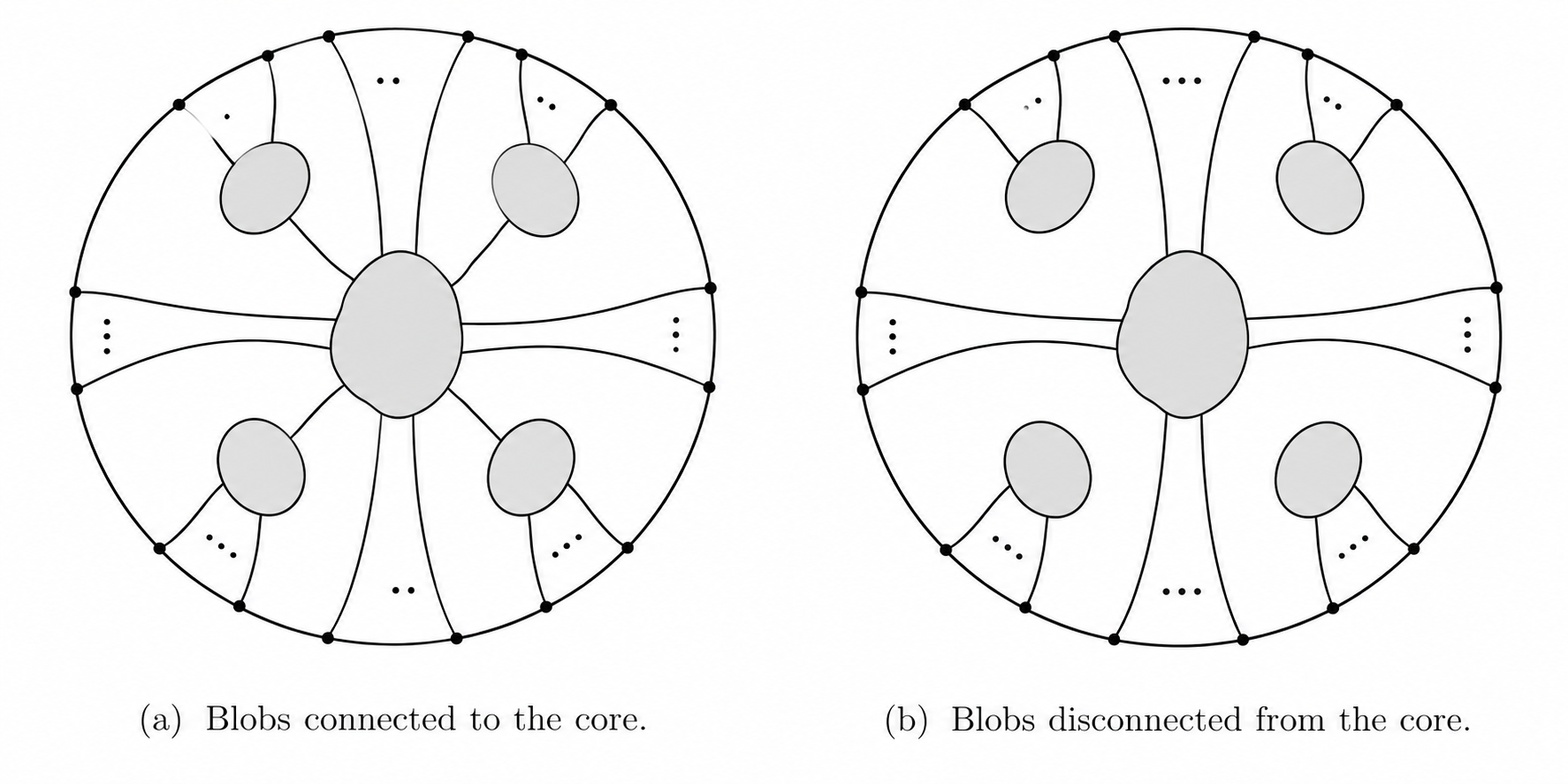}
    \caption{The two (schematic) types of mandelstam boundaries. In both figures we have a core, which represents the big cell of the nonnegative $G-$symmetric Grassmannian, and a $G-$symmetric collection of blobs. Blobs and connecting edges may carry a symmetry, which will always be a subsgroup of $\mathbb{Z}_2\times\mathbb{Z}_2,$ where the first factor is a reflection and the second is the exchange factor, and in this case represent the big cell of the corresponding symmetric nonnegative Grassmannian. Blobs which differ by rotation/reflection/exchange or a combination thereof correspond to Grassmannian factors which differ by the same type of symmetry. In particular, in the combinatorial level exchange corresponds to change of colors. In Figure (a) the blobs are connected to the core by an edge, which in the geometric level corresponds to a fiber product in the level of Grassmannians, and amalgamation in the level of vector spaces. The connecting edge may also carry a non trivial stabilizer. In Figure (b) they are disconnected, and this corresponds geometrically to a product (in the level of Grassmannian) and direct sums (in the level of vector spaces). The schematic generic families above also have end cases. The soft limit is an end case of the disconnected digaram, which each blob is a single lollipop. There is also the end case of a degenerate core.}
    \label{fig:schematic_bdries}
\end{figure}

\paragraph{The symmetry of a Mandelstam.}
\paragraph{Symmetry on momentum matrices.}
Choose bases of the two spinor planes, represented by the rows of
$\lambda$ and $\widetilde\lambda$. In \eqref{eq:sym:pair-action},
$T_\gamma$ permutes labels with the same sign on both spinor columns,
so this sign cancels in $p_i=\lambda_i\widetilde\lambda_i^T$.
Exchange transposes this product. Write $R_{\lambda,\gamma}$ and
$R_{\widetilde\lambda,\gamma}$ for the two-by-two matrices expressing
the symmetry in these bases after accounting for the label permutation
and signs; exchanging elements map between the two sectors.
Their representation types are fixed by
\eqref{eq:amp:spinor-exact-sequences}--\eqref{eq:amp:spinor-modes}.
The induced momentum action and its constraint at a channel are
\begin{subequations}
\label{eq:bdy:momentum-action}
\begin{align}
 Q_{g_\gamma I}&=\mathsf R_\gamma(Q_I),\qquad
 \mathsf R_\gamma(Q)=
 \begin{cases}
 R_{\lambda,\gamma}Q R_{\widetilde\lambda,\gamma}^T,&\epsilon(\gamma)=0,\\
 R_{\lambda,\gamma}Q^T R_{\widetilde\lambda,\gamma}^T,&\epsilon(\gamma)=1,
 \end{cases}
 \label{eq:bdy:Q-equivariance}\\
 \widehat G_{[I]}&=\{\gamma\in\widehat G:g_\gamma I\in\{I,I^c\}\},\qquad
 \mathsf R_\gamma(Q_I)=\delta_I(\gamma)Q_I.
 \label{eq:bdy:channel-stabilizer}
\end{align}
\end{subequations}
Here $\delta_I:\widehat G_{[I]}\to\{\pm1\}$ equals $+1$ when
$g_\gamma I=I$ and $-1$ otherwise. The last identity follows from
\eqref{eq:bdy:Q-equivariance} and $Q_{I^c}=-Q_I$, by momentum conservation.
Thus $s_I$, $s_{I^c}$ and $s_{g_\gamma I}$ have the same vanishing locus
(and agree in compatible normalizations).

\paragraph{Local forms of the partial momentum.}
The constraints \eqref{eq:bdy:channel-stabilizer} are linear in $Q_I$.
If the distinct translated intervals partition the labels, momentum
conservation additionally gives
\begin{equation}
\label{eq:bdy:orbit-sum}
 \sum_{J\in\mathcal O(I)}Q_J=0.
\end{equation}
The setwise symmetries alone give the following normal forms, before
imposing additional orbit relations or side-swapping symmetries;
see Appendix~\ref{app:bdy:determinant-types} for the proof.
Overall nonzero scalar factors in the determinants are suppressed.
\begin{equation}
\label{eq:bdy:determinant-types}
\renewcommand{\arraystretch}{1.15}
\begin{array}{l|c|c}
\text{setwise symmetry of }I&\text{allowed }Q_I& s_I\\\hline
\text{none}&\Mat_{2,2}&\det Q_I\\
\text{pure exchange}&\operatorname{Sym}_2&q_{11}q_{22}-q_{12}^2\\
\text{linear reflection}&\operatorname{diag}(\mu_I^-,\mu_I^+)&\mu_I^+\mu_I^-\\
\text{vertex reflection with exchange}&\operatorname{Sym}_2&q_{11}q_{22}-q_{12}^2\\
\text{edge reflection with exchange}&\mathbb R\!\left(\begin{smallmatrix}0&1\\-1&0\end{smallmatrix}\right)&\mu_I^2\\
\text{linear vertex reflection and pure exchange}&\operatorname{diag}(\mu_I^-,\mu_I^+)&\mu_I^+\mu_I^-\\
\text{linear edge reflection and pure exchange}&\mathbb R\!\left(\begin{smallmatrix}0&1\\1&0\end{smallmatrix}\right)&-\mu_I^2.
\end{array}
\end{equation}
Here $\operatorname{Sym}_2$ consists of matrices
$Q_I=\left(\begin{smallmatrix}q_{11}&q_{12}\\q_{12}&q_{22}\end{smallmatrix}\right)$.
In the linear-reflection row, $hI=I^c$ instead gives an off-diagonal
matrix, whose determinant is minus the product of its two entries.
The $\mu$ superscripts are calibrated by the connector character in
\eqref{eq:bdy:branch-character}, not by an intrinsic ordering of the factors.
The reflected Lagrangian square is \eqref{eq:bg:reflected-squares}.
\paragraph{Polynomial factorization and specializations.}
To determine the polynomial factorization of $s_I$, restrict the
determinant to the full constraints
\eqref{eq:bdy:channel-stabilizer}, including side-swapping symmetries,
and impose \eqref{eq:bdy:orbit-sum} when the translated intervals
partition the labels.

For a two-particle interval, we always have
$s_{\{i,i+1\}}=\langle i,i+1\rangle[i,i+1]$.
The equations $\langle i,i+1\rangle=0$ and $[i,i+1]=0$ are the
two \emph{adjacent chiral branches}.
This factorization holds even when the determinant of a general
matrix in the allowed momentum space is irreducible.
Pure exchange identifies the two brackets up to sign and
normalization, so their product becomes a square up to a nonzero
scalar.

More generally, if the constraints restrict the allowed momentum
space to a line, write $Q_I=\mu_IQ_0$ with $Q_0$ fixed. Then
\begin{equation}
\label{eq:bdy:reduced-square}
 s_I=\det(Q_0)\mu_I^2.
\end{equation}
When $\det Q_0\ne0$, the reduced vanishing equation is $\mu_I=0$.
If $\det Q_0=0$, the Mandelstam is identically zero and does not
define a boundary.
An example occurs for a \emph{maximal centered sector}: a
reflection-centered interval of length $\ell=N/p$, when permitted
by \eqref{eq:bdy:geometric-ranges}. Its $p$ rotational translates
partition the labels. In a regular linear $D_p$ model with
$p\geq2$ and $p\mid k$, the reflection and orbit-sum constraints
give the square specialization above, subject to $\det Q_0\ne0$.

This must be distinguished from a direct-sum factorization in the
domain, which forces $Q_I=0$ on its image but need not make $s_I$
a square on the surrounding kinematic space.
For example, in even edge type, a \emph{complete one-arc interval}
contains all labels of a canonical arc between neighboring fixed
edges, including its nearly fixed endpoints. Its length is
$\ell/2$, and its $2p$ distinct translates partition the labels.
In the linear model, when $2p\mid k$, the direct-sum construction
has rank $k/(2p)$ on each arc block and gives $Q_J=0$ for every
block $J$. Nevertheless, the determinant on the surrounding
allowed momentum space can remain nonsquare; a traceless normal
form, for instance, has determinant $-x^2-yz$.
Thus a direct-sum realization, a square polynomial, and an image
of codimension one are three different assertions.

\paragraph{Which families occur.}
We now summarize the channel families in the boundary ansatz,
subject to the length ranges above.
The action families are those of
Section~\ref{subsec:symmetry-families}; we abbreviate the chosen
dihedral label action $D_{p,c}$ by $D_p$.
For the mixed families, put
$G_\chi:=\Gamma_\chi(D_{p,c})$ with $\chi(r)=1$, so $p$ is even.

In the table, an entry $s_I$ denotes the Mandelstam equation,
with the adjacent chiral branches and further specializations
just described; it does not assert irreducibility.
The notation $\mu_I^\pm$ lists the two possible character-factor
equations $\mu_I^+=0$ and $\mu_I^-=0$.
A \emph{quadric} means the symmetric determinant
$q_{11}q_{22}-q_{12}^2$, and a \emph{square} means a nonzero scalar
multiple of $\mu_I^2$, whose reduced equation is $\mu_I=0$.
The centered column concerns intervals preserved setwise by a
reflection.
\begin{equation}
\label{eq:bdy:family-list}
\renewcommand{\arraystretch}{1.18}
\begin{array}{l|c|l}
\text{action}
 &\text{one-arc/rotational channels}
 &\text{centered channels}\\\hline
1,\ C_p\text{ linear}
 &s_I
 &\text{---}\\
\langle\mathcal D\rangle,\ C_p\times E
 &s_I,\ |I|\text{ odd}
 &\text{---}\\
\langle r\mathcal D\rangle\ (p\text{ even})
 &s_I
 &\text{---}\\
D_p\text{ linear}
 &s_I
 &\mu_I^\pm;\ \text{square at }|I|=\ell,\ p\mid k\\
G_R=\langle r,s\mathcal D\rangle
 &s_I
 &\text{vertex: quadric; edge: square}\\
G_\times=D_p\times E
 &s_I,\ |I|\text{ odd}
 &\text{vertex: }\mu_I^\pm;\ \text{edge: square}\\
G_\chi\text{ mixed}
 &s_I
 &\text{local cases described below}.
\end{array}
\end{equation}
For a mixed group, the reflection stabilizing $I$ determines
the entry: a linear reflection gives the character factors
$\mu_I^\pm$, a vertex reflection with exchange gives the
symmetric determinant, and an edge reflection with exchange
gives a square, as in \eqref{eq:bdy:determinant-types}.
A single linear reflection is included by $p=1$; the Karpman
model is the $p=1$ edge-reflection specialization of $G_R$.

Exchange also relates the factors on distinct image intervals,
rather than supplying independent choices.
An augmented blob of rank $\kappa$ on $|I|+1$ labels is carried
to its positive dual of rank $|I|+1-\kappa$; at an adjacent channel,
the angle and square branches are interchanged.
The odd-length entries in the one-arc/rotational column refer
to one-particle factorizations: pure exchange fixes the augmented
blob, so its $|I|+1$ labels must have even cardinality.
Even-length direct sums and terminal end cases are considered
separately below, not excluded by this rule.

The ordinary momentum and ABJM entries agree with their known
boundary constructions~\cite{Galashin2024,OrenPerlsteinTessler2025};
the Karpman specialization is that of~\cite{KroviTessler2026}.
For the remaining models, the table specifies candidate families.
A factorization contributes to the codimension-one ansatz only
when its ranks and signed representations admit a nonnegative
source plane and its image has codimension one, as computed in
\eqref{eq:bdy:image-dimension} below.
In particular, two algebraic factors need not give two boundary
components, and a direct-sum construction need not have
codimension-one image.
\paragraph{End cases.}
The adjacent chiral branches, complete one-arc intervals and maximal
centered sectors described above are end cases of the same two channel
families. Their factorizations retain the core--blob forms of
Figure~\ref{fig:schematic_bdries}; the core may carry only connecting
labels, or disappear altogether in a direct-sum configuration.
As explained above, the odd-length condition in the presence of pure
exchange concerns one-particle factorizations. It does not exclude
even-length direct-sum constructions, whose image dimensions must
be computed separately.

We also include two kinds of small factors in the constructions below.
An allowed orbit of reflection-fixed labels may consist of loops or
coloops. Their particle momenta vanish, and deleting these labels
leaves a smaller symmetric core; these are the soft limits.
Alternatively, a reflection-centered interval of three particles may
carry a blob of rank $1$ or $3$ on four labels---the three external
labels together with the connecting label---when its local symmetry
permits this rank. Such a blob still participates in one-particle
amalgamation, and its external particle momenta need not vanish.

\paragraph{Several Mandelstams on the same support.}
In some end cases, additional Mandelstam variables vanish on the
same codimension-one locus, without defining new boundary components.
Their indexing intervals may overlap, although the external-label
sets of the actual factorization blocks remain disjoint.
The following two identities explain how this can happen.

\smallskip\noindent
\emph{Deleting an endpoint of a zero-momentum interval.}
Suppose $K$ is a cyclic interval with $Q_K=0$, as in a direct-sum
factorization. Removing an endpoint $i$ leaves another cyclic interval,
and
\begin{equation*}
 Q_{K\setminus\{i\}}=-p_i,
 \qquad
 s_{K\setminus\{i\}}=\det(-p_i)=0,
\end{equation*}
because every particle momentum has rank at most one.
Thus the additional Mandelstam vanishes automatically on the
zero-momentum locus, rather than imposing a further condition.

\smallskip\noindent
\emph{Half-circles meeting reflection-fixed vertices.}
Let $N=2m$, and let a linear reflection $h$ fix two opposite labels
$a,c$. Write the cyclic label set as
\begin{equation*}
 [N]=\{a\}\sqcup B\sqcup\{c\}\sqcup D,
 \qquad |B|=|D|=m-1,\qquad hB=D.
\end{equation*}
Here $B,D$ are the two open intervals between the fixed labels.
The half-circle $I=\{a\}\cup B$ and its reflected image
$hI=\{a\}\cup D$ overlap at $a$.

Choose reflection eigenbases in both spinor planes, ordered so that
$R_{\lambda,h}=R_{\widetilde\lambda,h}
=D_2:=\operatorname{diag}(1,-1)$.
Then $\mathsf R_h(Q)=D_2QD_2$, and the fixed-particle momenta
$p_a,p_c$ are diagonal. Equivariance and momentum conservation give
\begin{equation*}
 Q_I+\mathsf R_h(Q_I)
   =Q_I+Q_{hI}=p_a-p_c.
\end{equation*}
Writing $\mu,\nu$ for the two off-diagonal entries of $Q_I$, we obtain
\begin{equation}
\label{eq:bdy:half-arc-correction}
 \begin{gathered}
 Q_I=\frac{p_a-p_c}{2}
       +\begin{pmatrix}0&\mu\\\nu&0\end{pmatrix},\\
 s_I=-\mu\nu-\frac14s_{\{a,c\}},
 \qquad
 s_B=-\mu\nu+\frac14s_{\{a,c\}}.
 \end{gathered}
\end{equation}
Indeed, $Q_B=Q_I-p_a$, and the determinant formulas follow from
$\det p_a=\det p_c=0$.
Here $s_{\{a,c\}}=\det(p_a+p_c)=\langle ac\rangle[ac]$;
the pair $\{a,c\}$ need not itself be a cyclic interval.

The signs of the signed reflection at $a,c$ have ratio
$(-1)^{k-1}$. For odd domain rank $k$, the endpoint spinors therefore
belong to the same reflection eigenline in each sector, so
$\langle ac\rangle=[ac]=0$.
Consequently, adjoining either or both fixed endpoints to $B$
does not change its Mandelstam:
\begin{equation}
\label{eq:bdy:half-arc-alias}
 s_{\{a\}\cup B}
   =s_B
   =s_{B\cup\{c\}}
   =s_{\{a,c\}\cup B}
   =-\mu\nu.
\end{equation}
These identities hold before imposing any Mandelstam vanishing.
Thus, whenever the shorter interval $B$ describes one of the
codimension-one factorization images, the other three intervals
give additional Mandelstams vanishing on that same locus.
For even $k$, the endpoint term $s_{\{a,c\}}$ is not forced to
vanish, and these equalities need not hold.

Only the single linear reflection was used in this calculation.
If the group also contains rotations, the conditions from the
entire orbit of $B$ must still be imposed. 
\subsection{Factorization subspaces, vanishing and codimension}
\subsubsection{The generic families}
\label{subsec:boundary:subspaces}
The codimension $1$ boundaries of the amplituhedron come in two generic families, whose corresponding domain cells are depict in Figure~\ref{fig:schematic_bdries} -- the one-particule fiber product and the direct sums. In this subsection we will study these families. First their domain structure and then finer questions regarding the codimensions of their images in the target. Appendix~\ref{app:bdy:determinant-types} analyzes which irreducible factors of mandelstam variables contain their images in their zero loci.  
\paragraph{One-particle fiber products.}
Consider a connected core--blob configuration as in
Figure~\ref{fig:schematic_bdries}(a), with disjoint external-label
sets for its outer blobs. Let $\mathcal T$ denote its star-shaped
graph of factors: the central vertex represents the core, the
outer vertices the symmetry-related blobs, and each edge their
single internal connection. This is the graph between factors,
not the internal plabic graph of any factor.
Each vertex carries a cyclically ordered source plane $C_v$ on its
external labels and one extra label for each incident internal edge.
Glue along the internal labels by the signed operation
\eqref{eq:bg:fiber-product}. Positive gluing is the usual
Grassmannian amalgamation~\cite{AHBCGPT2016}; equivalently, its maximal
minors are the signed Cauchy--Binet sums with the planar gluing signs.
We take all compatible nonnegative factors and require symmetry-related
copies to be carried to one another by the prescribed signed action. Denote the resulting subspace of
$X^{\geq}$ by $\mathcal S_{\mathcal T}$. A piece specified by one complete
positroid support is a symmetric positroid cell. The collective notation
does not identify a union of supports with one cell.
Its regular part has the specified source projection ranks, connectivity
one across each blob--core cut, and no vector supported solely on deleted
internal coordinates. Note that these are conditions in the domain. The nonzero-momentum
statement will follow from positivity below.

\smallskip
\emph{Ranks and blob's cardinality.}
For such a star, let $m$ be the common length of its pairwise disjoint
outer intervals $J\in\mathcal O(I)$, and let $b_J$ be the rank of the
augmented blob on $J\sqcup\{*_J\}$. Its core has
\begin{equation}
\label{eq:bdy:core-counts}
 N_c=N-|\mathcal O(I)|(m-1),\qquad
 k_c=k-\sum_{J\in\mathcal O(I)}(b_J-1),\qquad
 b_{g_\gamma J}=\begin{cases}b_J,&\epsilon(\gamma)=0,\\m+1-b_J,&\epsilon(\gamma)=1.\end{cases}
\end{equation}
Thus $|\mathcal O(I)|=p$ for a rotational orbit or a centered dihedral
orbit, and $2p$ for a generic one-arc dihedral orbit. In exchange-bearing
cases the compatible core remains at middle rank. The formula includes
a connector-only core, without assigning it a fictitious nonterminal
amplituhedron dimension. In the two-sided half-channel case its
bivalent core can be suppressed, leaving the single gluing between
$I$ and $I^c$; complementary representatives do not supply two
independent cuts or two independent internal scales.

\smallskip
\emph{Blob types.}
A proper interval is stabilized by no nontrivial rotation and by at most
one label reflection. Its possible setwise stabilizers, including exchange,
are therefore $1$, $\langle\mathcal D\rangle$, $\langle h\rangle$,
$\langle h\mathcal D\rangle$, and $\langle h,\mathcal D\rangle$.
The representative blobs consequently belong to precisely the small
family already introduced:
\begin{equation}
\label{eq:bdy:blob-dictionary}
\begin{array}{c|l}
\text{stabilizer}&\text{blob domain}\\\hline
1&\text{ordinary nonnegative Grassmannian}\\
\langle\mathcal D\rangle&\text{orthogonal/ABJM}\\
\langle h\rangle&\text{linear reflection, Example~\ref{ex:domain:linear-reflection}}\\
\langle h\mathcal D\rangle&\text{Karpman or vertex-exchange, Example~\ref{ex:domain:exchange-vertex}}\\
\langle h,\mathcal D\rangle&\text{reflection-fixed orthogonal, Example~\ref{ex:domain:reflection-orthogonal}}.
\end{array}
\end{equation}
The action is on the augmented label set when a connector is present.
If exchange stabilizes the blob, $2b=m+1$. In particular the extremal
$\Gr(1,4)$ and $\Gr(3,4)$ blobs cannot themselves be exchange-fixed;
they can be exchanged as distinct partners. A setwise exchange-edge
block in the square construction below is a direct-sum block, with no
added connector. The central factor can retain the full larger group.

\paragraph{Direct sums and soft limits.}
For $J\subseteq[N]$, write $C[J]:=C\cap\mathbb R^J$ for the subspace
of vectors in $C$ supported on $J$. By contrast, $C_J$ denotes the
submatrix on columns $J$, whose row space is the projection $\pi_J(C)$.
Thus $C[J]=\ker(\pi_{J^c}|_C)$ and
$\dim C[J]=k-\operatorname{rank}C_{J^c}$.
We identify a vector supported on $J$ with its $J$-coordinates, and
use the same conventions for other planes.
For a symmetry-stable noncrossing coordinate partition $\mathcal P$
and prescribed block ranks $q_J$, set
\begin{equation}
\label{eq:bdy:direct-sum-space}
 \mathcal S_{\mathcal P}^{\oplus}
 =\left\{C\in X^{\geq}:C=\bigoplus_{J\in\mathcal P}C[J],
                        \ \dim C[J]=q_J\right\}.
\end{equation}
The positive embeddings include the usual cyclic signs.
In the disconnected core--blob configuration of
Figure~\ref{fig:schematic_bdries}(b), the remaining external labels
form the core block, when nonempty. The core carries the induced
$G$-symmetry, just as in the one-particle construction.

The blobs obey the same symmetry relations as in that construction:
each representative belongs to the nonnegative Grassmannian fixed
by its full setwise stabilizer, with the local symmetry types listed
in \eqref{eq:bdy:blob-dictionary}. The other blobs in its orbit are
obtained by the prescribed signed relabellings, together with
duality when exchange is present. Here these actions are on the
block $J$ itself; no connecting label is added.

Accordingly, ranks propagate as $q_{g_\gamma J}=q_J$ when
$\epsilon(\gamma)=0$ and $q_{g_\gamma J}=|J|-q_J$ when
$\epsilon(\gamma)=1$. An exchanging element stabilizing $J$
therefore requires $2q_J=|J|$.
Thus the independent choices are the symmetric core, if present,
and one representative blob per orbit, each subject to its
local symmetry.

The soft subspace is the special case with singleton blocks:
\begin{equation}
\label{eq:bdy:soft-space}
 \mathcal S_{O_0,O_1}^{\mathrm{soft}}
 =\{C\in X^{\geq}:C_i=0\ (i\in O_0),\ e_i\in C\ (i\in O_1)\},
 \quad (N_c,k_c)=(N-|O_0|-|O_1|,\ k-|O_1|).
\end{equation}
The assignments obey \eqref{eq:atoms:lollipop-rule}.
Both loops and coloops are retained. Pure exchange and an exchanging
vertex stabilizer forbid an individual lollipop; mixed rotation-exchange
can alternate loops and coloops around a linearly stabilized orbit.
There is no nonzero internal momentum in either
\eqref{eq:bdy:direct-sum-space} or \eqref{eq:bdy:soft-space}.

\subsubsection{Vanishing of Mandelstam variables}
For cyclic intervals in the nonnegative domain, the vanishing
implication below also follows from Galashin's double-dimer formula
and connectivity criterion
\cite[Proposition~6.2 and Lemma~6.24]{Galashin2024}.
We give a direct incidence proof of the rank bound, requiring
neither positivity nor symmetry.
\begin{lemma}[Incidence vanishing]
\label{lem:bdy:vanishing}
Let $C\in\Gr(k,N)$ and let $L\subset C$, $T\subset C^\perp$ be the two
spinor planes. With $\pi_I$ the coordinate projection, put
\begin{equation}
\label{eq:bdy:connectivity}
 \kappa_C(I)=\operatorname{rank}C_I+\operatorname{rank}C_{I^c}-k
 =k-\dim(C\cap\mathbb R^I)-\dim(C\cap\mathbb R^{I^c}).
\end{equation}
Then $\operatorname{rank}Q_I\leq\kappa_C(I)$.
Consequently a one-particle cut has $s_I=0$, a direct-sum cut has
$Q_I=0$, and every label in \eqref{eq:bdy:soft-space} has $p_i=0$.
\end{lemma}
\begin{proof}
The pairing $\beta_I:C\times C^\perp\to\mathbb R$,
$\beta_I(x,y)=\langle\pi_Ix,y\rangle$, has left kernel
$(C\cap\mathbb R^I)\oplus(C\cap\mathbb R^{I^c})$:
$\pi_Ix\in C$ is exactly this condition. Its rank is therefore
$\kappa_C(I)$. Restricting to $L\times T$ gives $Q_I$ and cannot
increase rank. A blob--core cut joined by one internal edge has
connectivity at most one; a direct-sum block has connectivity zero. A loop gives $\lambda_i=0$, while a coloop
gives $\widetilde\lambda_i=0$.
\end{proof}
The lemma does not rely on positivity. 
Internal positive operations which preserve these cut ranks preserve
vanishing as well; this includes hyperbolic layers supported entirely
on one side. Note that it does not say that spinors are unchanged when the
source point varies with external data held fixed.

\paragraph{Positivity forces the required ranks.}
We first record a pairing fact and its consequence for the images
of the spinor planes under quotient maps. In our applications,
$\mathcal U=(\operatorname{colspan}\Lambda)^\perp$ and
$\mathcal W=\operatorname{colspan}\widetilde\Lambda$ are the fixed
positive external spaces of \eqref{eq:ext:positive-flag}.

\begin{lemma}[Positive pairing and supported kernels]
\label{lem:bdy:positive-pairing}
If $P\in\Gr_{\geq}(d,N)$ and $E\in\Gr_{>}(r,N)$, the Euclidean
pairing $P\times E\to\mathbb R$ has rank $\min(d,r)$.

Let $2\leq k\leq N-2$, $C\in\Gr_{\geq}(k,N)$ and
$\mathcal U\in\Gr_{>}(k-2,N)$, and put
$L=C\cap\mathcal U^\perp$.
For a surjective linear map $q:C\twoheadrightarrow F$, write
$K=\ker q$. Let $S_k$ be the signed cyclic shift of
\eqref{eq:sym:dihedral-lifts}. If
$S_k^jK\in\Gr_{\geq}(\dim K,N)$ for some $j\in[N]$, then
\begin{equation}
\label{eq:bdy:positive-kernel-rank}
 \dim q(L)=\min(2,\dim F).
\end{equation}

For $\mathcal W\in\Gr_{>}(k+2,N)$, put $D=C^\perp$ and
$T=D\cap\mathcal W$.
A surjective map $\widetilde q:D\twoheadrightarrow\widetilde F$
similarly satisfies
$\dim\widetilde q(T)=\min(2,\dim\widetilde F)$,
provided
$S_{N-k}^jA_N(\ker\widetilde q)$ is nonnegative for some integer $j$.
\end{lemma}

\begin{proof}
The sign-variation characterizations in
\cite[Theorem~1.1 and Corollary~1.12]{KarpVariation2017} imply
$P\cap E^\perp=0$ when $d\leq r$, and
$E\cap P^\perp=0$ when $d\geq r$.
These are respectively injectivity and surjectivity of the pairing
map $P\to E^*$, proving the first assertion.

Applying the first assertion to $C$ and $\mathcal U$ gives
$\dim L=2$. To compute the rank of the pairing on $K\times\mathcal U$,
choose $j$ as in the hypothesis. The plane $S_k^jK$ is nonnegative,
and $S_k^j\mathcal U$ is strictly positive because $k-2$ and $k$
have the same parity. Since $S_k^j$ is orthogonal,
\[
 \langle S_k^j x,S_k^j u\rangle=\langle x,u\rangle
 \qquad(x\in K,\ u\in\mathcal U).
\]
Thus the original pairing has the same rank as the pairing of
these two transformed spaces, to which the first assertion applies.
Writing $f=\dim F$, this rank is $\min(k-f,k-2)$, so
\[
 \dim(K\cap L)
 =\dim(K\cap\mathcal U^\perp)
 =k-f-\min(k-f,k-2)
 =\max(0,2-f).
\]
As $\ker(q|_L)=K\cap L$, subtracting from $\dim L=2$
proves \eqref{eq:bdy:positive-kernel-rank}.

For the dual assertion, replace $C,\mathcal U,L$ by
$A_ND,A_N\mathcal W^\perp,A_NT$.
The first two spaces are respectively nonnegative of rank $N-k$
and positive of rank $N-k-2$, by positive duality, and
\[
 A_NT=(A_ND)\cap(A_N\mathcal W^\perp)^\perp.
\]
The same argument, now using $S_{N-k}$, applies.
\end{proof}

\smallskip\noindent
\emph{Verification for supported interval kernels.}
Recall that $C[J]=C\cap\mathbb R^J$ consists of the vectors
supported on $J$, whereas $C_J$ denotes the column submatrix.
For a cyclic interval $J$, the space $C[J]$ is the kernel of
$\pi_{J^c}|_C$. 
\begin{lemma}\label{lem:C_J_nonneg}
    $C[J]$ is nonnegative. The same holds for $\bigoplus_\alpha C[J_\alpha]$ if $(J_\alpha)$ are pairwise disjoint cyclic intervals.
\end{lemma}
\begin{proof}
For this calculation, use $S_k^j$ to make $J$ the initial interval,
and retain the notation $C,J$ for the transformed plane and interval.
Put $d=\dim C[J]$; rank--nullity gives
$\operatorname{rank}C_{J^c}=k-d$.
Extend a basis of $C[J]$ to a basis of $C$. With column blocks
$J,J^c$, the resulting row representative has the form
\[
 \widehat C=
 \begin{pmatrix}P&0\\ R&Z\end{pmatrix},
 \qquad
 \operatorname{rank}P=d,\quad \operatorname{rank}Z=k-d.
\]
The block $R$ need not vanish: the additional basis vectors may
have coordinates on both sides of the cut.
Choose $B\subset J^c$ with $|B|=k-d$ and $\det Z_B\ne0$.
For every $A\subset J$ with $|A|=d$, block expansion gives
\begin{equation}
\label{eq:bdy:supported-minors}
 \Delta_{A\cup B}(\widehat C)
   =\det\begin{pmatrix}P_A&0\\R_A&Z_B\end{pmatrix}
   =\det(P_A)\det(Z_B).
\end{equation}
All maximal minors of $\widehat C$ have a common weak sign,
and $\det Z_B$ is nonzero and independent of $A$.
Thus the maximal minors of $P$ have a common weak sign;
choosing its orientation makes them nonnegative.
The ambient representative $(P\mid0)$ of $C[J]$ has all other
maximal minors zero, as claimed. 

For several pairwise disjoint cyclic intervals, choose one common
seam between them, so that none crosses it.
The same block expansion works for each interval; placing its
columns first introduces only a common sign, independent of $A$.
Hence all the supported spaces are nonnegative in this same
coordinate order.
Order their row bases according to their intervals. Every nonzero
maximal minor of their direct sum is then a product of block
minors, and is nonnegative.
Consequently, the lemma also holds for kernels of the form
$\bigoplus_\alpha C[J_\alpha]$. 
\end{proof}
\begin{proposition}[Local ranks and representation types]
\label{prop:bdy:positive-ranks}
Let $C\in X^{\geq}$, put $D=C^\perp$, and let $L,T$ be its
spinor planes. For every cyclic interval $J$,
\begin{equation}
\label{eq:bdy:interval-ranks}
 \dim\pi_JL=\min(2,\operatorname{rank}C_J),\qquad
 \dim\pi_JT=\min(2,\operatorname{rank}D_J).
\end{equation}
For a factor $C_v$ of rank $k_v$ on $N_v$ labels in a regular
connected core--blob configuration, the induced spinors $L_v,T_v$,
with connecting coordinates retained, satisfy
\begin{equation}
\label{eq:bdy:factor-ranks}
 \dim L_v=\min(2,k_v),\qquad
 \dim T_v=\min(2,N_v-k_v).
\end{equation}
Writing $\operatorname{Res}_{H_v}$ for restriction of the action
to the factor's signed linear stabilizer $H_v$, their types are
\begin{equation}
\label{eq:bdy:factor-types}
 \begin{aligned}
 L_v&\simeq
 \begin{cases}
 \operatorname{Res}_{H_v}L,&k_v\geq2,\\
 C_v,&k_v\leq1,
 \end{cases}
 &
 T_v&\simeq
 \begin{cases}
 \operatorname{Res}_{H_v}T,&N_v-k_v\geq2,\\
 C_v^\perp,&N_v-k_v\leq1.
 \end{cases}
 \end{aligned}
\end{equation}
The factor statements also hold for interval direct sums and
for residual factors obtained by deleting supported interval
or singleton summands.
\end{proposition}

\begin{proof}
The restriction map $\pi_J|_C:C\to\pi_J(C)$ kills exactly the
vectors vanishing on $J$, so its kernel is $C[J^c]$.
Lemma~\ref{lem:C_J_nonneg}, Lemma~\ref{lem:bdy:positive-pairing}, and its dual for $D$,
therefore give \eqref{eq:bdy:interval-ranks}.

\smallskip\noindent
\emph{The gluing and its assumptions.}
Let $\mathcal J=\mathcal O(I)$ be the collection of disjoint blob
intervals, and put
$J_0=[N]\setminus\bigcup_{J\in\mathcal J}J$.
The core $C_c$ has external labels $J_0$ and one connecting label
$*_J$ for each blob. Denote the blob on
$J\sqcup\{*_J\}$ by $C_{b,J}$.

In the connected construction, we require each connecting label
to be neither a loop nor a coloop in either incident factor.
Explicitly, its coordinate functional is nonzero, and the
coordinate vector supported only on that label does not belong
to the factor. Thus any real connecting value can be attained,
but it cannot be changed while all other coordinates are fixed.
These are the precise assumptions used below.

Fix the internal coordinate scales, and write
$\varepsilon_J\in\{\pm1\}$ for the prescribed gluing sign at $*_J$.
The amalgamation formula \eqref{eq:bg:fiber-product} says that
$x\in\mathbb R^N$ belongs to $C$ exactly when there are scalars
$t_J$, one for each $J\in\mathcal J$, satisfying
\[
 (x_{J_0},(t_J)_{J\in\mathcal J})\in C_c,
 \qquad
 (x_J,\varepsilon_Jt_J)\in C_{b,J}
 \quad\text{for every }J\in\mathcal J.
\]
Thus $t_J$ is the core coordinate at $*_J$, and
$\varepsilon_Jt_J$ is the matching blob coordinate.

\smallskip\noindent
\emph{The maps to the factors.}
For a given $x\in C$, these scalars are unique.
Indeed, two choices would give, in each blob, a vector whose
external coordinates are zero and whose only possibly nonzero
coordinate is $*_J$. The non-coloop assumption excludes such
a nonzero vector, so the two choices coincide.
We can therefore define
\[
 q_c(x)=(x_{J_0},(t_J)_{J\in\mathcal J}),\qquad
 q_{b,J}(x)=(x_J,\varepsilon_Jt_J).
\]
These are the \emph{factor-extraction maps}: they recover from
$x$ its vector in the core or in a chosen blob.
They are linear because the matching equations are linear and
their connecting coordinates are uniquely determined.

They are also surjective.
Given a core vector, choose a vector in each blob with the
required connecting value; this is possible because the blob's
connecting coordinate functional is nonzero.
Given a vector in one chosen blob, first choose a core vector
with the one required connecting value, and then match its other
connecting values in the remaining blobs.
Here we use the nonzero connecting functionals on the core as well.

\smallskip\noindent
\emph{Their kernels.}
For the blob map, $q_{b,J}(x)=0$ implies $x_J=0$.
Conversely, if $x_J=0$, the corresponding blob vector is supported
only on $*_J$, so it is zero. Hence
\[
 \ker(q_{b,J}:C\to C_{b,J})
   =C[J^c]=C\cap\mathbb R^{J^c}.
\]

For the core map, suppose $q_c(x)=0$.
Every blob vector then has connecting coordinate zero.
Each such vector can be glued to zero in the core and in all
other blobs, giving a vector of $C$ supported on its own interval.
Their sum is $x$.

Conversely, suppose $x\in C[J]$ for one blob interval $J$.
Every other blob has zero external coordinates, hence is zero.
The core vector consequently has all coordinates zero except
possibly $*_J$. Since this label is not a coloop of the core,
the core vector is zero too. Thus
\[
 \ker(q_c:C\to C_c)
   =\bigoplus_{J\in\mathcal J}C[J].
\]
The sum is direct because the external intervals are disjoint.

\smallskip\noindent
\emph{The dual maps and the spinor ranks.}
The orthogonal complement $D=C^\perp$ is obtained from
$C_c^\perp$ and the $C_{b,J}^\perp$ by the same construction,
with opposite matching signs: the dual connecting values are
$u_J$ in the core and $-\varepsilon_Ju_J$ in the blob.
Indeed, these signs cancel the two internal contributions to
the Euclidean pairing; the resulting space is contained in
$C^\perp$, and the amalgamation dimension count gives equality.

Orthogonal complementation interchanges loops and coloops.
The dual factors therefore satisfy the same connector assumptions.
The preceding construction gives surjective linear maps
$\widetilde q_v:D\to C_v^\perp$, with the same kernel descriptions
after replacing $C$ by $D$.

Now define the induced spinor spaces by
$L_v=q_v(L)$ and $T_v=\widetilde q_v(T)$.
All the displayed kernels are supported interval spaces or sums
of such spaces on disjoint intervals.
Their nonnegativity after the appropriate signed relabelling
was proved above; for the $D$-kernels, use positive duality.
Lemma~\ref{lem:bdy:positive-pairing} therefore gives
\eqref{eq:bdy:factor-ranks}.
For direct sums or soft deletion, the maps are coordinate
projections onto the retained factors, and their kernels are
the discarded supported summands, so the same argument applies.

\smallskip\noindent
\emph{The representation types.}
Let $H_v=\operatorname{Stab}_{\widehat H}(v)$ be the subgroup
of signed linear symmetries preserving the chosen factor.
For a blob, this is the linear part of its signed stabilizer
in \eqref{eq:bdy:blob-dictionary}; the core is preserved by
all of $\widehat H$.

The prescribed factor symmetries preserve the matching equations.
Uniqueness of the connecting coordinates therefore makes the
factor-extraction maps $H_v$-equivariant.
If $k_v\geq2$, then $q_v|_L$ has rank two, equal to $\dim L$,
so it identifies $L_v$ with the restriction of $L$ to $H_v$.
If $k_v\leq1$, its rank equals $\dim C_v$, so $L_v=C_v$.
Applying the same argument to $\widetilde q_v|_T$ proves
the remaining cases of \eqref{eq:bdy:factor-types}.
Exchanging symmetries relate the corresponding factors and
interchange the two spinor sectors.
\end{proof}
Thus the ranks are forced at \emph{every} point of the specified domain
part. No additional genericity of the external datum is required.
As a useful endpoint consequence,
$\lambda_i=0$ exactly at source loops, and $\widetilde\lambda_i=0$
exactly at source coloops, under the positive external assumptions here.

\begin{corollary}[The momentum on a low-connectivity cut]
\label{cor:bdy:momentum-rank}
For a cyclic interval $I$, put $\kappa=\kappa_C(I)$. Then
\begin{equation}
\label{eq:bdy:momentum-rank}
 \operatorname{rank}Q_I=\kappa\quad(\kappa\leq2),\qquad
 \operatorname{rank}Q_I\geq\max(0,4-\kappa)\quad(\kappa\geq2).
\end{equation}
In particular a connectivity-one source plane has nonzero null momentum for
all our ordinary positive external data.
\end{corollary}
\begin{proof}
The partial pairing of Lemma~\ref{lem:bdy:vanishing} descends to a
perfect pairing between
\begin{equation}
\label{eq:bdy:connector-quotients}
 E_C=C/(C[I]\oplus C[I^c]),\qquad
 E_D=D/(D[I]\oplus D[I^c]),\qquad \dim E_C=\dim E_D=\kappa.
\end{equation}
Let $q_C:C\twoheadrightarrow E_C$ and $q_D:D\twoheadrightarrow E_D$
be the quotient maps just displayed. Their kernels are sums of
supported interval spaces, so the quotient-map assertion
\eqref{eq:bdy:positive-kernel-rank} and its dual give
$\dim q_C(L)=\dim q_D(T)=\min(2,\kappa)$.
Thus these images fill both quotients when $\kappa\leq2$,
and the perfect pairing gives $\operatorname{rank}Q_I=\kappa$.
For $\kappa>2$, the annihilator of the two-plane $q_D(T)$
in $E_C$ has dimension $\kappa-2$, so its intersection with
$q_C(L)$ has dimension at most $\kappa-2$.
Hence the restricted pairing has rank at least
$2-(\kappa-2)=4-\kappa$.
\end{proof}
This does not assert nonvanishing of the determinant at $\kappa\geq3$.

For a reflection channel with two distinct factors, see Table~\eqref{eq:bdy:determinant-types}, each specified
connected one-particle construction selects one factor to vanish;
the other is nonzero on its connectivity-one locus.
Appendix~\ref{app:reflection-factors} identifies the selected factor
from the signed representations and gives conditions for global
nonvanishing of the other factor.
\subsubsection{Induced images and their dimensions.}
For $y=(L,T)\in\Phi(\mathcal S_{\mathcal T})$ with every internal
$Q_J$ nonzero of rank one, write $p_*=-Q_J=\lambda_*\widetilde\lambda_*^T$.
Restrict the external spinors and adjoin these internal columns at each
vertex, with opposite momenta on opposite sides. This determines the
augmented pairs $(L_v,T_v)$ up to internal rescaling. Conversely their
compatible gluing recovers $y$. The induced external equations are
obtained from the parent equations by eliminating the opposite source
rows; they are not, in general, raw coordinate restrictions of
$\mathcal U,\mathcal W$. Thus the image is an externally constrained
fiber product, not a Cartesian product with freely chosen lower flags.
For direct sums, take $(L_J,T_J)=(\pi_JL,\pi_JT)$ instead.

Here is a dimension calculation which also handles small factors.
For each independent factor $v$, let $s_v$ be its fixed-domain dimension
from \eqref{eq:dom:character-dimension}. Let $f_v$ be the dimension of
the fixed incidence Grassmannian
\begin{equation}
\label{eq:bdy:local-fiber}
 F_v(y)=\{E_v:L_v\subset E_v\subset T_v^\perp,
                 \ E_v\text{ has the prescribed signed source-plane type}\}.
\end{equation}
Its dimension is computed by the same character formula on
$C_v/L_v$ and $C_v^\perp/T_v$.
For example, for a linear factor, with real multiplicities
$c_{v\nu},d_{v\nu},a_{v\nu},b_{v\nu}$ and
$e_\nu=\dim_{\mathbb R}\End(R_\nu)$,
\begin{equation}
\label{eq:bdy:local-image-count}
 d_v:=s_v-f_v
 =\sum_\nu e_\nu\bigl[c_{v\nu}d_{v\nu}
               -(c_{v\nu}-a_{v\nu})(d_{v\nu}-b_{v\nu})\bigr].
\end{equation}
Use the types forced by Proposition~\ref{prop:bdy:positive-ranks}, with
the actual inherited signed action. A factor of source rank or corank
one has the corresponding one-dimensional spinor space.
For example, suppose $[N]=F\sqcup R$ and
$C=C[F]\oplus C[R]$, where $F$ is the union of the direct-sum
blocks to be removed and $R$ consists of the retained labels.
Put $k_R=\dim C[R]$ and $Z=\mathcal W^\perp$.
To describe the projected spinors on $R$, we eliminate the
coordinates on $F$ from the external equations, allowing those
coordinates to vary in $C[F]$ and $D[F]$, respectively.
The resulting external spaces are
\begin{equation}
\label{eq:bdy:effective-flag}
 \begin{gathered}
 U_{\mathrm{eff}}=\{u_R:u\in\mathcal U,\ u_F\perp C[F]\},\qquad
 Z_{\mathrm{eff}}=\{z_R:z\in Z,\ z_F\perp D[F]\},\quad
 W_{\mathrm{eff}}=Z_{\mathrm{eff}}^\perp,\\
 \dim U_{\mathrm{eff}}=\max(k_R-2,0),\qquad
 \dim W_{\mathrm{eff}}=\min(|R|,k_R+2),\qquad
 U_{\mathrm{eff}}\subset W_{\mathrm{eff}}.
 \end{gathered}
\end{equation}
Here $W_{\mathrm{eff}}=Z_{\mathrm{eff}}^\perp$
is taken inside $\mathbb R^R$.
Indeed Lemma~\ref{lem:bdy:positive-pairing} computes both eliminated
pairing ranks; projection to $R$ is injective on their annihilators,
by global positive transversality. 
The inclusion $U_{\mathrm{eff}}\subset W_{\mathrm{eff}}$ follows from
$\langle u_R,z_R\rangle
 =\langle u,z\rangle-\langle u_F,z_F\rangle=0$,
since $\mathcal U\perp Z$, $u_F\in D[F]$ and $z_F\in C[F]$.
\begin{proposition}[Dimension of a regular factorization image]
\label{prop:bdy:dimension}
Suppose the specified factorization subspace has a nonempty regular
positive part, using the supported-interval extractions above. Its local
spinor types are \eqref{eq:bdy:factor-types}. Let $g_{\mathcal T}$ be the dimension of its
symmetry-compatible internal rescaling group. Then
\begin{equation}
\label{eq:bdy:image-dimension}
 \dim\Phi(\mathcal S_{\mathcal T})
       =\sum_{v/G}(s_v-f_v)-g_{\mathcal T}.
\end{equation}
For a direct-sum or soft subspace the same formula holds with
$g_{\mathcal T}=0$ and the restricted pairs in place of augmented pairs.
\end{proposition}
\begin{proof}
Choose one factor per symmetry orbit. Cutting a source in the
connected core--blob locus recovers these factors up to
symmetry-compatible rescaling of the connecting coordinates.
Hence \[\dim\mathcal S_{\mathcal T}=\sum_{v/G}s_v-g_{\mathcal T}.\]

Fix $y=(L,T)$ and choose its internal spinor representatives,
thereby fixing this rescaling freedom. For each independent
factor, the remaining choice is $E_v\in F_v(y)$: it must contain
$L_v$, be orthogonal to $T_v$, and have the prescribed signed
representation type.
Subject to the open conditions preserving the cut ranks,
gluing these factors gives exactly the sources over $y$.
The fiber therefore has dimension $\sum_{v/G}f_v$, with no
further rescaling quotient. The dimensions $f_v$ are constant
because Proposition~\ref{prop:bdy:positive-ranks} fixes the local
spinor ranks and representation types.
Subtracting the fiber dimension from the domain dimension
proves \eqref{eq:bdy:image-dimension}.

For a direct sum, the incidence conditions hold separately on
each block:
\[\pi_JL\subset C[J]\subset(\pi_JT)^\perp\] inside $\mathbb R^J$.
There are no connecting coordinates and hence no rescaling
quotient, so the same calculation applies with $g_{\mathcal T}=0$.
Positivity is open on the chosen regular factorization part
and does not alter these dimensions.
\end{proof}
Let $t_e>0$ be the rescaling parameter at each internal connection
$e$. Linear symmetries permute these parameters, while exchanging
symmetries also invert them. In the coordinates $u_e=\log t_e$,
inversion becomes multiplication by $-1$, so the action is by
permutations with sign $(-1)^{\epsilon}$.
Consequently,
\[
 g_{\mathcal T}
 =\dim\bigl(\mathbb R^{E(\mathcal T)}
                 \otimes\mathrm{sgn}_{\epsilon}\bigr)^G.
\]
A single interface orbit contributes one unless an exchanging element
stabilizes that interface, when it contributes zero. For example,
ordinary gluing subtracts one and ordinary ABJM gluing subtracts none.
The image dimension is a sum over independent factor orbits, not over
the number of drawn blobs.

\paragraph{Codimension-one evaluations.}
Equation~\eqref{eq:bdy:image-dimension} gives explicit codimension-one
results for the indicated nonempty regular domain constructions. Their
spinor representations are determined by \eqref{eq:bdy:factor-types},
not by an additional genericity hypothesis.
For example,
\begin{equation}
\label{eq:bdy:basic-dimension-checks}
\begin{array}{l|l}
\text{factorization}&d_\partial=d_{\rm parent}-1\\\hline
\text{ordinary momentum}&(2N_L-4)+(2N_R-4)-1=2N-5\\
\text{Fraser, noncritical}&[2(\ell-m+1)-2]+(2m-2)-1=2\ell-3\\
\text{linear }D_p,\ \text{one arc}&[\ell-2(m-1)-1]+(2m-2)-1=\ell-2\\
\text{linear }D_p,\ \text{centered}&[\ell-m]+(m-1)-1=\ell-2\\
\text{symmetric ABJM, odd }m&[\ell-m]+(m-2)=\ell-2\\
\text{symmetric Shevchenko}&[2(\ell-m+1)-1]+(2m-2)-1=2\ell-2.
\end{array}
\end{equation}
Here $p\geq2$ in the linear-dihedral rows. A single linear reflection
gives $N-3$ for either nondegenerate gluing type; the basic reflected and
rotational Lagrangian one-particle constructions give $2n-2$.
These are regular-factor formulas, not substitutions for an extremal
or connector-only factor.

For the exchange-dihedral families, the same calculation is particularly
short. The following table records the drop to the induced core, the
independent blob contribution, and the internal rescaling dimension.
Each row gives $d_\partial=d_c+d_b-g=d_{\rm parent}-1$, where Here $g=g_{\mathcal T}$ is the dimension of the
symmetry-compatible internal rescaling group, and $d_b$ counts
the contribution of the independent blob representatives,
not all their symmetry-related copies.
\begin{equation}
\label{eq:bdy:exchange-dimension-checks}
\renewcommand{\arraystretch}{1.15}
\begin{array}{l|c|c|c}
\text{regular construction}&d_{\rm parent}-d_c&d_b&g\\\hline
G_R,G_\chi:\ \text{one arc}&2m-2&2m-2&1\\
G_\times:\ \text{one arc, odd }m&m-1&m-2&0\\
G_\chi:\ \text{linear center}&m-1&m-1&1\\
G_R,G_\chi:\ \text{exchange-vertex center}&m-1&m-2&0\\
G_\times:\ \text{vertex center}&(m-1)/2&(m-3)/2&0\\
G_R,G_\chi:\ \text{exchange-edge direct sum}&m&m-1&0\\
G_\times:\ \text{edge-centered direct sum}&m/2&m/2-1&0.
\end{array}
\end{equation}
The core uses its inherited action, including a change of vertex/edge
phase when a centered interval is replaced by its connecting label.
The counts follow by substituting those types and the blob types of
\eqref{eq:bdy:blob-dictionary} into \eqref{eq:bdy:local-image-count}
and its exchange analogue. They apply on the regular components with
these types; in particular, an even interval stabilized \emph{linearly}
in a mixed group is not the exchange-edge direct-sum row.

\begin{example}
For clarity, three end calculations are worth making explicit.
The Karpman direct sum on a centered interval of length $2r$ gives
\begin{equation}
\label{eq:bdy:Karpman-codimension}
 d_\partial=(2r-1)+(2(n-r)-1)=2n-2.
\end{equation}
At $r=1$ the induced two-label factor contributes its one domain
parameter, not an ordinary momentum dimension at rank one.
For a linear maximal-sector direct sum, $N=p\ell$, $k=pq$, there is
one independent reflection-fixed sector, and the local count gives
\begin{equation}
\label{eq:bdy:critical-linear-codimension}
 d_{\rm parent}=\ell-1,\qquad
 d_\partial=\begin{cases}
 \ell-2,&2\leq q\leq\ell-2,\\
 \lfloor(\ell-1)/2\rfloor,&q=1\text{ or }q=\ell-1.
 \end{cases}
\end{equation}
Indeed, in the interior range each reflection character occurs once in
each restricted spinor plane. The two summands in
\eqref{eq:bdy:local-image-count} are $m_+-1$ and $m_--1$, totaling
$\ell-2$. When the sector rank is $q=1$, its source plane is represented
by a positive row $(c_1,\ldots,c_\ell)$ with
$c_i=c_{\ell+1-i}$. Its $\lceil\ell/2\rceil$ independent entries,
modulo common scaling, give $\lceil\ell/2\rceil-1$ parameters.
The projected $\lambda$-line recovers this row up to scale,
so these parameters also give the image dimension.
The corank-one case follows by positive duality. Thus the latter
cases yield codimension $1$ strata only for $\ell=2,3$.

A regular fixed-vertex soft deletion gives $\ell-1\to\ell-2$ for a
linear $D_p$ model and $N-2\to N-3$ for one reflection. Coloop deletions are included as well, but reduce the core rank
by the number of deleted coloops. Their image dimension must
therefore be computed using this residual rank and the inherited
signed action in \eqref{eq:bdy:local-image-count}, rather than
inferred from the number of remaining labels alone. For example, the
regular pure-coloop component on $\{1,5\}$ in vertex $D_2$, $(N,k)=(8,4)$,
has domain and fiber dimensions $2,0$, against $4,1$ for its parent;
hence the image drops from three to two. For a linearly
centered extremal block of length $m$, the positive splitting ratios
have dimension $\lceil m/2\rceil-1$. 
For the end case of a rank-one augmented blob, the external $\lambda$-columns
are nonzero and proportional: $\lambda_i=c_i v$ for $i\in I$.
Their proportionality coefficients recover the splitting ratios,
up to a common scale. The corank-one case is dual, with
$\widetilde\lambda$ in place of $\lambda$. Thus
\begin{equation}
\label{eq:bdy:extremal-dimension}
 d_\partial=d_c+\lceil m/2\rceil-1.
\end{equation}
In the regular linear-dihedral case $d_c=\ell-m$, so the defect is
$\lfloor m/2\rfloor$; the triple is a codimension-one end case.
Adjacent $\Gr(1,3)$ and $\Gr(2,3)$ blobs give the two usual chiral
ends by the same computation.
\end{example}
All codimensions here are with respect to the ambient symmetric amplituhedron under consideration. 
In the stated nonempty constructions,
positivity has supplied the local ranks, so the dimension calculation
requires no additional external-genericity or Mandelstam-sign assumption.
We conjecture that the listed codimension-one images are precisely the intended
boundary supports. 
\section{The geometric BCFW procedure}
\label{sec:bcfw}
The BCFW recursion \cite{BrittoCachazoFeng2005,BrittoCachazoFengWitten2005} is a method to calculate a scattering amplitude from its poles' residues, using Cauchy's theorem. When the poles locations are known, e.g. Mandelstam divisors, and when the residues satisfy a recursive structure, e.g. factorization patterns as the ones discussed in the previous section, the outcome is a recursive method for calculating the scattering amplitude, or in positive geometries scenarios, the canonical form. 

BCFW tilings are geometric counter parts of this idea. They are decompositions of amplituhedra into a collection of images of positroid cells of the target's dimension. The collections of the domain cells are calculated in a recursive process which imitates the recursion for the amplitudes: Starting from an internal point of an amplituhedron, and performing a BCFW shift, which is the flow generated by one of the atomic operations in the theory, e.g. $x_i(t)$ or a boost, until a boundary component is reached. Assuming all boundary strata are the vanishing loci of Mandelstam variables $s_I$, and using the factorization of these boundaries to fiber or direct products of smaller cells, one obtains the recursive recipe. In practice the parameter $t$ is found as a root of the polynomial expression obtained by shifting corresponding momentum $P_I$ by $t,$ calculating the determinant, which is the shifted Mandelstam variable. In the original amplituhedron and its momentum version the degree of the polynomial is always $1.$ The same is true for the reflected Lagrangian amplituhedron, but there one has to shift the irreducible factor $\mu_I$ for symmetric $I,$ in order to obtain an affine shift. In the ABJM setting the degree is $2.$ In the general symmetric setting we consider here the degree varies, depending on the theory, see Appendix~\ref{app:bcfw:degrees}. Still, in all our experiments, even with assuming only basic positivity of the external data, and not the stronger immanant positivity, we have obtained a single positive preimage for every candidate BCFW collection, also in cases of higher degree polynomials. We have constructed the cells based on the boundary ansatz of~\ref{subsec:boundary:channels}.


\subsection{Shifted equations and the components that are hit}
\label{subsec:bcfw:shifts}

Fix the external datum, write $d=\dim\mathcal M$, and choose one
orbit operation $M(t)$ at the current recursion node.
Here $t$ denotes the algebraic parameter: $t>0$ for additive
operations, and $t=u=e^\tau>1$ for boosts. Denote its positive
range by $T_M$.

For $y=(\lambda,\widetilde\lambda)$, moving the operation to the
external datum as in \eqref{eq:bg:bridge-transfer} gives
\begin{equation}
\label{eq:bcfw:promoted}
 \begin{gathered}
 (\Lambda_t,\widetilde\Lambda_t)
       =(M(t)^{-T}\Lambda,M(t)\widetilde\Lambda),\qquad
 (\mathcal U_t,\mathcal W_t)=(M(t)\mathcal U,M(t)\mathcal W),\\
 \lambda_t=\lambda M(t)^{-1},\qquad
 \widetilde\lambda_t=\widetilde\lambda M(t)^T,\qquad
 Q_I^M(t)=\lambda M(t)^{-1}P_I M(t)\widetilde\lambda^T,
 \end{gathered}
\end{equation}
where $P_I$ is the diagonal coordinate projector. Note that this is a removal
calculation at fixed original target, not the path
$\Phi_{\mathcal U,\mathcal W}(CM(t))$ with the original external datum.
The common transformation of $\mathcal U,\mathcal W$ preserves nesting;
positivity and symmetry are preserved by the allowed operations.

If $I$ is such that $s_I$ is reducible, as studied in Section~\ref{subsec:boundary:channels},
choose its reduced channel function $F$: an adjacent chiral bracket,
a selected $\mu_I^\eta$, the reduced $\mu_I$ of a square, or an
unsplit Mandelstam. The removal equation is the
\begin{equation}
\label{eq:bcfw:removal}
 F^M(y;t):=F(\lambda_t,\widetilde\lambda_t)=0.
\end{equation}
Substitution gives a polynomial equation in the algebraic parameter $t$
after possibly clearing denominators. A component is \emph{algebraically hit}
when its reduced equation depends nontrivially on $t.$
A root contributes to the BCFW process only if it admits the positive reconstruction
described below.

\paragraph{Locality and degree.}
For $M(t)=e^{tX}$,
\begin{equation}
\label{eq:bcfw:commutator}
 \frac{d}{dt}(M(t)^{-1}P_I M(t))
       =M(t)^{-1}[P_I,X]M(t).
\end{equation}
For boosts, the same identity uses $u\,d/du$ and the normalized
generator. Thus operations entirely inside $I$ or $I^c$ leave
$Q_I$ unchanged. The partial momentum can change only if an arrow of the operation
joins $I$ to $I^c$. For a vertex palindrome or three-vertex boost,
this requires that exactly one or two of its three labels lie in $I$. For a single ordinary arrow $a\to b$,
the change is
\[
 t\sigma(\mathbf1_I(a)-\mathbf1_I(b))
       \lambda_a\widetilde\lambda_b^T,
\]
of rank at most one. A commuting ordinary packet therefore gives
$Q_I^M(t)=Q_I+tA_I$ and
\begin{equation}
\label{eq:bcfw:quadratic-pencil}
 s_I^M(t)=\det Q_I^M(t)
        =s_I+t\operatorname{tr}(\operatorname{adj}(Q_I)A_I)
                       +t^2\det A_I.
\end{equation}
If exactly one constituent ordinary bridge joins $I$ to $I^c$,
the degree is at most one; if two do, it is at most two. If the momentum is restricted to a reflection character plane, its
two linear entries are the factors $\mu_I^\pm(t)$, each affine-linear.
At an adjacent channel $I=\{i,i+1\}$ only the moving chiral factor
is included. These principles recover the usual and reflected linear
removal equations~\cite{Galashin2024,KroviTessler2026}. 
For two-column boosts, after clearing negative powers, the removal
equation has degree at most four in $u^2$, where $u=e^\tau>1$;
the positive parameter $u$ is then recovered by taking the positive
square root. 
For disjoint vertex palindromes of
\eqref{eq:atoms:vertex-operations}, the determinant has degree at
most two in $t$ when one triple meets both $I$ and $I^c$, and
at most four when two triples do.
For a single orbit of disjoint three-vertex boosts, the two-cut
determinant numerator splits into factors of degree at most four
in $u$, possibly after adjoining fixed algebraic constants
(such as $i$ to split a sum of two squares).
We stress again that whenever $s_I$ factorizes, including when it is a perfect square, the removal equations are obtained from the relevant irreducible factore. 
The vertex palindrome and three-vertex boost of
\eqref{eq:atoms:vertex-operations} are included as well; the latter
also uses $u$, but need not depend only on $u^2$.
The degree bounds and parameter normalizations are collected in
Appendix~\ref{app:bcfw:degrees}.

\begin{remark}
Crossing a cut need not make every reduced factor of its
Mandelstam vary. We determine the hit components from
\eqref{eq:bcfw:removal}. 
Different operations hit
different boundaries and hence yield different recursions, just as a change of shift location in usual BCFW recursion.
\end{remark}

\subsection{From a factorization to a BCFW tile}
\label{subsec:bcfw:construction}

Each component algebraically hit by the chosen operation, corresponds via the recipe of Section~\ref{subsec:boundary:subspaces} to a factorizable subspace. Assume, recursively, that the factors (the core and blobs) themselves can be tiled via smaller symmetric positroid cells. Then we obtain, by taking fiber/direct products or a soft deletion, and  quotient the internal rescalings if needed, $(d-1)$-dimensional symmetric domain cells $\mathcal B$
with image dimension $d-1,$ and positive parameterizations thereof. Note that by symmetric positroid cells we mean different objects for the core and blobs: the core inherits an action of the symmetry group $G,$ while the blob has its local symmetry group, which is a subgroup of $\mathbb{Z}_2\times\mathbb{Z}_2$ and might be trivial.

Let $C_{\mathcal B}(\mathbf a)$, with
$\mathbf a\in\mathcal P_{\mathcal B}$, be the resulting parametrization
by independent recursive parameters, and write $T_M$ for the domain of the parameter $t$. Define
\begin{equation}
\label{eq:bcfw:cell}
 \begin{gathered}
 \beta_{\mathcal B,M}:
 \mathcal P_{\mathcal B}\times T_M\longrightarrow X^{\geq},
 \qquad
 (\mathbf a,t)\longmapsto C_{\mathcal B}(\mathbf a)M(t),\qquad
 \Pi_{\mathcal B,M}=\mathcal B\,M(T_M).
 \end{gathered}
\end{equation}
We use cell-building steps for which $\Pi_{\mathcal B,M}$ is the
whole symmetric positroid cell supplied by
Theorem~\ref{thm:dom:cells}, and retain dimension $d$ and full image
rank. The domain parametrization $\beta_{\mathcal B,M}$ need not
be birational, but is injective when restricted to the positive subspaces. 
The cells $\Pi_{\mathcal{B},M}$ are called \emph{BCFW} cells, and the depend on the choices of shift operations (in the current notation $M$ is the last operation, and $\mathcal B$ encapsulates the lower recursive choices). 

$\Pi_{\mathcal{B},M}$ is said to be a \emph{BCFW tile} if it maps injectively to the amplituhedron. In practice we are interested in finding a \emph{positive reconstruction} for a target point $y.$ 
Positive reconstruction is a preimage problem in the fiber of $y$: for a root $t\in T_M$ of
\eqref{eq:bcfw:removal}, one seeks $C_*\in X^{\geq}$ with
$\Phi(C_*)=y$ such that $C_*M(t)^{-1}\in\mathcal B$, using the
transferred external datum of \eqref{eq:bcfw:promoted}.
For a direct-sum term, the shifted total momentum must vanish on
every prescribed block, not merely have zero determinant.
For a soft term, the reconstructed lower source must have the
prescribed loops and coloops, forcing the corresponding
$\lambda$- or $\widetilde\lambda$-columns to vanish.
These conditions are not implied by a root of
$\det Q_I^M(t)=0$ alone, and follow our finer geometric/combinatorial analysis of Section~\ref{sec:boundaries}.
\subsection{The BCFW tiling conjecture}
\label{subsec:bcfw:conjectures}
\paragraph{BCFW tilings.}
\begin{conjecture}[BCFW tiling by whole cells]
\label{conj:bcfw:tiling}
Fix $N,k,G,\rho,$ and external data satisfying \eqref{eq:ext:requirements}. 
Let $\mathfrak T$ be a collection of domain cells constructed recursively via the above process.
Then every domain cell $\Pi\in\mathfrak T$ has dimension $d=\dim\SymAmp{G}{\rho}_{N,k}$ and maps injectively to $\SymAmp{G}{\rho}_{N,k},$
and
\begin{equation}
\label{eq:bcfw:tiling}
 \Phi(\Pi)\cap\Phi(\Pi')=\varnothing\quad(\Pi\ne\Pi'\in\mathfrak{T}),
 \qquad
 \mathcal M=\bigcup_{\Pi\in\mathfrak T}\overline{\Phi(\Pi)}.
\end{equation}
\end{conjecture}
In particular, a shift which induces a polynomial of degree greater than $1$ does not prescribe several
tiles or a cover of a tile. In all our experiments, in these cases, there was a unique root which provided a positive preimage, 
and the conjecture asserts positive injectivity on the whole resulting cell even beyond algebraic degree one.

\paragraph{Experimental evidence for BCFW tilings.}
We tested the proposed recursion in small cyclic, linear-reflection
and dihedral models, with and without exchange, as well as in the
reflected and rotated Lagrangian models. In the two Lagrangian
families the tests reached $(N,k)=(10,5)$, and in the rotated
family they continued to $(12,6)$.
The experiments used several choices of positive external data,
without imposing immanant positivity as a general restriction.
Targets were generated by sampling points in the positive domain,
independently of the candidate cells, and their preimages were
sought in every cell of a fixed recursive collection.
Reconstructed sources were checked for nonnegativity, membership
in the prescribed cell, and equality of both projected target
planes. For the retained collections in the completed coverage
tests, each sampled generic target had exactly one preimage in
exactly one cell. We also compared different recursive choices
and tested points close to boundaries. In small cases, these
checks were supplemented by symbolic inverse formulas and exact
descriptions of the nonnegative fibers.

The experiments reveal some exotic features, which do not occur in the study of previously known amplituhedra. In the linear vertex-$D_3$ model at
$(N,k)=(6,3)$, two distinct domain cells map injectively onto the
same open image. The explanation is rather elegant: several tangential directions of a common parent cell of the two domain cells are contained in representations, which belong to the kernel of the amplituhedron map, and hence the map is blind to these directions. In the linear $D_3$ model at $(9,4)$, different
shifts give candidate tilings with one and two cells, respectively,
both passing the coverage tests. Thus neither the source cell
representing a tile nor the cardinality of a tiling is determined
by the image alone. Higher-degree removal provides another
substantial test: in a twelve-label reflection-and-exchange
example, an irreducible quartic equation yields a unique positive
lower reconstruction at each audited target. Exact positivity
certificates exclude the other admissible shift roots throughout
the regular source fiber, not merely at the source used to
generate the target; see Appendix~\ref{app:bcfw:degrees}.
These results support the conjecture, but of course do not substitute a proof.
\subsection{Two symmetric examples}
\label{subsec:bcfw:examples}

\paragraph{Fraser's cyclic model.}
Let $N=p\ell$, $\ell\geq2$, $k=\alpha p+\beta$, $0\leq\beta<p$,
with linear rotation $r(i)=i+\ell$. Use the folded ordinary bridge
\begin{equation}
\label{eq:bcfw:Fraser-bridge}
 M_a(t)=\prod_{j=0}^{p-1}x_{a+j\ell}(t)
\end{equation}
(or its transpose), with the inherited seam signs. For
$I=[b,b+m-1]$ and $2\leq m<\ell$, a contribution to the BCFW process can occur when
a cut of $I$ has residue $a\pmod\ell$; exactly one bridge member
then crosses $I$ and its removal is linear. At $m=\ell$, the
two cuts have the same residue. When this is $a$, the determinant
equation has degree at most two, with quadratic coefficient
$\det A_I$ in \eqref{eq:bcfw:quadratic-pencil}; the cyclic
representations can force this coefficient to vanish.
\begin{example}
Here is a small example of quadratic shift. $K_{k,N}$ represents, as usual, the cyclically symmteric point of $\Gr_{\geq}(k,N).$ Take $(p,\ell,k)=(4,3,2)$, the input
$C=K_{2,12}$, $\mathcal U=0$, $\mathcal W=K_{4,12}$ gives,
up to a nonzero constant,
\[
 s_{\{1,2,3\}}^{M_3}(t)
 \ \propto\ \frac{\sqrt3}{2}
       +\frac{\sqrt6-3\sqrt2}{4}\,t-\frac12t^2.
\]
\end{example}
For a representative blob of rank $b_0$ on $m+1$ labels, the
factorization cell and the corresponding BCFW cell have the schematic form
\begin{equation}
\label{eq:bcfw:Fraser-cell}
 \begin{gathered}
 \mathcal B_{I,b_0}=\Pi_c\star_{C_p}\Pi_{\rm mom},\qquad
 (N_c,k_c)=\bigl(p(\ell-m+1),\,k-p(b_0-1)\bigr),\\
 \Pi_{I,b_0}=\mathcal B_{I,b_0}M_a(\mathbb R_{>0}).
 \end{gathered}
\end{equation}
The star denotes gluing $p$ equally parametrized copies of the
ordinary blob to the core. The local ranks and image dimensions
are those of Section~\ref{subsec:boundary:subspaces}. For $m=2<\ell$, the moving branches for $M_a(t)$ are
$\langle a+1,a+2\rangle_t=0$ with blob rank $b_0=1$, and
$[a-1,a]_t=0$ with blob rank $b_0=2$, together with their
rotational translates. Transposing the bridge exchanges the
two chiralities. When $m=\ell=2$, both cuts are crossed,
so both chiral factors must be checked separately.
At the critical length $m=\ell$, the core has $p$ labels.
For $0<\beta<p$, it is the \emph{zero-dimensional} rotationally
symmetric point $K_{\beta,p}$:
\begin{equation}
\label{eq:bcfw:Fraser-zero-core}
 \Pi_c=\{K_{\beta,p}\},\qquad b_0=\alpha+1.
\end{equation}
The point still fixes the connector incidence and representations,
and is retained in the critical construction.
For $\beta=0$, the rank equation instead gives
$(k_c,b_0)=(0,\alpha+1)$ or $(p,\alpha)$.
The core then forces all connecting coordinates to zero or
leaves them unconstrained, respectively. In either case,
deleting the connecting labels gives a direct sum of rank-$\alpha$
planes on the $\ell$-label sectors, with $Q_J=0$ on each sector.
Its image dimension is therefore computed by
\eqref{eq:bdy:image-dimension} with $g_{\mathcal T}=0$.
The resulting tiling is an instance of
Conjecture~\ref{conj:bcfw:tiling}.

\paragraph{Shevchenko's half-turn model.}
Put $N=2n$, $k=n$, and use the half-turn $\vartheta(i)=i+n$ of
Section~\ref{subsec:bg:Shevchenko}. The domain is the fixed locus
of $S_n^n\mathcal D$, a half-turn followed by an exchange.
We denote its image by $\RotAmp_n$ and use the paired bridge
$M(t)=M_1^{\mathrm{rot},\rightarrow}(t)$ of
\eqref{eq:bg:rotated-atoms}:
\begin{equation}
\label{eq:bcfw:Shevchenko-bridge}
 M(t)=x_1(t)y_{n+1}(t)
     =I+t e_{12}+t e_{n+2,n+1}.
\end{equation}
For $n\geq3$, the contributing interval families, modulo half-turn
and complement, are
\begin{equation}
\label{eq:bcfw:Shevchenko-hits}
 I_m^-=[2,m+1],\qquad I_m^+=[n+2-m,n+1]\quad(2\leq m<n),
 \qquad H=[2,n+1].
\end{equation}
For $m=2$, keep one moving chiral component of each displayed
interval (the angle-bracket component for
\eqref{eq:bcfw:Shevchenko-bridge}; its half-turned partner is
square-bracket). For $3\leq m<n$, use ordinary blob ranks
$2\leq b_0\leq m-1$. The half-circle contributes the sectors
$2\leq b_0\leq n-1$. The short-channel equations are linear;
the half-circle has a symmetric momentum pencil and a
generally quadratic determinant equation.

The domain factorization is
\begin{equation}
\label{eq:bcfw:Shevchenko-cell}
 \begin{gathered}
 \mathcal B_{I,b_0}
   =\Pi_{\rm mom}\star_*\Pi^{\rm rot}_{n-m+1}
                \star_{\vartheta(*)}\Pi_{\rm mom}^{\vartheta,\vee},\\
 \Pi_{I,b_0}=\mathcal B_{I,b_0}M(\mathbb R_{>0}).
 \end{gathered}
\end{equation}
The half-turned color-dual blob shares the parameters of the
first blob. At $m=n$, the middle factor is the connector-only
$\LG_{\geq}(1,2)^{\mathrm{rot}}$, which constrains the internal
momentum to be symmetric and null.

\begin{figure}[!ht]
\centering
\begin{minipage}[b]{0.48\linewidth}\centering
\includegraphics[height=52mm]{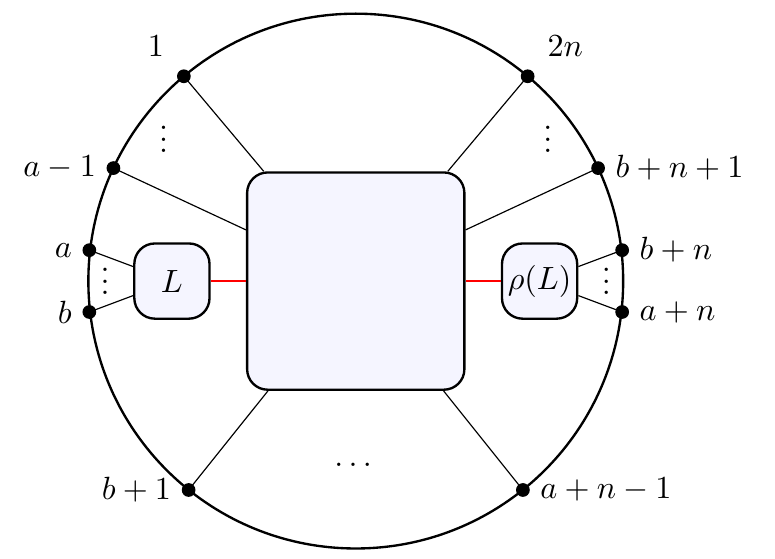}\\[-1mm]
\small $2\leq |I|<n$
\end{minipage}\hfill
\begin{minipage}[b]{0.48\linewidth}\centering
\includegraphics[height=52mm]{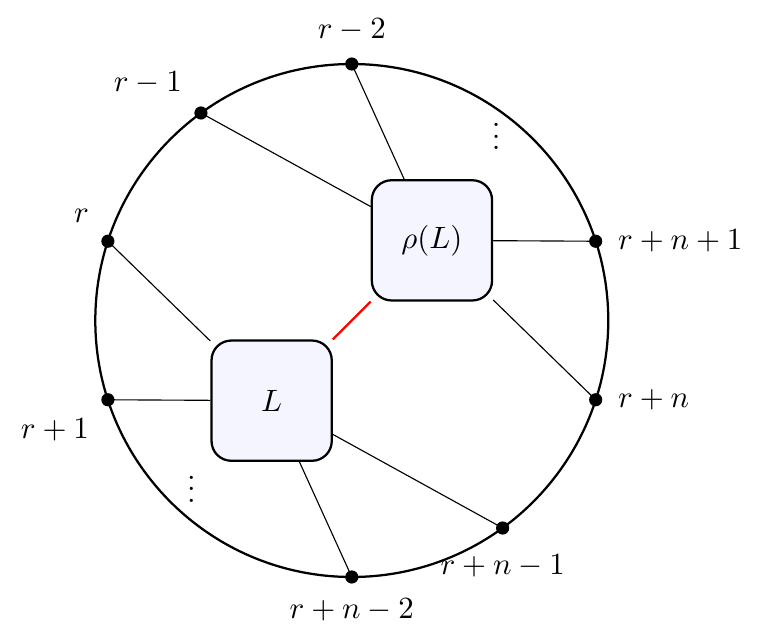}\\[-1mm]
\small $|I|=n$
\end{minipage}
\caption{The rotational factorization tangles. Here $L$ is an ordinary
blob; we write $\rho(L)=\vartheta(L)^\vee$ for its half-turned
color-dual. The two blobs share one parameter set. The middle rotational
core in the left diagram becomes connector-only in the right. The normal
pair \eqref{eq:bcfw:Shevchenko-bridge} is then attached.}
\label{fig:bcfw:rotational-factorization}
\end{figure}

\begin{figure}[!ht]
\centering
\includegraphics[width=0.94\linewidth]{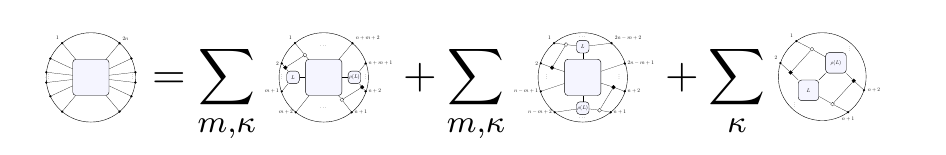}
\caption{One step for \eqref{eq:bcfw:Shevchenko-bridge}: the two
short-interval families and the half-circle family. The sums include
the admissible ranks and the selected adjacent ends; the lower pieces
are subsequently expanded.}
\label{fig:bcfw:rotational-recursion}
\end{figure}

A single recursive step gives $(n-2)^2+2$ contributions, before
the lower factors are themselves expanded.

Although differential forms are not pursued here, we also tested
the pushforwards of the natural logarithmic bridge forms, which are the candidate canonical forms, in this
model at ranks three and four. The resulting sums satisfy cyclic
covariance when all distinct algebraic inverse branches are
included once; retaining only one branch in a genuinely two-sheeted
sector fails this test. A systematic study of these forms and
their interpretation is deferred to future work.

\appendix
\section{Dimensions of the symmetric Grassmannians}
\label{app:domain-dimensions}
For convenience we evaluate \eqref{eq:dom:character-dimension} for the
families of Section~\ref{subsec:symmetry-families}. Write
$F_{p;k,N}=d^{\mathrm{cyc}}_{p;k,N}$ from
\eqref{eq:bg:Fraser-component}, with $F_{1;k,N}=k(N-k)$, and set
\begin{equation}
\label{eq:dom:axis-sign}
 \eta_\rho=\begin{cases}
 0,&\ell\text{ odd},\\
 1,&\ell\text{ even, edge type},\\
 -1,&\ell\text{ even, vertex type},
 \end{cases}
 \qquad \varepsilon_m=m\bmod2\in\{0,1\}.
\end{equation}
For exchange families put $N=2n=p\ell$, $k=n$, and $F=F_{p;n,2n}$.
The two cyclic exchange dimensions are
\begin{equation}
\label{eq:dom:cyclic-exchange-dimensions}
 \begin{aligned}
 A_{p,\ell}&=\begin{cases}
 p\ell^2/8,&p,\ell\text{ even},\\
 p(\ell^2-1)/8,&p\text{ even},\ \ell\text{ odd},\\
 (p\ell^2-2\ell)/8,&p\text{ odd}\ (\ell\text{ even}),
 \end{cases}\\
 S_{p,\ell}&=\begin{cases}
 p\ell^2/4,&p\equiv0\pmod4,\\
 p\ell^2/4+\ell/2,&p\equiv2\pmod4.
 \end{cases}
 \end{aligned}
\end{equation}
For a reflection $h$ on the even label set, let $\eta(h)=1$ for an edge
axis and $-1$ for a vertex axis. In a mixed family write
$\eta_0=\eta(s)$ and $\eta_1=\eta(rs)=(-1)^\ell\eta_0$.
Then the domain dimensions are
\begin{equation}
\label{eq:dom:family-dimensions}
\renewcommand{\arraystretch}{1.22}
\begin{array}{l|l}
 \text{family}&d_\rho\\\hline
 \text{linear cyclic}&F_{p;k,N}\\
 \text{linear dihedral}&\tfrac12(F_{p;k,N}-\varepsilon_k\eta_\rho)\\
 \text{symmetric ABJM}&A_{p,\ell}\\
 \text{symmetric Shevchenko}\ (p\text{ even})&S_{p,\ell}\\
 \text{reflection-twisted}&\tfrac12 F+\tfrac12 n\eta_\rho\\
 \text{dihedral ABJM}&\tfrac12 A_{p,\ell}+\tfrac14(n-\varepsilon_n)\eta_\rho\\
 \text{mixed phase }0\ (p\text{ even})&
       \tfrac12 S_{p,\ell}+\tfrac14(-\eta_0\varepsilon_n+\eta_1n)\\
 \text{mixed phase }1\ (p\text{ even})&
       \tfrac12 S_{p,\ell}+\tfrac14(\eta_0n-\eta_1\varepsilon_n).
\end{array}
\end{equation}
To check the evaluation, a linear lift $S_k^b$ has character
$\chi_K(S_k^b)=\sum_{t=0}^{k-1}z_t^b$; an exchanging lift satisfies
$(S_n^bA_N)^2=(-1)^bS_n^{2b}$. Substitution in
\eqref{eq:dom:character-dimension} and summation of roots of unity gives
\eqref{eq:dom:cyclic-exchange-dimensions}. The reflection cosets are
accounted for by
\begin{equation}
\label{eq:dom:reflection-tangent-traces}
 \operatorname{tr}_{T_K\Gr(n,2n)}(h)=-\eta(h)\varepsilon_n,
 \qquad
 \operatorname{tr}_{T_K\Gr(n,2n)}(h\mathcal D)=\eta(h)n.
\end{equation}
These follow respectively from \eqref{eq:dom:reflection-trace} and the
signed square of $\widehat hA_N$. Their average over all reflections is
governed by $p^{-1}\sum_{b=0}^{p-1}\eta(r^bs)=\eta_\rho$, giving the
last four rows. At arbitrary rank a linear reflection has tangent trace
zero when $N$ is odd or $k$ is even, and $-\eta(h)$ otherwise; this gives
the linear-dihedral row. Thus the displayed dimensions include all rank
and axis parities allowed by the symmetry datum. They are dimensions of
the domain symmetric Grassmannian.
\section{Proof of the local momentum normal forms}
\label{app:bdy:determinant-types}
We prove Table~\eqref{eq:bdy:determinant-types} using the spinor types
in \eqref{eq:amp:spinor-exact-sequences}--\eqref{eq:amp:spinor-modes}.
We will assume that the indicated setwise symmetries are imposed.

\emph{Linear reflection.}
The reference mode pairs $\{0,k-1\}$ and $\{-1,k\}$ in
\eqref{eq:ext:mode-windows} are exchanged by a signed linear reflection.
The constant spinor types therefore have one $+1$ and one $-1$
eigenline each. Choosing their eigenbases gives
$R_{\lambda,h}=R_{\widetilde\lambda,h}=D_2:=\operatorname{diag}(1,-1)$.
For $hI=I$, \eqref{eq:bdy:channel-stabilizer} becomes
\begin{equation*}
 D_2QD_2=Q,\qquad
 D_2\begin{pmatrix}q_{11}&q_{12}\\q_{21}&q_{22}\end{pmatrix}D_2
   =\begin{pmatrix}q_{11}&-q_{12}\\-q_{21}&q_{22}\end{pmatrix}.
\end{equation*}
Thus $Q$ is diagonal and its determinant is the product of its two
entries. For $hI=I^c$, the right side is $-Q$ instead: the diagonal
entries vanish and $\det Q=-q_{12}q_{21}$.

\emph{A single exchanging symmetry.}
Let $\gamma$ be pure exchange or a reflection with exchange, and write
$T_\gamma^2=\sigma I$. Choose a basis of the $\widetilde\lambda$-plane and use its images
under $T_\gamma$ as the basis of the $\lambda$-plane.
In these bases, $R_{\lambda,\gamma}=I_2$.
Applying $T_\gamma$ a second time multiplies each original basis
vector by $\sigma$, since $T_\gamma^2=\sigma I$; hence
$R_{\widetilde\lambda,\gamma}=\sigma I_2$.
Consequently a setwise-fixed interval satisfies $Q=\sigma Q^T$.
Pure exchange has $T_\gamma=A_N$ and $\sigma=1$.
For reflection with exchange, $T_\gamma=\widehat h A_N$ and
\begin{equation*}
 \widehat h A_N=\sigma A_N\widehat h,\qquad
 \sigma=\begin{cases}
 +1,&\text{vertex axis},\\
 -1,&\text{edge axis}.
 \end{cases}
\end{equation*}
Indeed, on the even polygon the reflection respectively preserves or
reverses label parity; since $\widehat h^2=I$, this also gives
$T_\gamma^2=\sigma I$. Thus pure exchange and vertex exchange allow
symmetric $Q$, with determinant $q_{11}q_{22}-q_{12}^2$; edge exchange
allows $Q=\mu\left(\begin{smallmatrix}0&1\\-1&0\end{smallmatrix}\right)$,
with determinant $\mu^2$. This is where the signed square matters.

\emph{Reflection and pure exchange imposed separately.}
Keep the bases identifying pure exchange with transposition, so
$Q=Q^T$. Diagonalize the reflection on the first plane as $D_2$.
The same commutation relation with $A_N$ forces its matrix on the
second plane to be $\sigma D_2$, not an independently chosen $D_2$.
The simultaneous equations are therefore
\begin{equation*}
 Q=Q^T,\qquad Q=\sigma D_2QD_2.
\end{equation*}
For a vertex axis they give the diagonal plane; for an edge axis they
give $Q=\mu\left(\begin{smallmatrix}0&1\\1&0\end{smallmatrix}\right)$,
whose determinant is $-\mu^2$. These are the last two rows.

Without a nontrivial setwise symmetry there is no local linear
restriction on $Q$, giving the first row. This proves the table.
The full channel stabilizer and the orbit-sum relation may impose
further constraints. Finally, the two-particle identity follows from
Cauchy--Binet, and on any line $Q=\mu Q_0$ homogeneity gives
$\det Q=\mu^2\det Q_0$. 

\section{Reflection factors: selection and nonvanishing}
\label{app:reflection-factors}

\paragraph{Which reflection factor vanishes.}
The functions $\mu_I^\pm$ are the two factors in the
linear-reflection row of \eqref{eq:bdy:determinant-types}.
Let $h$ be a signed linear reflection preserving $I$ setwise.
In reflection eigenbases, ordered $(+,-)$ in both spinor planes,
\[
 Q_I=\begin{pmatrix}\mu_I^-&0\\0&\mu_I^+\end{pmatrix}.
\]
When $\kappa_C(I)=1$, the quotient $E_C$ of
\eqref{eq:bdy:connector-quotients} is a line.
Write $\chi_\eta$ for its signed one-dimensional representation,
where $h$ acts by $\eta\in\{+1,-1\}$ and the central $-I$
acts by $-1$.
As the proof below shows, a connector of sign $\eta$ makes
$\mu_I^\eta$ vanish and leaves $\mu_I^{-\eta}$ nonzero.
Thus the superscript labels the connector selecting the
vanishing factor, not the eigenline carrying the nonzero entry.

We write $[V]$ for the class of a representation $V$:
an equality of such classes is an equality of the multiplicities
of all irreducible representations, equivalently of characters.
Induction from a factor stabilizer assembles the representation
on all symmetry-related copies of that factor.

\begin{proposition}[Representation selection]
\label{prop:bdy:branch-selection}
Assume that the group acts without exchange.
Consider a connected core--blob configuration with outer intervals
$\mathcal J=\mathcal O(I)$, where $I$ is preserved by a linear
reflection. Require the connecting labels to be neither loops
nor coloops in their incident factors.
Let $C_c$ be the core and $C_b$ the representative augmented blob
on $I$, with connector character $\chi_\eta$.
Then
\begin{equation}
\label{eq:bdy:branch-character}
 [C]=[C_c]+\operatorname{Ind}_{H_I}^{\widehat G}[C_b]
             -\operatorname{Ind}_{H_I}^{\widehat G}[\chi_\eta],
 \qquad H_I=\operatorname{Stab}_{\widehat G}(I).
\end{equation}
For fixed parent, core and blob representation types, this permits
at most one connector character if the two induced connector
representations are nonisomorphic.
A permitted character must also be realizable by nonnegative
factors with the prescribed signed action on their external
and connecting labels.

Let $\mathcal S_\eta$ be the locus obtained by such a realized
connected construction. On its image,
\begin{equation}
\label{eq:bdy:selected-mu}
 \mu_{\mathrm{good}}=\mu_I^\eta=0,\qquad
 \mu_{\mathrm{other}}=\mu_I^{-\eta}\ne0.
\end{equation}
The selected factor also vanishes on
$\overline{\Phi(\mathcal S_\eta)}$; the other factor need not
remain nonzero on this closure.
\end{proposition}

\begin{proof}
\emph{The representation identity.}
For each blob interval $J$, let $\ell_J$ be the one-dimensional
space of values of its connecting coordinate.
The signed action transports these lines along the blob orbit.
In particular, $\ell_I$ has character $\chi_\eta$.
The common connecting coordinate on $C$ has kernel
$C[I]\oplus C[I^c]$, so $\ell_I$ identifies with $E_C$.

The coordinate-matching construction gives an exact sequence
of $\widehat G$-representations
\[
 0\longrightarrow C
 \xrightarrow{\ \iota\ }
 C_c\oplus\bigoplus_{J\in\mathcal J}C_{b,J}
 \xrightarrow{\ \Delta\ }
 \bigoplus_{J\in\mathcal J}\ell_J
 \longrightarrow0.
\]
Here $\iota(x)$ is the tuple of factor vectors determined by
$x$, including their connecting coordinates.
If the matching equation at $J$ is
$(u_J)_{*_J}=\varepsilon_J(u_c)_{*_J}$, then
\[
 \Delta\bigl(u_c,(u_J)_J\bigr)
   =\bigl((u_J)_{*_J}
             -\varepsilon_J(u_c)_{*_J}\bigr)_J.
\]
Thus $\ker\Delta$ consists exactly of the matching tuples.
Deleting their connecting coordinates identifies them with $C$,
as proved above. The map $\Delta$ is surjective because each
blob's connecting coordinate functional is nonzero.

The two direct sums are respectively the representations induced
from $C_b$ and $\chi_\eta$ at the stabilizer $H_I$.
Taking representation classes proves
\eqref{eq:bdy:branch-character}.
Theorem~\ref{thm:dom:ball} fixes the types on the chosen
nonnegative domains, with Fraser's distinguished component as
the cyclic input
\cite[Definition/Lemma~4.8]{Fraser2020}.
For fixed types, two different induced connector classes cannot
both equal the prescribed difference.
This is a necessary representation test; it does not by itself
construct nonnegative factors with the proposed connector sign.

\smallskip\noindent
\emph{The pairing and the vanishing factor.}
The pairing used here is the partial coordinate pairing from
Lemma~\ref{lem:bdy:vanishing}, descended to the quotients of \eqref{eq:bdy:connector-quotients}:
\[
 \overline\beta_I:E_C\times E_D\longrightarrow\mathbb R,
 \qquad
 \overline\beta_I(\bar x,\bar y)=\sum_{i\in I}x_i y_i,
 \qquad x\in C,\quad y\in D=C^\perp.
\]
Here $\bar x,\bar y$ denote quotient vectors. This pairing is well defined and perfect by the preceding
corollary. Its pullback to $L\times T$ is the bilinear form
represented by $Q_I$.

For any signed linear symmetry $\gamma$ preserving the unordered
cut $[I]$, one has
\[
 \overline\beta_I(\gamma\bar x,\gamma\bar y)
   =\delta_I(\gamma)\overline\beta_I(\bar x,\bar y).
\]
Indeed, relabelling replaces the sum over $I$ by the sum over
$g_\gamma^{-1}I$; if this is $I^c$, the sum changes sign because
$\langle x,y\rangle=0$.
Equivalently, $E_D$ is the dual representation of $E_C$ twisted
by $\delta_I$.

For our setwise reflection $h$, $\delta_I(h)=1$.
At connectivity one, both quotient lines therefore have
reflection sign $\eta$.
The projections of $L$ and $T$ onto these lines are surjective
by \eqref{eq:bdy:positive-kernel-rank} and its dual.
Only their $\eta$-eigenlines can contribute, and their pairing
is nonzero. Thus the $(\eta,\eta)$ entry of $Q_I$ is nonzero,
while the other diagonal entry vanishes.
With the displayed ordering, these entries are respectively
$\mu_I^{-\eta}$ and $\mu_I^\eta$, proving
\eqref{eq:bdy:selected-mu}.
Continuity gives vanishing of the selected factor on the closure
of the factorization image.
\end{proof}

For a linear reflection exchanging $I$ and $I^c$, the same
pairing law has $\delta_I(h)=-1$: a connector of sign $\eta$
pairs with a dual connector of sign $-\eta$.
The matrix is off-diagonal, and the same convention labels
its vanishing entry by $\mu_I^\eta$.
Such a reflection belongs to the stabilizer of $[I]$, not to
the setwise stabilizer $H_I$ used in
\eqref{eq:bdy:branch-character}.

When exchange is present, the representation identity is taken
on the linear subgroup $\widehat H$, summing over its blob
orbits. Exchanging elements additionally relate the factors
and interchange their two spinor sectors.

\begin{remark}
The representation identity \eqref{eq:bdy:branch-character}
need not determine the reflection sign on the representative
connector: the two distinct local characters $\chi_+$ and $\chi_-$
may induce isomorphic $\widehat G$-representations.
In that case, one must also check which sign is compatible with
the prescribed nonnegative blob, as the following example shows. At linear vertex
$D_2$, $(N,k)=(8,4)$, take $r=S_4^4$, $h=S_4J_8$, so $r^2=-I$.
For $H_I=\{\pm I,\pm h\}$ and
$\chi_\pm(-I)=-1$, $\chi_\pm(h)=\pm1$, one has
\begin{equation}
\label{eq:bdy:induction-coincidence}
 \operatorname{Ind}_{H_I}^{\widehat G}\chi_+
   \simeq\operatorname{Ind}_{H_I}^{\widehat G}\chi_-.
\end{equation}
Conjugation by $r$ exchanges $h$ with $-h$. Nevertheless, on a positive
rank-two blob with labels $(4,5,6,*)$, $h$ swaps $4,6$ and fixes $5$
with sign $+$. If its connector has sign $\eta$, positivity of
$\Delta_{46}$ gives determinant $-1$ for the row action, while positivity
of $\Delta_{5*}$ gives determinant $\eta$. Thus $\eta=-1$.
This selects the factor for that construction; it is not a global
exclusion for all source-plane types.
\end{remark}

\paragraph{When the other factor is globally nonzero.}
We use the same perfect pairing
$\overline\beta_I:E_C\times E_D\to\mathbb R$, now without
assuming $\kappa=1$.
Its transformation law pairs a character $\chi$ of the signed
linear channel stabilizer with $\delta_I\chi^{-1}$.
If the entire subspaces carrying these two characters are
one-dimensional, the pairing between them is perfect.
Consequently, if $\mu_{\mathrm{other}}$ is obtained by pairing
nonzero projected spinor vectors on these lines, it cannot vanish.
When these conditions hold throughout the fixed nonnegative
domain, they prove nonvanishing on the entire amplituhedron.
Here is one uniform application.
\begin{corollary}[A half-turn and gap reflection at rank three]
\label{cor:bdy:global-mu-rank-three}
Let $N=2m\geq6$, $k=3$, and suppose the linear symmetry contains
$r=S_3^m$ and $h=J_N$. For $I=\{1,\ldots,m\}$ and all external data
satisfying \eqref{eq:ext:requirements}, the factor not selected by the
connectivity-one construction is nonzero on the \emph{entire}
amplituhedron. The same holds at $k=N-3$ by positive duality, with the
factor labels transported.
\end{corollary}
\begin{proof}
Write $\chi_{\alpha\beta}$ for the character with
$r\mapsto\alpha$, $h\mapsto\beta$. The already fixed Fourier windows give
\begin{equation}
\label{eq:bdy:rank-three-characters}
 C\simeq\chi_{++}\oplus\chi_{-+}\oplus\chi_{--},\qquad
 L\simeq\chi_{-+}\oplus\chi_{--},\qquad
 T\simeq\chi_{++}\oplus\chi_{+-}.
\end{equation}
Both $r,h$ interchange $I,I^c$. Thus $\kappa$ is odd and at most three.
If $\kappa=1$, the supported summands have trace zero for $r,h$, giving
$E_C=\chi_{-+}$; positive rank gives nonzero projected spinors.
If $\kappa=3$, $E_C=C$, and both spinor projections inject.
In either case the $\chi_{-+}$ slot of $E_C$ and its paired
$\chi_{+-}$ slot of $E_D$ are one-dimensional and are filled by the
respective spinor lines. Their pairing is the surviving factor at
$\kappa=1$, hence is $\mu_{\mathrm{other}}$, and is nonzero everywhere.
Positive duality exchanges the spinor planes and transposes $Q_I$,
proving the corank-three statement.
\end{proof}
This conclusion persists on further symmetric fixed loci. A single
gap reflection at arbitrary odd rank still selects the $+$ connector
at $\kappa=1$ by its trace, but that alone does not control the pairing
at larger connectivity. Global zeros of $\mu_{\mathrm{other}}$ in
higher-multiplicity cases, including those not resolved by
\eqref{eq:bdy:induction-coincidence}, will be studied in a sequel.
On the direct-sum subspaces both factors vanish; their image dimensions
are counted below and are not a separate mechanism.

\begin{remark}
\label{rem:bdy:strong-positivity}
Under the strong (immanant-positive) external-data assumption,
Galashin's connectivity criterion gives
\[
 s_I\bigl(\Phi(C)\bigr)=0
 \quad\Longleftrightarrow\quad \kappa_C(I)\leq1
\]
for cyclic intervals $I$
\cite[Corollary~6.21 and Lemma~6.24]{Galashin2024}.
In fact, immanant nonnegativity already suffices.
If all realized connectivity-one domain pieces select the same factor,
then the other factor can vanish only on the direct-sum image.
For an odd-rank gap reflection interchanging $I,I^c$, the trace selects
one character and odd connectivity excludes a direct sum; under this
additional hypothesis the other factor is nowhere zero for every odd
rank. This remark is not used in the vanishing, rank or image-dimension arguments of this paper.
\end{remark}

\section{Shift degrees and reduced removal equations}
\label{app:bcfw:degrees}
Fix a target $y=(\lambda,\widetilde\lambda)$ and an allowed operation
$M$. The removal calculation of \eqref{eq:bcfw:promoted} uses
\begin{equation}
\label{eq:bcfw:appendix-removal}
 Q_I^M(t)=\lambda M(t)^{-1}P_I M(t)\widetilde\lambda^T,
 \qquad s_I^M(t)=\det Q_I^M(t),
\end{equation}
where $P_I$ is the diagonal coordinate projector. We use $t>0$
for additive operations and write the same expressions in
$u=e^\tau>1$ for boosts. Note that
\eqref{eq:bcfw:promoted} describes the motion of the external data.

We first bound the determinant numerator and then pass to the
reduced factors used for removal. Structural factorizations are taken
on complexified kinematics, keeping the target coordinates as variables;
numerical algebraic degrees are stated separately over their field of
arithmetic. Symmetry and source rank can lower the polynomial degrees. For example, in the linear $D_3$ model at $(N,k)=(9,4)$, the
two-cut removal equation reduces to a product of two linear
factors in $u^2$ after discarding the denominator-clearing
monomial; see Table~\ref{tab:bcfw:reported-polynomials}.
In Table~\ref{tab:bcfw:degrees}, $r_I$ counts the disjoint consecutive
two- or three-label blocks of the operation which meet both $I$
and $I^c$. For a cyclic interval, $r_I\leq2$; if $r_I=0$, the
partial momentum is constant. For operations consisting of
two-column blocks, $r_I$ is the number of edges crossing the cut.
Note that these are degrees \emph{before} factorization of the
Mandelstam.

\begin{table}[htbp]
\centering\small
\renewcommand{\arraystretch}{1.3}
\begin{tabular}{>{\raggedright\arraybackslash}p{0.27\linewidth}>{\raggedright\arraybackslash}p{0.10\linewidth}>{\raggedright\arraybackslash}p{0.23\linewidth}>{\raggedright\arraybackslash}p{0.28\linewidth}}
\hline
Operation & Parameter & Partial momentum & Determinant numerator\\\hline
One ordinary arrow & $t>0$ & $Q_0+tR$, $\operatorname{rank}R\leq1$ & degree $\leq1$ in $t$\\
Commuting ordinary packet & $t>0$ & degree $\leq1$ in $t$ & degree $\leq1$ for one crossed arrow; $\leq2$ otherwise\\
Disjoint two-column boosts & $u>1$ & Laurent exponents $-2,0,2$ & $u^{2r_I}s_I^M(u)$; degree $\leq2r_I$ in $u^2$\\
Disjoint vertex palindromes & $t>0$ & degree $\leq2$ in $t$ & degree $\leq2r_I$ in $t$\\
Disjoint three-vertex boosts & $u>1$ & Laurent exponents $[-2,2]$ & $u^{2r_I}s_I^M(u)$; degree $\leq4r_I$ in $u$\\\hline
\end{tabular}
\caption{Degree bounds before factorization, with the boost normalizations
of \eqref{eq:atoms:orthogonal-layer} and \eqref{eq:atoms:vertex-boost}.
For two-column boosts, degree four in $u^2$ means degree eight in $u$.
The two-block three-vertex bound is refined by
Proposition~\ref{prop:bcfw:no-primitive-octic}.
Roots introduced by clearing negative powers are discarded.}
\label{tab:bcfw:degrees}
\end{table}

\paragraph{Derivation of the bounds.}
For ordinary packets we use \eqref{eq:bcfw:quadratic-pencil}.
For disjoint two-column boosts, direct calculation gives
\begin{equation}
\label{eq:bdy:hyperbolic-Laurent}
 Q_I^M(u)=Q_{I,0}+u^2R_{I,+}+u^{-2}R_{I,-}.
\end{equation}
Suppose exactly one block crosses the cut, say $e=\{a,b\}$,
with $a\in I$ and $b\in I^c$. This does not require $|I|=2$;
all other blocks leave their contributions to $Q_I$ unchanged.
The coefficients of $u^{\pm2}$ are
\[
 4R_{e,\pm}=(\lambda_a\mp\sigma_e\lambda_b)
 (\widetilde\lambda_a\pm\sigma_e\widetilde\lambda_b)^T.
\]
Each is an outer product, so
$\det R_{e,+}=\det R_{e,-}=0$.
Consequently, the $u^4$ and $u^{-4}$ terms in
$\det Q_I^M(u)$ vanish when there is only one crossed edge.
Together with the two-edge bound, this gives
\begin{equation}
\label{eq:bdy:hyperbolic-degree}
 \mathcal P_I^M(u):=u^{2r_I}s_I^M(u)\in\mathbb C[u^2],
 \qquad \deg_{u^2}\mathcal P_I^M\leq2r_I.
\end{equation}
Here $\deg_{u^2}$ means degree as a polynomial in $u^2$.
At two crossed edges, the extreme coefficients are sums of two
rank-one matrices and degree four in $u^2$ is possible.
A root with $u^2>1$ gives $u>1$ by the positive square root.
The raw degree in $u$ can still be eight, as in the two-column
example below.
Since $u\mapsto u^2$ is bijective for $u>1$, uniqueness of an
admissible value of $u^2$ is equivalent to uniqueness of an
admissible $u$. Thus finding the shift requires solving at most
a quartic in $u^2$ and taking the positive square root, together
with the prescribed positive-reconstruction checks.

For the palindrome on $(a,b,c)$, let $d_j=\mathbf1_{j\in I}$.
In a non-seam chart, direct multiplication gives
\begin{equation}
\label{eq:bcfw:palindrome-selector}
 \begin{aligned}
 U(t)^{-1}P_IU(t)-P_I
  ={}&t(d_a-d_b)e_{ab}+t(d_b-d_c)e_{bc}\\
    &+\tfrac12t^2(d_a-2d_b+d_c)e_{ac}.
 \end{aligned}
\end{equation}
If exactly one or two labels of the triple lie in $I$, its
contribution to $Q_I^M$ is a moving rank-one matrix with entries
of degree at most two in $t$, or the full triple momentum minus
such a matrix. The full triple momentum is unchanged.
The determinant is affine-linear in a rank-one contribution,
so one such triple gives degree at most two in $t$, not four.
Two such triples allow a cross-term of degree four.
The inherited cyclic signs do not change these bounds.

For the normalized three-vertex boost, the generator has eigenvalues
$-1,0,1$. Its conjugated selector has Laurent exponents between $-2$
and $2$ in $u$. The same rank-one argument gives degree at most four
after multiplication by $u^2$ for one triple meeting both sides.
Two such triples give the raw bound eight after multiplication by
$u^4$. Unlike the two-column case, odd powers of $u$ need not vanish.
For a single symmetry orbit of arrows forming disjoint three-vertex
boosts, Proposition~\ref{prop:bcfw:no-primitive-octic} reduces the
two-triple numerator to factors of degree at most four in $u$
on complexified kinematics.

For an overlapping orbit use the full generator $X_{\mathcal O}$.
The expansion terminates when $X=X_{\mathcal O}$ is nilpotent:
\begin{equation}
\label{eq:bcfw:full-generator-degree}
 e^{-tX}P_Ie^{tX}
    =\sum_{j\geq0}\frac{(-t)^j}{j!}\operatorname{ad}_X^j(P_I),
 \qquad \operatorname{ad}_X(A)=[X,A].
\end{equation}
For a diagonalizable generator with eigenvalues $\nu_\alpha$,
its conjugated selector has spectral exponents
$\nu_\beta-\nu_\alpha$. The actual spectrum and an algebraic
parameter must be determined before clearing denominators.
The disjoint-block bounds do not apply to an arbitrary overlapping
orbit, and we did not try to obtain a uniform degree bound for such orbits for this write-up.
\paragraph{Reduced factors before roots.}
On a reflection character plane, first select the component
$\mu_I^\eta=0$ specified by the domain factorization. Under an ordinary
packet each $\mu_I^\eta(t)$ is affine-linear, even when their product
is quadratic. If the other factor $\mu_I^{-\eta}(t)$ is a nonzero
constant, it gives no removal root. If
$Q_I^M(t)=\mu_I^M(t)Q_*$ with $Q_*$ independent of $t$ and
$\det Q_*\ne0$, solve $\mu_I^M=0$; its square in the determinant
is a multiplicity, not another component.
In the two-column case, if a degree-two polynomial $F$ satisfies
$F(u^2)\mid u^2Q_I^M(u)$ entrywise, then
\begin{equation}
\label{eq:bdy:hyperbolic-square}
 u^2Q_I^M(u)=F(u^2)M_*,\qquad
 u^4s_I^M(u)=\det(M_*)F(u^2)^2,
\end{equation}
with $M_*$ constant: each entry of $u^2Q_I^M(u)$ has degree at
most two in $u^2$. For $\det M_*\ne0$, the reduced equation
$F(u^2)=0$ is equivalent to $Q_I^M(u)=0$, whether or not $F$
factors further. For direct-sum or soft terms, the full
momentum-zero or loop--coloop conditions must still be checked,
as in Section~\ref{subsec:bcfw:construction}.

\begin{table}[htbp]
\centering\small
\renewcommand{\arraystretch}{1.3}
\begin{tabular}{>{\raggedright\arraybackslash}p{0.46\linewidth}>{\raggedright\arraybackslash}p{0.43\linewidth}}
\hline
Model and channel & Numerator $u^4s_I^M(u)$ and reduction\\\hline
Linear $D_3$, $(N,k)=(9,3)$, maximal sector
 &$c\bigl(\ell_1(u^2)\ell_2(u^2)\bigr)^2$;
 the reduced quadratic splits into linear factors\\
Linear edge $D_2$, $(8,4)$, and linear $D_3$, $(12,6)$, maximal sectors
 &$cF(u^2)^2$; $F$ is a generically irreducible quadratic\\
Linear $D_3$, $(9,4)$, two shifted cuts
 &$c\,u^2\ell_1(u^2)\ell_2(u^2)$; after discarding $u^2$,
 the removal equation is a product of two linear factors\\\hline
\end{tabular}
\caption{Symbolic fixed-edge two-column calculations. Here
$\ell_1,\ell_2$ are linear, $F$ is quadratic, and $c\ne0$ is
independent of $u$; their coefficients are case-dependent.
Irreducibility is over the rational function field, with complex
constants, of the symmetry-adapted spinor coordinates subject to
momentum conservation;
specializations may split the quadratic or lower its degree.
The factor $u^2$ in the last row is introduced by denominator clearing
and gives no root for $u>1$.}
\label{eq:bdy:hyperbolic-examples}
\label{tab:bcfw:reported-polynomials}
\end{table}

\paragraph{A two-column equation of degree four in $u^2$.}
Take
$(N,k)=(12,6)$, $r=S_6^3$, and $G=\langle r,\mathcal D\rangle$,
with no reflection imposed, and set
\[
 M(u)=\mathsf H_3(\log u)\mathsf H_6(\log u)
       \mathsf H_9(\log u)\mathsf H_{12}(\log u),\qquad I=\{1,2,3\}.
\]
The seam block has its inherited negative sign. All eight arrows
belong to one symmetry orbit, and both cuts of $I$ are crossed.

Pick
$\theta_i=(i-1)\pi/12$,
$c_m=(\cos(m\theta_i))_{i=1}^{12}$ and
$s_m=(\sin(m\theta_i))_{i=1}^{12}$, and take
\[
 \begin{aligned}
 C&=\operatorname{span}(c_1,s_1,c_3,s_3,c_5,s_5),\\
 \mathcal U&=\operatorname{span}\bigl(c_1,s_1,
           c_3-\tfrac1{10}c_5-\tfrac1{20}s_5,
           s_3-\tfrac1{20}c_5+\tfrac1{10}s_5\bigr),\\
 \mathcal W&=(\mathcal U A_{12})^\perp.
 \end{aligned}
\]
The ordered maximal minors of $C$ and $\mathcal U$ are strictly
positive; both are $r$-invariant and $A_{12}$-isotropic.
Thus $\mathcal U\subset\mathcal W$, and positive duality gives
positivity of $\mathcal W$. In the spinor frames
\[
 \lambda=\begin{pmatrix}
 c_5+\tfrac1{10}c_3+\tfrac1{20}s_3\\
 s_5-\tfrac1{10}s_3+\tfrac1{20}c_3
 \end{pmatrix},\qquad \widetilde\lambda=\lambda A_{12},
\]
the numerator is
\begin{equation}
\label{eq:bcfw:two-column-octic}
 102400u^4s_I^M(u)=A_8u^8+A_6u^6+A_4u^4+A_2u^2+A_0,
\end{equation}
where
\[
 \begin{aligned}
 A_8&=21447+7474\sqrt2-4124\sqrt3-7304\sqrt6,\\
 A_6&=-56720+22276\sqrt2-9672\sqrt3+27408\sqrt6,\\
 A_4&=-7694-21784\sqrt3,\\
 A_2&=-56720-22276\sqrt2-9672\sqrt3-27408\sqrt6,\\
 A_0&=21447-7474\sqrt2-4124\sqrt3+7304\sqrt6.
 \end{aligned}
\]
Thus the bound of degree four in $u^2$ is attained. The raw
numerator has degree eight in $u$, but the positive shift is recovered
by solving the quartic in $u^2$ and taking the positive square root.

\paragraph{Two-cut vertex boosts: reduction to quartics.}
For a single symmetry orbit of arrows forming disjoint three-vertex
boosts, the degree-eight numerator factors already on complexified
kinematics. 

\begin{proposition}
\label{prop:bcfw:no-primitive-octic}
Suppose one symmetry orbit of arrows forms a disjoint collection
of the boosts \eqref{eq:atoms:vertex-boost}. If a cyclic interval $I$
meets exactly two of these triples on both sides of the cut, its
Mandelstam factors before shifting.
In fixed symmetry-adapted spinor bases, the displayed factors below
have Laurent exponents in $[-2,2]$ under the shift. Consequently,
\[
 u^4s_I^M(u)=cF_1(u)F_2(u),\qquad \deg F_1,\deg F_2\leq4,
\]
with $c\ne0$, over the function field of complexified kinematics. 
The factors may coincide or degenerate further.
\end{proposition}
\begin{proof}
A symmetry preserving a three-label block fixes its middle label.
To put its four directed arrows in one orbit requires both the
reflection exchanging its neighbors and pure exchange separately.
Thus $G=D_p\times E$, where $E=\langle\mathcal D\rangle$,
$N=p\ell=2k$, and the triple centers form
one rotation orbit, spaced by $\ell\geq3$ labels. Two such triples
require $p\geq2$. Distinct center orbits, when
present, define separate shifts.

Let $a,b=a+j\ell$, $1\leq j<p$, be the centers at the two cuts,
and let $B=(a,b)$ be the open cyclic interval between them.
Each cut lies immediately before or after its center, so the only
possibilities for $I$ are $B$, $B\cup\{a,b\}$, $B\cup\{a\}$ and
$B\cup\{b\}$. The corresponding factorizations, up to nonzero
constant normalizations, are
\[
\renewcommand{\arraystretch}{1.2}
\begin{array}{c|c|c}
 I & j\ell & s_I\\\hline
 B\ \text{or}\ B\cup\{a,b\}&\text{even}&\alpha\beta\\
 B\ \text{or}\ B\cup\{a,b\}&\text{odd}&-\mu^2\\
 B\cup\{a\}\ \text{or}\ B\cup\{b\}&\text{even}&-(x+\sqrt{-1}y)(x-\sqrt{-1}y)\\
 B\cup\{a\}\ \text{or}\ B\cup\{b\}&\text{odd}&-\mu(\mu+2\varepsilon xy).
\end{array}
\]
To verify the table, use pure exchange to set
$\widetilde\lambda=\lambda A_N$, so $Q_I$ is symmetric.
The reflection $i\mapsto a+b-i$ preserves the first two interval
shapes. It is of vertex type when $j\ell$ is even, giving
$Q_I=\operatorname{diag}(\alpha,\beta)$, and of edge type otherwise,
giving $Q_I=\left(\begin{smallmatrix}0&\mu\\\mu&0\end{smallmatrix}\right)$,
as in \eqref{eq:bdy:determinant-types}.

For the remaining shapes, $|I|=j\ell$. In a spinor basis with
orthogonal rotation matrix $R$, the alternating signs give
$Q_{rJ}=(-1)^\ell RQ_JR^T$.
Write $I$ as $j$ consecutive translates of an $\ell$-label interval
$S$, whose $p$ translates partition $[N]$.
If $\ell$ is even, momentum conservation gives
$p\operatorname{tr}Q_S=0$; if $\ell$ is odd and $j$ is even,
the traces cancel in pairs. Thus $j\ell$ even forces
$Q_I=\left(\begin{smallmatrix}x&y\\y&-x\end{smallmatrix}\right)$,
proving the third row.

Finally, if $j\ell$ is odd, $B$ is edge-centered. In its reflection
eigenbasis, write $I=B\cup\{e\}$, $\lambda_e=(x,y)^T$ and
$\varepsilon=(-1)^{e-1}$. Then
\[
 Q_I=\begin{pmatrix}0&\mu\\\mu&0\end{pmatrix}
       +\varepsilon\begin{pmatrix}x^2&xy\\xy&y^2\end{pmatrix},
 \qquad s_I=-\mu(\mu+2\varepsilon xy).
\]
Thus the last factorization uses the centered subinterval $B$,
although $I$ itself need not have a setwise reflection.

The equivariant shift preserves the chosen symmetry matrices,
so these bases can be kept fixed as $u$ varies.
The shifted spinor entries have Laurent exponents in $[-1,1]$.
Every factor above is linear in momentum entries or quadratic in
spinor entries, hence has exponents in $[-2,2]$.
Multiplication by $u^2$ for each factor proves the degree bound.
\end{proof}
In the third row, \emph{real} vanishing requires $x=y=0$, hence $Q_I=0$;
it is not a nonzero-null one-particle channel.

\paragraph{A degree-eight numerator with quartic factors.}
For a concrete example, take $(N,k)=(12,6)$ and
\[
 G=\langle r,h,\mathcal D\rangle,\qquad
 r=S_6^6,\quad h=S_6^3J,\qquad
 M(u)=V_{123}(\log u)V_{789}(\log u).
\]
Here $r(i)=i+6$ and $h(i)=4-i$ modulo $12$; the arrows of the
two triples form one symmetry orbit. The following matrix gives an explicit
positive input:
\begingroup\small
\setcounter{MaxMatrixCols}{12}\setlength{\arraycolsep}{3pt}
\[
 B=\begin{pmatrix}
 217&0&-217&-183&0&108&0&-44\sqrt2&0&108&0&-183\\
 0&0&0&116&217&177&0&-48\sqrt2&0&177&217&116\\
 0&0&0&-12&0&64&217&207\sqrt2&217&64&0&-12
 \end{pmatrix}.
\]
\endgroup
Set $C=\operatorname{rowspan}\left(\begin{smallmatrix}B\\Br^T\end{smallmatrix}\right)$,
$B_U=\left(\begin{smallmatrix}-3&4&0\\1&0&4\end{smallmatrix}\right)B$,
$\mathcal U=\operatorname{rowspan}\left(\begin{smallmatrix}B_U\\B_Ur^T\end{smallmatrix}\right)$
and $\mathcal W=(\mathcal U A_{12})^\perp$.
The ordered maximal minors of $C,\mathcal U,\mathcal W$ are strictly
positive in suitable orientations; $C,\mathcal U$ are $r,h$-invariant
and $A_{12}$-isotropic, and $\mathcal U\subset\mathcal W$.
Choose $\lambda$ with rows $v=(4B_1+3B_2-B_3)/31$ and $vr^T$,
and put $\widetilde\lambda=\lambda A_{12}$.

For $I=\{3,4,5,6,7\}$ and $H=\{2,3,4,5,6,7\}$, put
$a(u)=5u-12/u$, $d(u)=a(u)^2+240$ and $b(u)=56a(u)-1137$.
Direct multiplication gives
\[
 Q_I^M(u)=\begin{pmatrix}d&b\\b&d\end{pmatrix},\qquad
 Q_H^M(u)=\begin{pmatrix}d&b\\b&-d\end{pmatrix}.
\]
Writing $f(u)=(5u^2+12)^2$ and
$g(u)=u(280u^2-1137u-672)$, the two numerators are
\[
 u^4s_I^M(u)=(f+g)(f-g),\qquad
 -u^4s_H^M(u)=f^2+g^2=(f+ig)(f-ig).
\]
Both raw numerators have degree eight, but the displayed factors
have degree at most four over $\mathbb C$. The centered numerator
factors further, since
$f+g=(5u^2-13u-12)(5u^2+69u-12)$.
For real $u$, $f(u)>0$, so the half-circle numerator has no real zero.
Thus it gives neither an irreducible octic boundary factor on
complexified kinematics nor a positive removal root.

\paragraph{Vertex and overlapping orbit operations.}
Separate exact experiments tested the palindrome, three-vertex boost and
an overlapping orbit, as summarized in
Table~\ref{tab:bcfw:local-atom-reconstructions}.
Here $h=S_kJ$ is the signed vertex reflection; in the mixed
$(8,4)$ model, $r=S_4^2$ rotates by two labels.
No rotation is imposed in the other rows.
In the eight-label overlapping and boost models, the removal polynomial
used for the prescribed recursion has a linear factor with coefficients
rational in the target coordinates. This was checked symbolically,
with the target coordinates left variable; the factor supplies the
accepted shift despite the higher degree of the full channel equation.
The positive-reconstruction tests used
two choices of positive external data per model, with exact arithmetic
over $K=\mathbb Q(\sqrt2,\sqrt3)$ and its real algebraic extensions.

\begin{table}[htbp]
\centering\small
\renewcommand{\arraystretch}{1.2}
\begin{tabular}{@{}>{\raggedright\arraybackslash}p{0.12\linewidth}>{\raggedright\arraybackslash}p{0.38\linewidth}>{\raggedright\arraybackslash}p{0.44\linewidth}@{}}
\hline
$(N,k)$ & Symmetry and operation & Reconstruction\\\hline
$(6,3)$ & $\langle h\mathcal D\rangle$; palindrome on $\{3,4,5\}$
 & Rational; two alternative one-term recursions\\
$(8,4)$ & $\langle h\mathcal D\rangle$; palindrome on $\{4,5,6\}$
 & Quadratic cell inverse in a fiber coordinate; one positive preimage in all $36$ inversions\\
$(8,4)$ & $\langle r\mathcal D,h\rangle$; full overlapping exponential
 & Rational, despite cubic and quartic channel equations\\
$(8,4)$ & $\langle h,\mathcal D\rangle$; boost on $\{4,5,6\}$
 & Rational boost inverse; three roots with $u>1$ and one positive reconstruction in the tests\\
$(12,6)$ & $\langle h,\mathcal D\rangle$; boost on $\{6,7,8\}$
 & Genuinely quartic accepted $u$; exact exclusions over entire regular fibers\\\hline
\end{tabular}
\caption{Local-atom reconstruction checks. The linear factors supplying
the eight-label overlapping and boost shifts were verified symbolically;
positive-root counts refer to the tested inputs.
The two palindrome orientations define alternative collections,
not two tiles to be combined.
Algebraic degrees at the tested inputs are over $K$ and do not count
positive sources or tiles.}
\label{tab:bcfw:local-atom-reconstructions}
\end{table}

In the eight-label palindrome model, the full domain, image and regular
fiber have dimensions $6,5,1$. Each declared five-dimensional symmetric
positroid intersection was parametrized completely. Its inverse solved
a quadratic in a fiber coordinate, not directly in the shift parameter $t$.

In the overlapping model, let $C_{\mathrm{seed}}$ be the source with
coloops $1,5$, loops $3,7$, and the rigid $K_{2,4}$ on the even labels,
with the inherited insertion signs. The entire strictly positive domain is
$C(s)=C_{\mathrm{seed}}M(s)$, $s>0$, and removal gives $t=s$.
This is a soft recursion, without an internal particle or gluing scale.
The quartic channel polynomials found irreducible over $K$ in the
experiment have no real roots.

In the twelve-label boost model, reflection and exchange are separately
imposed, with no rotation. Its full domain, image and regular fiber have
dimensions $6,4,2$. The prescribed factorization has
rank-two triples $\{4,5,6\}$ and $\{8,9,10\}$, each of connectivity one;
the augmented factors are reflection-related $\OG^{\geq}_{2,4}$ factors
attached to an eight-label core with reflection and exchange separately
imposed.
The factorization has domain and image dimension three, and adjoining
the boost gives a four-dimensional family with full image rank.
For each of the ten tested target/external-data pairs, we
successfully reconstructed a nonnegative source in the prescribed
factorization and an allowed boost mapping it to the target.
In eight tests the accepted shift had minimal polynomial of
degree four over $K$; the remaining two were consistency checks
constructed using a prescribed rational shift, which was recovered.

A further experiment checked eight of these target/external-data pairs,
including the two parametrization-generated targets, over their entire
regular source fibers. Among all roots with $u>1$ of the tested cyclic
channel equations, exactly one per target admitted a nonnegative lower
source, and that source was unique. The two triple ranks and connectivity
conditions were not assumed: every nonnegative reconstruction was shown
to satisfy them. The lower factors and their induced external data were
also checked, including the subsequent recursion of the core.


\begin{thebibliography}{99}

\fontsize{9.5}{11.4}\selectfont
\setlength{\itemsep}{1pt}

\bibitem{ArkaniHamedBaiLam2017}
N.~Arkani-Hamed, Y.~Bai and T.~Lam. \emph{Positive Geometries and Canonical Forms}. JHEP \textbf{11} (2017), 039. \href{https://arxiv.org/abs/1703.04541}{arXiv:1703.04541}.

\bibitem{AHBCGPT2016}
N.~Arkani-Hamed, J.~L.~Bourjaily, F.~Cachazo, A.~B.~Goncharov, A.~Postnikov and J.~Trnka. \emph{Grassmannian Geometry of Scattering Amplitudes}. Cambridge University Press, 2016. \href{https://arxiv.org/abs/1212.5605}{arXiv:1212.5605}.

\bibitem{ArkaniHamedThomasTrnka2018}
N.~Arkani-Hamed, H.~Thomas and J.~Trnka. \emph{Unwinding the Amplituhedron in Binary}. JHEP \textbf{01} (2018), 016. \href{https://arxiv.org/abs/1704.05069}{arXiv:1704.05069}.

\bibitem{ArkaniHamedTrnka2014}
N.~Arkani-Hamed and J.~Trnka. \emph{The Amplituhedron}. JHEP \textbf{10} (2014), 030. \href{https://arxiv.org/abs/1312.2007}{arXiv:1312.2007}.

\bibitem{BaoHe2019}
H.~Bao and X.~He. \emph{The $m=2$ amplituhedron}. 2019. \href{https://arxiv.org/abs/1909.06015}{arXiv:1909.06015}.

\bibitem{BrittoCachazoFeng2005}
R.~Britto, F.~Cachazo and B.~Feng. \emph{New recursion relations for tree amplitudes of gluons}. Nucl. Phys. B \textbf{715} (2005), 499--522. \href{https://arxiv.org/abs/hep-th/0412308}{arXiv:hep-th/0412308}.

\bibitem{BrittoCachazoFengWitten2005}
R.~Britto, F.~Cachazo, B.~Feng and E.~Witten. \emph{Direct proof of tree-level recursion relation in Yang--Mills theory}. Phys. Rev. Lett. \textbf{94} (2005), 181602. \href{https://arxiv.org/abs/hep-th/0501052}{arXiv:hep-th/0501052}.

\bibitem{DamgaardEtAl2019}
D.~Damgaard, L.~Ferro, T.~{\L}ukowski and M.~Parisi. \emph{The Momentum Amplituhedron}. JHEP \textbf{08} (2019), 042. \href{https://arxiv.org/abs/1905.04216}{arXiv:1905.04216}.

\bibitem{EvenZoharEtAl2025}
C.~Even-Zohar, T.~Lakrec, M.~Parisi, M.~Sherman-Bennett, R.~Tessler and L.~Williams. \emph{BCFW tilings and cluster adjacency for the amplituhedron}. Proc. Natl. Acad. Sci. USA \textbf{122}(12) (2025), e2408572122. \href{https://arxiv.org/abs/2504.01217}{arXiv:2504.01217}.

\bibitem{EvenZoharLakrecTessler2025}
C.~Even-Zohar, T.~Lakrec and R.~J.~Tessler. \emph{The Amplituhedron BCFW Triangulation}. Invent. Math. \textbf{239} (2025), 1009--1138. \href{https://arxiv.org/abs/2112.02703}{arXiv:2112.02703}.

\bibitem{PlabicTangles}
Chaim Even-Zohar, Matteo Parisi, Melissa Sherman-Bennett,
Ran Tessler, and Lauren Williams.
\newblock Plabic tangles and cluster promotion maps.
\newblock \emph{Journal of Algebra}, \textbf{695} (2026), 14--102.
\newblock \href{https://doi.org/10.1016/j.jalgebra.2026.01.027}
{doi:10.1016/j.jalgebra.2026.01.027}.
\newblock \href{https://arxiv.org/abs/2508.02891}
{arXiv:2508.02891}.

\bibitem{Fraser2020}
C.~Fraser. \emph{Cyclic symmetry loci in Grassmannians}. 2020. \href{https://arxiv.org/abs/2010.05972v1}{arXiv:2010.05972v1}.

\bibitem{Galashin2024}
P.~Galashin. \emph{Amplituhedra and origami, I: tree level}. 2024. \href{https://arxiv.org/abs/2410.09574}{arXiv:2410.09574}. Version 2, 2026.

\bibitem{GalashinKarpLam2022}
P.~Galashin, S.~N.~Karp and T.~Lam. \emph{The totally nonnegative Grassmannian is a ball}. Adv. Math. \textbf{397} (2022), 108123. \href{https://arxiv.org/abs/1707.02010}{arXiv:1707.02010}.

\bibitem{HeKuoZhang2021}
S.~He, C.-K.~Kuo and Y.-Q.~Zhang. \emph{The momentum amplituhedron of SYM and ABJM from twistor-string maps}. JHEP \textbf{02} (2022), 148. \href{https://arxiv.org/abs/2111.02576}{arXiv:2111.02576}.

\bibitem{HuangEtAl2021}
Y.-t.~Huang, R.~Kojima, C.~Wen and S.-Q.~Zhang. \emph{The orthogonal momentum amplituhedron and ABJM amplitudes}. JHEP \textbf{01} (2022), 141. \href{https://arxiv.org/abs/2111.03037}{arXiv:2111.03037}.

\bibitem{HuangWen2014}
Y.-t.~Huang and C.~Wen. \emph{ABJM amplitudes and the positive orthogonal Grassmannian}. JHEP \textbf{02} (2014), 104. \href{https://arxiv.org/abs/1309.3252}{arXiv:1309.3252}.

\bibitem{KarpVariation2017}
S.~N.~Karp. \emph{Sign variation, the Grassmannian, and total positivity}.
J. Combin. Theory Ser. A \textbf{145} (2017), 308--339.
\href{https://arxiv.org/abs/1503.05622v3}{arXiv:1503.05622v3}.

\bibitem{Karp2019}
S.~N.~Karp. \emph{Moment curves and cyclic symmetry for positive Grassmannians}. Bull. Lond. Math. Soc. \textbf{51} (2019), 900--916. \href{https://arxiv.org/abs/1805.06004}{arXiv:1805.06004}.

\bibitem{KarpWilliams2019}
S.~N.~Karp and L.~K.~Williams. \emph{The $m=1$ amplituhedron and cyclic hyperplane arrangements}. Int. Math. Res. Not. \textbf{2019}(5) (2019), 1401--1462. \href{https://arxiv.org/abs/1608.08288}{arXiv:1608.08288}.

\bibitem{Karpman2018}
R.~Karpman. \emph{Total positivity for the Lagrangian Grassmannian}. Adv. Appl. Math. \textbf{98} (2018), 25--76. \href{https://arxiv.org/abs/1510.04386}{arXiv:1510.04386}.

\bibitem{KarpmanSu2015}
R.~Karpman and Y.~Su. \emph{Combinatorics of symmetric plabic graphs}. J. Combin. \textbf{9}(2) (2018), 259--278. \href{https://arxiv.org/abs/1510.02122}{arXiv:1510.02122}.

\bibitem{KimLee2014}
J.~Kim and S.~Lee. \emph{Positroid Stratification of Orthogonal Grassmannian and ABJM Amplitudes}. JHEP \textbf{09} (2014), 085. \href{https://arxiv.org/abs/1402.1119}{arXiv:1402.1119}.

\bibitem{KroviTessler2026}
S.~T.~Krovi and R.~J.~Tessler. \emph{The Lagrangian Amplituhedron}. 2026. \href{https://arxiv.org/abs/2608.28561}{arXiv:2608.28561}.

\bibitem{Lam2015}
T.~Lam. \emph{Totally nonnegative Grassmannian and Grassmann polytopes}.
2015. \href{https://arxiv.org/abs/1506.00603}{arXiv:1506.00603}.

\bibitem{OrenPerlsteinKrovi}
M.~Oren-Perlstein and S.~T.~Krovi. To appear; proofs of the reflected boundary and tiling statements announced by Krovi and Tessler.

\bibitem{OrenPerlsteinTessler2025}
M.~Oren-Perlstein and R.~Tessler. \emph{The BCFW Tiling of the ABJM Amplituhedron}. 2025. \href{https://arxiv.org/abs/2501.07576}{arXiv:2501.07576}.

\bibitem{ParisiShermanBennettWilliams2023}
M.~Parisi, M.~Sherman-Bennett and L.~K.~Williams. \emph{The $m=2$ amplituhedron and the hypersimplex: signs, clusters, tilings, Eulerian numbers}. Commun. Amer. Math. Soc. \textbf{3} (2023), 329--399. \href{https://arxiv.org/abs/2104.08254}{arXiv:2104.08254}.

\bibitem{Postnikov2006}
A.~Postnikov. \emph{Total positivity, Grassmannians, and networks}. 2006. \href{https://arxiv.org/abs/math/0609764}{arXiv:math/0609764}.

\bibitem{Scott1879}
R.~F.~Scott. \emph{Note on a theorem of Prof.~Cayley's}.
Messenger of Mathematics \textbf{8} (1879), 155--157.

\bibitem{Shevchenko2025}
O.~Shevchenko. \emph{Rotationally symmetric plabic graphs and the Lagrangian Grassmannian}. 2025. \href{https://arxiv.org/abs/2511.23446}{arXiv:2511.23446}.

\end{thebibliography}
\end{document}